\documentclass[reqno,12pt,letterpaper]{amsart}
\usepackage{amsmath,amssymb,amsthm,graphicx,mathrsfs,url}
\usepackage[usenames,dvipsnames]{color}
\usepackage[colorlinks=true,linkcolor=Red,citecolor=Green]{hyperref}
\usepackage{amsxtra}
\usepackage{enumitem}
\usepackage{epstopdf}
\usepackage{stix}

\makeatletter
\def\Ddots{\mathinner{\mkern1mu\raise\p@
\vbox{\kern7\p@\hbox{.}}\mkern2mu
\raise4\p@\hbox{.}\mkern2mu\raise7\p@\hbox{.}\mkern1mu}}
\makeatother
\usepackage{graphicx}
\newcommand\smallO{
  \mathchoice
    {{\scriptstyle\mathcal{O}}}% \displaystyle
    {{\scriptstyle\mathcal{O}}}% \textstyle
    {{\scriptscriptstyle\mathcal{O}}}% \scriptstyle
    {\scalebox{.7}{$\scriptscriptstyle\mathcal{O}$}}%\scriptscriptstyle
  }

\def\?[#1]{\textbf{[#1]}\marginpar{\Large{\textbf{??}}}}

\def\smallsection#1{\smallskip\noindent\textbf{#1}.}
\let\epsilon=\varepsilon % sorry Knuth

\newcommand{\RR}{{\mathbb R}}
\newcommand{\ZZ}{{\mathbb Z}}

\newcommand{\CC}{{\mathbb C}}

\newcommand{\R}{\mathbb{R}}
\newcommand{\C}{\mathbb{C}}

\newtheorem{theo}{Theorem}
\newtheorem{prop}{Proposition}[section]	

\newtheorem{Assumption}{Assumption}

\newtheorem{lemm}[prop]{Lemma}

\newtheorem{rem}{Remark}
\newtheorem{ex}{Example}
\numberwithin{equation}{section}

\DeclareMathOperator{\Spec}{Spec}

\let\Im=\Imag

\let\Re=\Real

\DeclareMathOperator{\tr}{tr}

\def\indic{\operatorname{1\hskip-2.75pt\relax l}}

\title[Spectral gap for networks of oscillators]{Spectral analysis and phase transitions for long-range interactions in harmonic chains of oscillators }
\author{Simon Becker}
\address{Department of Mathematics, R\"amistrasse 101, 8092 Z\"urich, Switzerland}
\email{simon.becker@math.ethz.ch}

\author{Angeliki Menegaki}
\address{Institut des hautes études Scientifiques, 35 Rte de Chartres, 91440, Bures-sur-Yvette, France }
\email{menegaki@ihes.fr}

\author{Jiming Yu}
\address{University of Chicago, 5801 S Ellis Ave, Chicago, IL 60637, USA}
\email{tommenix@uchicago.edu}

\begin{document}
\maketitle
%%%%%%%%%%%%%%%%%%%%%%%%%%%%%%%%%%%%%%%%%%%%%%%%%%%%%%%%%%%%%%%%%%%%%%%%%%%%%%%%
%                                 INTRODUCTION                                 %
%%%%%%%%%%%%%%%%%%%%%%%%%%%%%%%%%%%%%%%%%%%%%%%%%%%%%%%%%%%%%%%%%%%%%%%%%%%%%%%%
%\addtocounter{section}{1}
\begin{abstract}
We consider chains of $N$ harmonic oscillators in two dimensions coupled to two Langevin heat reservoirs at different temperatures - a classical model for heat conduction introduced by Lebowitz, Lieb, and Rieder \cite{RLL67}.
We extend our previous results \cite{BM20} significantly by providing a full spectral description of the full Fokker-Planck operator allowing also for the presence of a constant external magnetic field for charged oscillators. We then study oscillator chains with additional next-to-nearest-neighbor interactions and find that the spectral gap undergoes a phase transition if the next-to-nearest-neighbour interactions are sufficiently strong and may even cease to exist for oscillator chains of finite length.
\end{abstract}
\tableofcontents

\section{Introduction}
The chain of oscillators is a model that has been first introduced for the rigorous derivation of Fourier's law or to obtain a mathematically rigorous proof of its breakdown. This history of developments has been described in quite a few overview articles on the subject \cite{BLR00,Lep16, Dhar08,BF19}. This article focuses on the case of harmonic interactions which has been studied first in \cite{RLL67}, where, by solving several Lyapunov equations, the unique invariant state was constructed and the breakdown of Fourier's law was derived. Aside from the harmonic setting, there exist many results for anharmonic potentials both on the existence of steady states \cite{EPR99a, EPR99b, EH00} as well as on the convergence to equilibrium \cite{RBT02, Car07}.

In most works, the $N$-dependence of the convergence to equilibrium has not been studied and we are only aware of an approach based on hypocoercivity, c.f. \cite[Section 9.2]{Villani09}. The approaches discussed by Villani however only led to far-from optimal estimates on the convergence to equilibrium with respect to the number of oscillators. More recently, a weak perturbation of harmonic oscillator chains was studied in \cite{Me20} where better estimates were obtained on the convergence rate by a hypocoercivity-inspired machinery.  

In our previous work \cite{BM20}, we started the study of sharp spectral gaps in terms of $N$, that is providing with the optimal exponential factor in the convergence rate to equilibrium for the chain of oscillators.

In this article, we study the full $L^2$-spectrum of the Fokker-Planck operator of the chain of harmonic oscillators. We significantly extend our previous work \cite{BM20}, where we focused on the $L^2$-gap in various regimes and for different configurations of the harmonic oscillator networks connected to heat baths. 
First, we provide a precise description of the full spectrum of the Fokker-Planck operator associated with oscillator chains connected to heat baths. In particular, we allow for charged oscillators inside a constant magnetic field. The generalization to such charged oscillators inside a constant magnetic field is motivated by several recent works in which the conductivity of harmonic chains connected to heat baths inside magnetic fields, breaking momentum conservation, has been studied \cite{BCBD21a,BCBD21b} motivated by studies \cite{TSS17,SS18} of heat transport of weakly charged atom configurations in strong external magnetic fields. In such configurations the Lorentz force dominates over the lattice oscillations. In \cite{TS18} the Nernst effect, which is usually non-existent in standard metals, but common in semiconductors, in a flexible (unpinned in the bulk) and nonlinear chain is studied, where the average positions of particles deviate in the perpendicular direction to the heat flow. 

Our mathematical results on the full spectrum are even new in the setting without magnetic fields.  Only in \cite{EckmannHairer03}, for highly degenerate H\"{o}rmander type of Fokker Planck operators the full spectrum has been studied in a general setting and applied to anharmonic oscillators chains, showing that it lies in a cusp. Here we quantify this result in terms of the number of oscillators $N$.

Our second main contribution is the study of interactions between oscillators that are not limited to nearest neighbour interactions. Perhaps surprisingly, this leads to phase transitions in the behaviour of the spectral gap. Here we use the term phase transition to describe that the spectral gap abruptly changes and even sometimes ceases to exist as a function of $N$, depending on the regime of the next-to-nearest-neighbour interaction's strength.
%In fact, we observe that as a rule of thumb long-range interactions seem to reduce conductivity. 

This opens many interesting questions, such as how the spectral gap behaves when one considers really long range interactions and not just next to nearest neighbour ones, as well as how does this affect the conductivity. In particular, in contrast to the hunt for anharmonic potentials one may also consider if long-range interactions for the harmonic chain could prohibit the linear growth in $N$ of the conductivity leading to the breakdown of Fourier's law. Recent results on the hydrodynamic limit for such chains with exponentially decaying interactions perturbed by a random exchange of momentum are in \cite{KomOlla17} where the energy current follows macroscopically a diffusion equation and also in \cite{Suda22} with polynomially decaying interactions where one sees superballistic behaviour.

\subsection{Chain of oscillators with constant magnetic field}
We consider labelled oscillators on the sites of a linear chain $[N]:=\{1,\dots, N\}$ confined by a quadratic pinning potential and interacting with their nearest neighbours by a quadratic interaction potential. We assume that each oscillator has mass $m>0$ and a charge $e\in \RR$ which we shall just normalize to one\footnote{Everything is then fully determined by the magnetic field strength, only.}. We denote by $\textbf{m}_{[N]}:= m I_{\mathbb C^n}$ the mass matrix, where we assume the masses of the oscillators to coincide. According to classical mechanics, the dynamics of each oscillator in phase space is described fully by position variables $q_i \in \R^2$ and momentum variables $p_i \in \R^2$. 
In addition, we also allow for the presence of a magnetic field perpendicular to the plane of the network
\[ B = (\partial_{x_1} A_2 - \partial_{x_2} A_1) \ dx_1 \wedge dx_2
\] 
where $A: \R^2 \to \R^2$ is the electromagnetic vector potential. For our analysis, we will focus on constant magnetic fields $B_0$ which are obtained by choosing \[ A(x_1,x_2) := \frac{B_0}{2} (-x_2, x_1).\]

The energy of the system is then given by the Hamilton function
\begin{equation}
\begin{split}
\label{eq:potentials}
H(\textbf{q},\textbf{p}) &= \frac{\langle \textbf{p} - eA(\textbf{q}),\textbf{m}_{[N]}^{-1} (\textbf{p} -  eA(\textbf{q}))  \rangle}{2}+ V_{\eta,\xi}(\textbf{q}) \text{ where } \\
V_{{\bf \eta,\xi}}(\textbf{q}) &=  \frac{1}{2}\sum_{i=1}^N \eta_{i} \vert q_i \vert^2 + \frac{1}{2}\sum_{i =1}^{N-1}\xi_{i, i+1}  \vert q_i-q_{i+1} \vert^2
\end{split}
\end{equation}   
where $e$ is the charge of the particle. We shall assume in this article that all $\eta_i>0$ and $\xi_{ii+1}>0$ coincide, respectively. For studies of disordered or localized impurities, see \cite{BM20}.
\iffalse
Expanding the inner product, one of the terms that we have due to the constant magnetic field is quadratic in the spatial variables, it equals to
\[\frac{e^2}{2} \langle A(\textbf{q}), \textbf{m}_{[N]}^{-1} A(\textbf{q}) \rangle = \frac{(eB_0)^2}{2} \langle \textbf{q}, \textbf{m}_{[N]}^{-1} \textbf{q} \rangle.
\]
We can therefore absorb this term in the quadratic potential.  
\fi

To model a heat flow through the system, we couple the linear chain to two heat reservoirs at different temperatures at the terminal ends of the chain. 
Here the reservoirs at the terminal ends are assumed to contain Gaussian noise such that the dynamics becomes an Ornstein--Uhlenbeck process. This means that the particles at the boundaries are subject to reservoirs at different temperatures $T_i = \beta_i^{-1}$, $i \in \{1, N\}$ as well as to friction $\gamma_i>0$. 

The time evolution for particles $i \in [N]$ is then described by the following system of SDEs:
\begin{equation}
\begin{split}
\label{eq:SDE}
dq_i(t)&= \nabla_{p_i} H \ dt \text{ and }\\
dp_i(t)&= \left(-\nabla_{q_i} H- \gamma_i p_i \delta_{i \in F} \right) \ dt+ \delta_{ i \in F} \sqrt{2m \gamma_i \beta_i^{-1}}\ dW_i 
\end{split}
\end{equation}
where  $W_i$ with $i \in F$ are iid Wiener processes, $\gamma_i> 0$ a friction parameter, and $F \subset  \{1,N\}  $, with $F \neq \emptyset$, the set of the particles subject to friction. 

The generator of the associated strongly continuous semigroup is the Fokker-Planck operator
\begin{equation}
\label{eq:L}
 \mathcal{L} f(z)  = - \langle z,  M_{[N]}   \nabla_z f(z) \rangle + \langle \nabla_p ,  \Gamma \textbf{m}_{[N]}  \vartheta \nabla_p  f(z) \rangle
 \end{equation}
where $M_{[N]}  \in \mathbb C^{4N \times 4N}$ is the matrix containing the first-order coefficients of the above generator and $z = (p,q)$. In particular, with $\Gamma \in \mathbb R^{N \times N}$ the matrix containing the friction parameters, the parameter matrix takes the form
\begin{equation}
\begin{split}
\label{eq:MN}
M_{[N]}  &:= \left(\begin{matrix} \Gamma \otimes I_{\CC^2}  + \textbf{m}_{[N]}^{-1} JB_0 & -\textbf{m}_{[N]}^{-1} \otimes I_{\C^2} \\ B_{[N]} \otimes I_{\C^2}    &  -\textbf{m}_{[N]}^{-1} JB_0  \end{matrix}\right) \text{ and }\Gamma = \text{diag}(\gamma_{1} \delta_{1 \in F},0, \dots,0, \gamma_{N}  \delta_{N \in F}).
\end{split}
\end{equation}
The matrix 
\begin{equation}
\label{eq:omega}
    J := \operatorname{diag}(\Omega, \dots, \Omega)\text{ where }\Omega = \left( \begin{matrix} 0& -1 \\ 1 &0 
\end{matrix} \right),
\end{equation} 
is trace-free. 
The temperature matrix $\vartheta$ is of the form $$\vartheta = \text{diag}(\beta_{1}^{-1}\delta_{1 \in F},\dots,\beta_{N}^{-1}\delta_{N\in F}).$$

Finally, the forces are described by the matrix $B_{[N]}$. To explicitly state its form, we define for $i,j \in [N]$ self-adjoint operators $\langle u, L_{i,j}u \rangle_{\ell^2([N]; \CC)} := \vert u(i)-u(j) \vert^2$ that decompose the negative weighted Neumann Laplacian on $\ell^2([N]; \CC)$ as $$-\Delta_{[N]} = \sum_{i=1}^{N-1} \xi_{i,i+1}L_{i, i+1} \text{ with } \xi_{i,i+1} \text{ as in }\eqref{eq:potentials}.$$
Thus, the discrete Neumann Laplacian describes the nearest neighbor interaction. We then write the matrix $B_{[N]} \in \RR^{N \times N}$ appearing in $M_{[N]}$ in terms of a Schr\"odinger operator
\begin{equation}
\label{eq:Schroe}
B_{[N]} =-\Delta_{[N]} + \sum_{i=1}^N \eta_i \delta_i + \frac{B_0^2}{2 m} I
\end{equation}
where $(\delta_i(u))(j)=u(i)\delta_{ij}.$
In other words the operator $B_{[N]}$ is just a Jacobi (tridiagonal) matrix 
\[(B_{[N]} f)_n = -\xi_{n,n+1} f_{n+1} - \xi_{n-1, n} f_{n-1} +\Bigg(\eta_n+\frac{B_0^2}{2 m}+\xi_{n,n+1}+ \xi_{n-1,n}\Bigg)f_n \text{ for } n\in [N]\] 
with the convention that $\xi_{0,1} = \xi_{N,N+1} =0.$

\subsection{One-dimensional chain of oscillator with next-to-nearest-neighbor interactions} In our second part, we consider the one-dimensional chain as defined above, for simplicity without a constant magnetic field, but allow for next to nearest neighbor interactions of strength $\omega>0.$ That means that we consider quadratic interactions given by the potential 
\begin{equation}
\label{eq:NNN}
V(q) = \frac{1}{2}\sum_{i=1}^{N-1} (q_i-q_{i+1})^2 + \frac{1}{2}\omega \sum_{i=1}^{N-2}(q_i-q_{i+2})^2.
\end{equation}

The generator describing the evolution of the dynamics is given by $\mathcal{L}$ as in \eqref{eq:L} with $\Gamma$ and $\vartheta$ as in the previous subsection.  The matrix
 $M_{[N]}  \in \mathbb C^{2N \times 2N}$ is again of the form, but of half the matrix size since the oscillators are now just assumed to have one degree of freedom. 
\begin{equation} \label{eq: matrix M}
\begin{split}
M_{[N]}  &:= \left(\begin{matrix} \Gamma & -\textbf{m}_{[N]}^{-1} \\ B_{[N]}  & 0 \end{matrix}\right) 
\end{split}
\end{equation}
where $B_{[N]} $ can be expressed, up to a low-rank perturbation, now as a quadratic function of a Schr\"odinger operator  $-\Delta_{[N]} + \sum_{i \in [N]} \chi_i \delta_i$ for suitable $\chi_i$. This will be specified in later subsections, as the specific form of this term depends on the choice of boundary conditions.

\subsection{Main results} 

We consider a chain of $N$ oscillators with two-dimensional phase space variables in a constant magnetic field of strength $B_0 \in \mathbb R$ with a potential as in \eqref{eq:potentials}.  The spectrum of the Fokker-Planck operator generating the dynamics satisfies then in terms of eigenvalues $\mu_i^{\pm}$ with $\mu_i^+ = -\mu_i^-$ and associated eigenvectors $V_i^{\pm}$, for $i \in [N]$ of the matrix $Q_{[N]}  := \left(\begin{matrix} 0 & -\textbf{m}_{[N]}^{-1} \\ B_{[N]}  & 0 \end{matrix}\right) \in \CC^{2N \times 2N}$ (with zero friction!).

Before we state our main results, let us specify the two boundary conditions we include in this work. We consider the model given by \eqref{eq:SDE} with either \begin{itemize}
    \item[(i)]The discrete Dirichlet Laplacian is defined for $\xi_{i,i+1} >0$ defined by $$ -\Delta_{[N]}^D  = \sum_{i=0}^{N} \xi_{i,i+1}L_{i,i+1}$$ where $L_{i,i+1}$ are defined through the quadratic form $\langle u , L_{i,i+1}u\rangle = |u(i) - u(i+1)|^2$ with the convention that $u(0) = u(N+1)=0$. 
    \item[(ii)] Or with Neumann boundary conditions (free boundaries) by setting
    $$ -\Delta_{[N]}^N  = \sum_{i=1}^{N-1}\xi_{i,i+1} L_{i,i+1}.$$
\end{itemize}
We denote the invariant state $\mu$ of our dynamics with density $d\mu(p,q)=f_{\infty}(p,q)dp dq$, also as the non-equilibrium steady state due to the presence of non-zero fluxes. It is a probability measure on the phase space so that when $(S_t)_{t\geq 0}$ is the associated Markov semigroup:  $$\int_{\mathbb{R}^{2N}} (S_t h)(p,q) f_{\infty}(p,q) dp \ dq= \int_{\mathbb{R}^{2N}} h(p,q)f_{\infty}(p,q) dp \ dq, \text{ for all } t\geq 0,\ h \in C_b(\mathbb{R}^{2N}).$$
Finally the spectral gap of the operator $\mathcal{L}$ in \eqref{eq:L} is given by the formula 
\begin{align} \label{def:spectral gap}
g(N):=\inf\{\Re(\lambda): \Re(\lambda)>0 \text{ and } (\lambda - \mathcal{L}) \text{ is invertible with bounded inverse} \}.
\end{align}
The definition of the spectral gap has a dynamical interpretation in terms of the optimal rate of convergence to the invariant state: 

The spectral gap $g(N)$ is the smallest constant for which there is a constant $C(N) >1$ such that 
\[ \Vert S_t h \Vert_{L^2(\mathbb R^{N}; f_{\infty}(x) dx )} \le C(N) e^{-tg(N)} \Vert h \Vert_{L^2(\mathbb R^{N}; f_{\infty}(x) dx )} \text{ for }h \perp 1\text{ and all } t>0.\]

\begin{theo}[Full spectrum \& $B_0$] \label{theo:mag}
Let the friction $\gamma>0$ be sufficiently small, with $F=\{1\},$ consider Dirichlet or Neumann boundary conditions, $B_0 \in \mathbb R$ arbitrary, with at least one of $B_0$ or $\eta$ non-zero and $N$ large enough. Then, there exist $4N$ numbers $\lambda_i^{\pm}(j)$, with two of them in annuli 
\begin{equation}
\label{eq:eigenvalues}
\lambda_i^{\pm} \in B_{\CC}(\mu_i^{\pm},c\vert V_i^{\pm}(1)\vert^2) \backslash B_{\CC}(\mu_i^{\pm}, \varepsilon(\gamma)\vert V_i^{\pm}(1)\vert^2 )
\end{equation} and real parts satisfying 
\begin{equation}
\label{eq:real_part}
\varepsilon(\gamma,B_0)\vert V_i^{\pm}(1)\vert^2 \lesssim \Re(\lambda_i^{\pm}(j)) \lesssim \vert V_i^{\pm}(1)\vert^2
\end{equation}for some suitable $c>0$ and $\varepsilon(\gamma)>0$ small enough but independent of $N$.  The spectrum of the Fokker-Planck operator $\mathcal L$ on $L^2(\mathbb R^{2N}; f_\infty(x) dx)$, where $f_\infty$ is the invariant state, is given by
\[ \Spec_{L^2(\mathbb R^{2N}, f_\infty(x) dx)}(A) = \sum_{j \in \ZZ_2} \sum_{\pm}\sum_{i=1}^{2N} \lambda_i^{\pm}(j) \cdot \mathbb N_0\]
where the addition of sets is defined element-wise.

For arbitrary $\gamma>0$ the spectral gap always decays like $1/N^3.$
\end{theo}
Notice here the effect of the magnetic field which due to the Lorentz force confines the particles in space and therefore has a similar effect as the pinning potential with parameter $\eta.$

\begin{rem} 
We stress that eigenvectors $V_{i}^{\pm}$ and eigenvalues $\mu_i^{\pm}(1)$ appearing in \eqref{eq:eigenvalues} and \eqref{eq:real_part} have fully explicit closed-form expressions.
\end{rem}

Next we consider the one-dimensional harmonic chain of oscillators, with zero magnetic field $B_0=0$ but with additional next-to-nearest-neighbour interactions of strength $\omega$ instead, i.e. the second term, the nearest neighbour interaction, in \eqref{eq:potentials} is replaced by \eqref{eq:NNN}. We then find that there is a critical value of $\omega_{\text{crit}}=\frac{1}{4}$ such that the behaviour of the spectral gap of the generator undergoes a phase transition at this point.
We recall that the operator $\mathcal{L}$ is said to be \emph{hypoelliptic} when every distribution $u \in \mathscr D'(\mathbb{R}^{2N})$ so that $\mathcal{L} u \in C^\infty(\mathbb{R}^{2N})$, is itself smooth.

For Dirichlet boundary conditions, the matrix $M_{[N]}$ is as in \eqref{eq: matrix M} with $B_{[N]} = (1+4\omega)\Delta_{[N]}^D - \omega(\Delta_{[N]}^D)^2 + P$ where $\Delta_{[N]}^D$ denotes the discrete Dirichlet Laplacian and $P= \operatorname{diag}(-2\omega,0,\dots,0,-2\omega)$. Then we decompose $M_{[N]}$ as $ M_{[N]}=A + \begin{pmatrix}
\Gamma & 0 \\ P & 0
\end{pmatrix}$ or as  $ M_{[N]}=A_\omega + \begin{pmatrix}
\Gamma & 0 \\ 0 & 0
\end{pmatrix}$ with $A$ and $A_{\omega}$ defined accordingly.

\begin{theo}[Phase transition \& NNN] \label{theo:longer_range} Let $\mathcal{L}$ as in \eqref{eq:L} be the generator of the dynamics 
with Dirichlet boundary conditions and friction $\gamma>0$ with $F=\{1\}$. Let $(\mu_j^{\pm},v_j^{\pm})$ be the eigensystem of $A$ and $(\xi_j^{\pm},w_j^{\pm})$ be the eigensystem of $A_{\omega}$ with $\mu_j^+ = - \mu_j^-$ and $\xi_j^+= -\xi_j^-$ due to basic symmetries. 
\begin{itemize}
    \item[(i)] \emph{Hypoellipticity of $\mathcal{L}$:}  When $|\omega|>\frac{1}{4}$, there is an explicit dense set of such $\omega$'s so that for a specific number of particles $N=N(\omega)$ the operator $\mathcal{L}$ is not hypoelliptic. When $|\omega|\leq \frac{1}{4}$,  $\mathcal{L}$ is always hypoelliptic for any $N \in \mathbb N.$
    \item[(ii)] \emph{Spectral gap of $\mathcal L$:} For $\eta>0$, the spectral gap of the Fokker-Planck operator is then given by $g(N):=\inf_{\lambda \in \Spec(M_{[N]})}\Re(\lambda)>0$ when $|\omega|\leq \frac{1}{4}$ for $P=\operatorname{diag}(-2\omega, 0, \dots, 0).$
    \begin{itemize}
        \item[(iia)] When $\omega= \frac{1}{4}$, $ g(N) \lesssim N^{-4}$.
        \item[(iib)] When $\omega$ is sufficiently small, $N^{-3}\lesssim g(N) \lesssim N^{-3}$.    
        Restricting further to small frictions $\gamma$, -as in the case $\omega=0$- we find 
        $$ \lambda_j^{\pm} \in B_{\mathbb{C}}(\xi_j^{\pm}, |w_j^{\pm}(1)|^2) /B_{\mathbb{C}}(\xi_j^{\pm}, \varepsilon(\gamma)|w_j^{\pm}(1)|^2) \text{ for } j\in [N]$$
        and $\varepsilon(\gamma)>0$ small enough. 
        The $L^2(\mathbb R^{2N}; f_\infty(x) \ dx)$-spectrum of $\mathcal{L}$ then, when $f_\infty$ is the invariant state, is given by linear combinations of the eigenvalues of $M_{[N]}$, i.e. 
    \end{itemize}
\end{itemize}
\end{theo}

\begin{rem}
We restrict us to $P = (-2\omega,0,...,0)$ instead of  $P= (-2\omega,0,...,0,-2\omega)$ and $F=\{1\}$ instead of $F=\{1,N\}$ for simplicity to avoid higher rank perturbations which can be treated with the same techniques but lead to more intricate estimates. This does not affect the behaviour of the spectral gap, as indicated by numerical experiments in this work. A treatment of rank $2$ perturbations, at least for the spectral gap, can also be found in \cite{BM20} and since this analysis carries over to this work, we shall not consider this additional layer of complication, here. 
%{\color{blue} Additional comments on Neumann maybe? -  Mention also that the quantities in the previous theorem are quite explicit $V_j(1)$ and all that... Explain the concept of spectral gap in the above theorem!}
\end{rem}
\begin{rem} 
We study the scaling of the spectral gap in case of the next-to-nearest-neighbour interaction only for Dirichlet boundary conditions, since the perturbation $P$ under Neumann boundary conditions is of higher rank. In this case however we still have hypoellipticity as long as $|\omega|\leq 1/4$ and lack of hypoellipticity otherwise but for a different set of $\omega$'s. These are discussed in sec. \ref{subsec:crit for hypoell}.
\end{rem}

\smallsection{Notation}
We write $f(z) = \mathcal O(g(z)) $ to indicate that there is $C>0$ such that $\left\lvert f(z) \right\rvert \le C \left\lvert g(z) \right\rvert$ and $f(z)= \smallO(g(z))$ for $z \rightarrow z_0$ if there is for any $\varepsilon>0$ a neighbourhood $U_{\varepsilon}$ of $z_0$ such that $\left\lvert f(z)\right\rvert \le \varepsilon \left\lvert g(z) \right\rvert.$ Instead of writing $f(z) = \mathcal O(g(z))$, we sometimes also write $f(z) \lesssim g(z).$ In addition, we write $f \sim g$ if $ g(z) \lesssim f(z) \lesssim  g(z).$
The eigenvalues of a self-adjoint matrix $A$ shall be denoted by $\lambda_1(A)\le ... \le \lambda_N(A)$. We also employ the Kronecker delta where $\delta_{n \in I}=1$ if $n \in I$ and zero otherwise. The inner product of two vectors $x,y \in \RR^m$ is denoted by $\langle x,y\rangle.$ The ball of radius $r$ centered at $x$ is denoted by $B(x,r)$ and occasionally by $B_{\CC}(x,r)$ to emphasize that it is a ball in $\CC.$ We use the notation $[N]:= \{1,\dots,N\}$ and 
finally with $C_b(X)$ we denote the space of bounded continuous functions acting on $X$.

\section{The spectrum of the harmonic chain of oscillators}

In this section, we study the spectrum of the chain of oscillators and allow for the presence of an external constant magnetic field which couples position and momenta of particles in different spatial directions. Before we start with the analysis of the concrete model, we want to recall basic spectral theoretic properties of the operator. 
\subsection{Hypoellipticity, invariant measure, and the spectrum of the OU operator} \label{subsect: hypoell_OU}
As we saw in the introduction, the Fokker-Planck operator of the chain of oscillator is of the special form 
\[ \mathcal L = \tr(QD^2) + \langle Bx,D \rangle\]
with $x \in \mathbb R^N$, $Q=Q^*$ and $B$ both real non-zero square matrices. An operator of that form is called an \emph{Ornstein-Uhlenbeck operator}. 
In the study of OU operators, the self-adjoint matrix 
\[ Q_t :=\int_0^t e^{sB} Q e^{sB^*} \ ds \]
plays a special role. Indeed, it is well-known that the condition $\operatorname{det}(Q_t)>0$, for one $t>0$ and hence for all $t>0$, is equivalent to the hypoellipticity of the operator $\partial_t-\mathcal L.$

Moreover it is a direct consequence of H\"{o}rmander's regularity theorem \cite{Ho69} that an operator of the form $\mathcal L = \sum_{j=1}^{K} X_j^2 + X_0$, where $X_0,X_j$ are real smooth vector fields, is hypoelliptic once the Lie algebra generated by $X_j$'s for $j=1,\cdots,K$ has rank $N$.

Under the assumption of hypoellipticity, i.e. $\operatorname{det}(Q_t)>0$, the condition $\Spec(B) \subset \{\lambda \in \CC; \Re(\lambda)<0\}$ implies the existence of a unique invariant measure 
\[ d\mu(x) = \frac{1}{(4\pi)^{N/2}\sqrt{\operatorname{det}(Q_{\infty})}} e^{-\langle Q_{\infty}^{-1} x,x \rangle/4} \ dx\]
to the OU semigroup.
The operator $\mathcal L$ then generates a contraction semigroup on $L^p(\RR^N, d\mu(x)).$ In this setting, perhaps rather surprisingly, the spectrum of the OU generator $\mathcal L$ is fully determined by the spectrum of $B$ for $p \in (1,\infty]$.  This result was, to our knowledge, first established in \cite{MPP02}. 

\begin{theo}{\cite[Theorem $3.1$]{MPP02}}\label{theo:metafune}
Let $\mathcal L$ be the generator of an Ornstein-Uhlenbeck process.
In addition, we assume that $\partial_t-\mathcal L$ is hypoelliptic and its associated semigroup possesses a unique invariant measure $d\mu$.
Let $\lambda_1, \cdots, \lambda_m$ be the (distinct) eigenvalues of $B$ and $\mu$ be the unique invariant measure of the semigroup. Then the spectrum of the generator is given by $$\Spec_{L^2(\RR^N, d\mu(x))}(\mathcal L)= \Big\{  \sum_{i=1}^m \lambda_i n_i, n \in \mathbb{N}_0^m\Big\}.$$ 
\end{theo}
This result in particular implies that the spectral gap of the operator $\mathcal L$, defined by the formula \eqref{def:spectral gap}, is in fact given by $g:= \inf\{ \Re(\lambda): \lambda \in \Spec{(B)} \}. $ 

More results on the optimal exponential rate of such possibly degenerate Ornstein-Uhlenbeck operators in $L^2(\mu)$ or in relative entropy distance can be found in \cite{AE14, Mon15}.  
\bigskip

Therefore whenever the operator of the chain of oscillators satisfies conditions (1) and (2) we immediately have  for the spectral gap of the Fokker-Planck operator that  
\[g(N):=\inf\{ \Re(\lambda): \lambda \in \Spec(M_{[N]}) \}  .\] 

\subsection{The linear chain of oscillators}
Returning to our model, we observe that by Theorem \ref{theo:metafune}, it suffices to understand the spectrum of the matrix \eqref{eq:MN}.
As a first step we disentangle the $(q_x, p_x)$ and $(q_y,p_y)$ coordinates of each oscillator that are coupled by the magnetic field. For that we diagonalize the matrix $J$, see \eqref{eq:omega}, by choosing a new basis
\[ 
f_{2j-1} = \frac{ e_{2j}-i e_{2j-1} }{\sqrt{2}},\quad f_{2j} = \frac{ e_{2j}+i e_{2j-1} }{\sqrt{2}}
\]
in terms of standard basis vectors $e_i.$
By performing this change of variables, the matrix $M_{[N]}$ can be decomposed into the direct sum of two matrices corresponding to even and odd indices: 
\begin{equation}
\left(\begin{matrix} \Gamma  + i \textbf{m}_{[N]}^{-1} B_0 & -\textbf{m}_{[N]}^{-1}  \\ B_{[N]}    &  -i \textbf{m}_{[N]}^{-1} B_0  \end{matrix}\right)\ \text{and } \left(\begin{matrix} \Gamma  - i \textbf{m}_{[N]}^{-1} B_0 & -\textbf{m}_{[N]}^{-1} \\ B_{[N]}     &  i \textbf{m}_{[N]}^{-1} B_0  \end{matrix}\right).
\end{equation}
We see right away that the matrices are the same up to a change of sign in the constant magnetic field and we shall therefore focus on the first one.

We then introduce the matrix $A  := \left(\begin{matrix}  i \textbf{m}_{[N]}^{-1} B_0   & -\textbf{m}_{[N]}^{-1}  \\ B_{[N]}     & -i \textbf{m}_{[N]}^{-1} B_0 \end{matrix}\right)$ and $A_{\Gamma} := i\hat{\Gamma}+iA$, where $\displaystyle\Hat{\Gamma} = \left[ \begin{matrix}\Gamma & 0 \\ 0& 0\end{matrix}\right]\in \CC^{2N\times 2N}$, with $\Gamma = \text{diag}(\gamma_{1} \delta_{1 \in F},0, \dots,0, \gamma_{N}  \delta_{N \in F}).$ Hence, 
\begin{equation}
\label{eq:spectrum}
    \Spec(M_{[N]}) = -i (\Spec(A_{\Gamma}) \cup \overline{\Spec(A_{\Gamma})}).
\end{equation}

As our first Proposition shows, it is then rather straightforward to obtain a full spectral decomposition of $A$ in terms of the spectral decomposition of $B_{[N]}$.

\begin{prop}
\label{spec:A0}
To any eigenpair $(\lambda_j,v_j)$ of $B_{[N]}$, there exist two eigenvalues of $ A$  $$\mu_j^{\pm} = \pm \frac{i\sqrt{B_0^2+\lambda_j m}}{m}$$ with eigenvectors
 $$V^{\pm}_j = \frac{(v_j,u_j)^T}{\vert(v_j,u_j)^T\vert},\text{
 where } u_j = i\left(B_0\mp\sqrt{B_0^2+\lambda_j m}\right)v_j.$$
\end{prop}
\begin{proof}
We note that all blocks in $A$ aside from $B_{[N]}$ are multiples of the identity matrix. Hence, since $B_{[N]}$ is self-adjoint, we can write $B_{[N]} = S^{-1}\Lambda S$, where $\Lambda = \left(\begin{matrix}  \lambda_1   & &  \\ &\ddots&\\ & & \lambda_N \end{matrix}\right) $ contains the eigenvalues and $S$ the eigenbasis.

 So $$\left( \begin{matrix} S^{-1}& 0 \\ 0 &S^{-1} 
\end{matrix} \right) A \left( \begin{matrix} S& 0 \\ 0 &S 
\end{matrix} \right) = \left(\begin{matrix}  i \textbf{m}_{[N]}^{-1} B_0 & -\textbf{m}_{[N]}^{-1}  \\\Lambda     &  -i \textbf{m}_{[N]}^{-1} B_0  \end{matrix}\right)$$

We observe that this matrix just decomposes into blocks of 2 by 2 matrices, of the form $\left(\begin{matrix} iB_0m^{-1}& -m^{-1} \\ \lambda_i &-iB_0m^{-1}
\end{matrix} \right)$, indexed by $(i,i)$, $(N+i,i)$, $(i,N+i)$, and $(N+i,N+i)$, $i\in \{1,2\cdots N\}.$

Thus, the eigenvalues are $\mu_j^{\pm} = \pm \frac{i\sqrt{B_0^2+\lambda_j m}}{m}$.
Let $v_j$ be the eigenvectors of $B_{[N]}$, we then try to find two eigenvectors of the form $V_j = (v_j,u_j)$ to the matrix $A,$ where $u_j$ is to be determined. 
 From the two equations 
$$B_{[N]}v_j = \lambda_j v_j \text{ and }A V_j^{\pm} = \mu_j^{\pm}V_j^{\pm},$$
we obtain two vector-valued equations
 $$\begin{cases}-\frac{1}{m}B_0v_j-\frac{i}{m}u_j = \pm \left(\frac{1}{m}\sqrt{B_0^2+\lambda_j m}\right) v_j \\ i\lambda_j v_j+\frac{1}{m}B_0u_j = \pm \left(\frac{1}{m}\sqrt{B_0^2+\lambda_j m} \right)u_j\end{cases}\;\;\text{for all } j\in\{1,2,\dots,N\}.$$
 From the first equation we get $$ u_j = i\left(B_0\pm\sqrt{B_0^2+\lambda_j m}\right)v_j$$ and we readily verify that this choice of $u_j$ also satisfies the second set of equations. 
\end{proof}
Thus, the operator $\mathcal A$ that we are interested in is almost diagonalizable. We shall now turn to the spectral analysis of $\mathcal A$, but in order to do so, we recall \emph{Sylvester’s determinant identity}.
\begin{lemm}[Sylvester’s determinant identity]
\label{lemm:Sylvester}
 Let $A\in\CC^{p\times n}$ and $B\in\CC^{n\times p}$, then $$\det(I_p+AB) = \det(I_n+BA).$$
\end{lemm}

We can now relate the spectrum of $A_{\Gamma}$ to the matrix $A$ whose full eigendecomposition we have already exhibited in Prop. \ref{spec:A0}.
\begin{prop}[Wigner reduction]
\label{prop:Wigner}
 With $A_{\Gamma}$ as defined before, we have for $$\lambda\in\CC \backslash \Spec(A) \subset \CC \setminus i \mathbb R,$$ the spectral equivalence $$\lambda\in \Spec(A_{\Gamma})\iff -i\in \Spec(W_F(\lambda)).$$
We introduced 
 \begin{equation}
 \label{eq:WF}
     W_F(\lambda) = \sum_{\mu\in \Spec(A)}(\lambda-i\mu)^{-1}\sum_{i_1,i_2\in F}\alpha_{i_1,i_2}(\mu)\pi_{i_1,i_2}^{|F|}
 \end{equation}
with rank one operators $\pi_{i_1,i_2}^{|F|} = e_{i_1}^{|F|}\otimes e_{i_2}^{|F|}\footnote{Hence, $\pi_{i_1,i_2}^{|F|}\in \CC^{|F|\times|F|}$ is a matrix with zeros everywhere except for a single one at the $i_1$th row, $i_2$th column},$ where $e_k^j$ is the $k$-th unit vector in $\mathbb R^j$ and $$\alpha_{i_1,i_2}(\mu = \mu_j^{\pm}):=\sqrt{\gamma_{i_1}\gamma_{i_2}}\left\langle V^{\pm}_j, e^{2N}_{i_1}\right\rangle\left\langle e^{2N}_{i_2}, V^{\pm}_j\right\rangle$$
 with $V^{\pm}_j$ the eigenvectors of $A$ introduced in Prop. \ref{spec:A0}.
\end{prop}
 \begin{proof}
We start by introducing $\displaystyle Q_{F} = \left\{\sqrt{\gamma_a}e^{2N}_a(i)\right\}_{i\in[2N], a\in F}\in\CC^{2N\times |F|}$.

Thus, $Q_{F}Q_{F}^* = \hat{\Gamma}$ by an explicit computation. 
Further, for $\lambda \notin \pm \Spec(A)$, one observes that $$W_{F}(\lambda)={Q_F^*}(\lambda I-iA)^{-1}Q_F\in\CC^{|F|\times |F|}$$
and thus we have
\[
\begin{split}
    \det\left( I_{|F|}-iW_{F}(\lambda)\right) & = \det\left(I_{|F|}-iQ_F^*(\lambda I-iA)^{-1}{Q_F}\right)\\
    & = \det\left(I_{2N}-i(\lambda I-iA)^{-1}Q_FQ_F^*\right) \\
    & = \det\left(I_{2N}-i(\lambda I-iA)^{-1}\hat{\Gamma}\right) \\
    & = \det\left((\lambda I-iA)^{-1}(\lambda I-iA-i\hat{\Gamma})\right) \\
    & = \det\left((\lambda I-iA)^{-1}\right)     \det\left(\lambda I-A_{\Gamma}\right),
\end{split}
\]
where in the second line we applied Lemma \ref{lemm:Sylvester}. 
Since $\det\left((\lambda I-iA)^{-1}\right)\neq 0$, we get that $$\lambda\in \Spec(A_{\Gamma})\iff -i\in \Spec(W_F(\lambda)).$$

\iffalse 
It remains to show that the above definition of $W_F(\lambda)$ coincides with the definition in the statement of the Lemma. 

Let $V$ be the matrix of eigenvectors of $A$, so that each column of $V$ is $V^{\pm}_j$ for some j and some sign.

By Spectral decomposition we have $$A = V\Lambda V^{-1}$$ and hence 
$\left(\lambda I-V\Lambda V^{-1}\right)^{-1} = V \left(\lambda I -\Lambda\right)^{-1}V^{-1}.$
Now, let us see what $\mathcal{A_F^*}D\mathcal{A_F}$ gives us for a matrix $D$ of proper size: $\mathcal{A_F}$ on the right picks out the columns of $D$ with indices $i_2\inF$ and times them by $\sqrt{\gamma_{i_2}}$; then, $\mathcal{A_F^*}$ on the left picks out out rows of the truncated $D$ with indices $i_1\inF$ and times them by $\sqrt{\gamma_{i_1}}$.

So $$\mathcal{A_F^*}D\mathcal{A_F} = \sum_{i_1,i_2\inF}\sqrt{\gamma_{i_1}\gamma_{i_2}}\left\langle D, e^{2N}_{i_1}\right\rangle\left\langle e^{2N}_{i_2}, D\right\rangle \pi^{|F|}_{i_1,i_2}$$

Plugging in $D = V \left(\lambda I -\Lambda\right)^{-1}V^*$, we get
$$\mathcal{A_F^*}V \left(\lambda I -\Lambda\right)^{-1}V^*\mathcal{A_F} = \sum_{\pm}\sum_{j=1}^{N}(\lambda\mp\lambda_j)\sum_{i_1,i_2\inF}\sqrt{\gamma_{i_1}\gamma_{i_2}}\left\langle V^{\pm}_j, e^{2N}_{i_1}\right\rangle\left\langle e^{2N}_{i_2}, V^{\pm}_j\right\rangle \pi^{|F|}_{i_1,i_2}$$
$$=\sum_{\mu\in \pm \Spec(B_{[N]})}(\lambda-\mu)^{-1}\sum_{i_1,i_2\in F}\alpha_{i_1,i_2}(\mu)\pi_{i_1,i_2}^{|F|}$$
By rearrangement.
\fi
This completes the proof.
\end{proof}
So far we have not imposed any assumptions on $\eta$, $\xi$ and $\gamma$. From now on, we shall assume that all  $\eta$, $\xi$, and $\gamma$ coincide, respectively. In our previous article \cite{BM20}, we discussed the case $F=\{1,N\}$ in great detail, which implies that the Wigner matrix $W_{F}$ in Prop. \ref{prop:Wigner} is $2 \times 2.$ One can then diagonalize this matrix explicitly and reduce the spectral analysis to a scalar-valued Wigner matrix again. From now on we will assume that the masses are normalised to $1$ in order to simplify the notation. 
%To avoid this additional layer of technicality, we shall assume $F=\{1\}$ to start with. 
By our assumptions on $\eta$ and $\xi$, the matrix $B_{[N]}$ has explicit normalized eigenvectors, namely the eigenvectors of the discrete Dirichlet Laplacian
\begin{equation}
    v_j(i) = \begin{cases}N^{-\frac{1}{2}}, & \quad j=1\\ \sqrt{\frac{2}{N+1}}\sin \left(\frac{ij \pi }{N+1}\right),&\quad j\neq 1\end{cases}
\end{equation}
with corresponding eigenvalues $\lambda_j(B_{[N]}) = 4\xi \sin^2\left(\frac{\pi j}{2(N+1)}\right)+\eta+\frac{B_0^2}{2}$ for $ j\in [N].$ In particular, we have $$\lambda_1(B_{[N]}) \le \lambda_2(B_{[N]}) \le ...\le \lambda_N(B_{[N]}).$$

\subsection{Full spectrum with magnetic field}
We are now ready to prove our main theorem for the full spectrum of $M_{[N]}$ and thus on the spectrum of the full Fokker-Planck operator. The smallness assumption on the friction is necessary as Fig.\ref{fig:my_label} shows. Theorem \ref{theo:mag} then follows from the following theorem and Theorem \ref{theo:metafune}. 

\begin{figure}
    \centering
    \includegraphics[width=7.5cm]{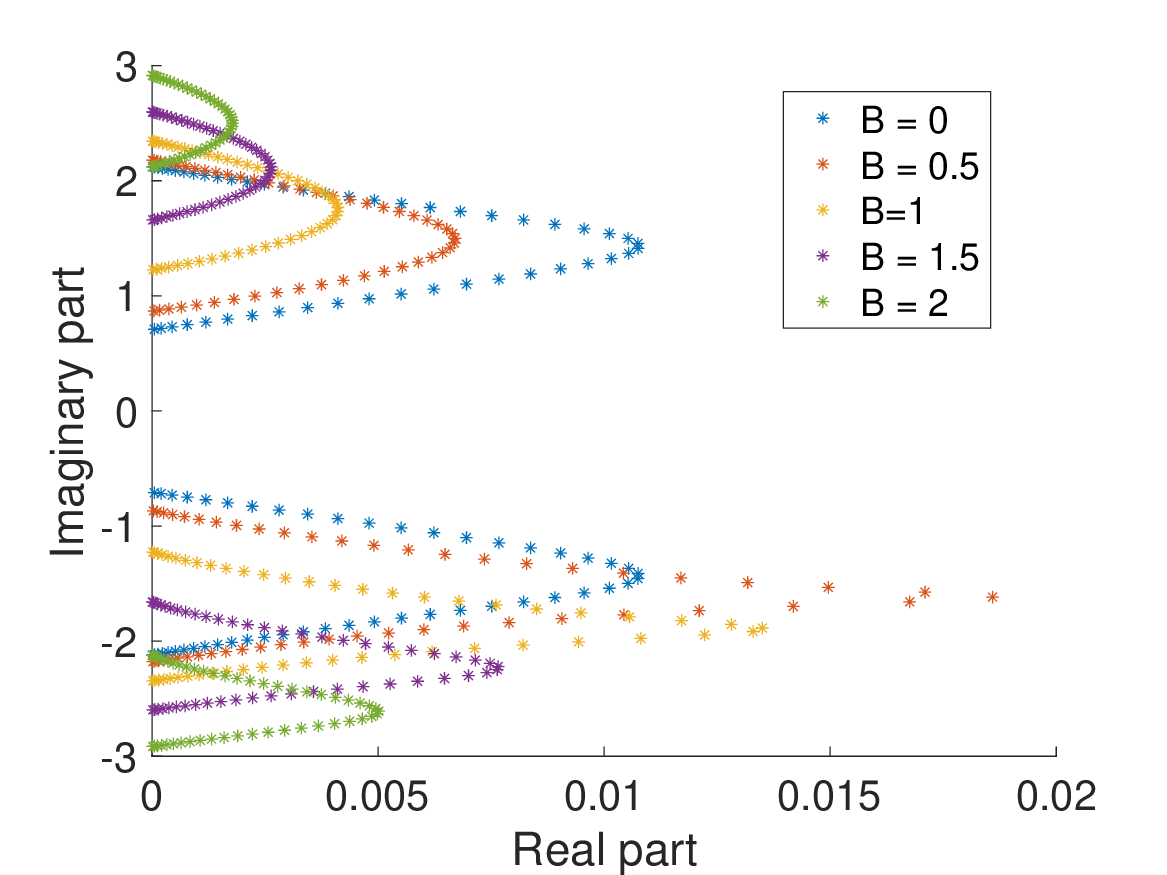}
    \includegraphics[width=7.5cm]{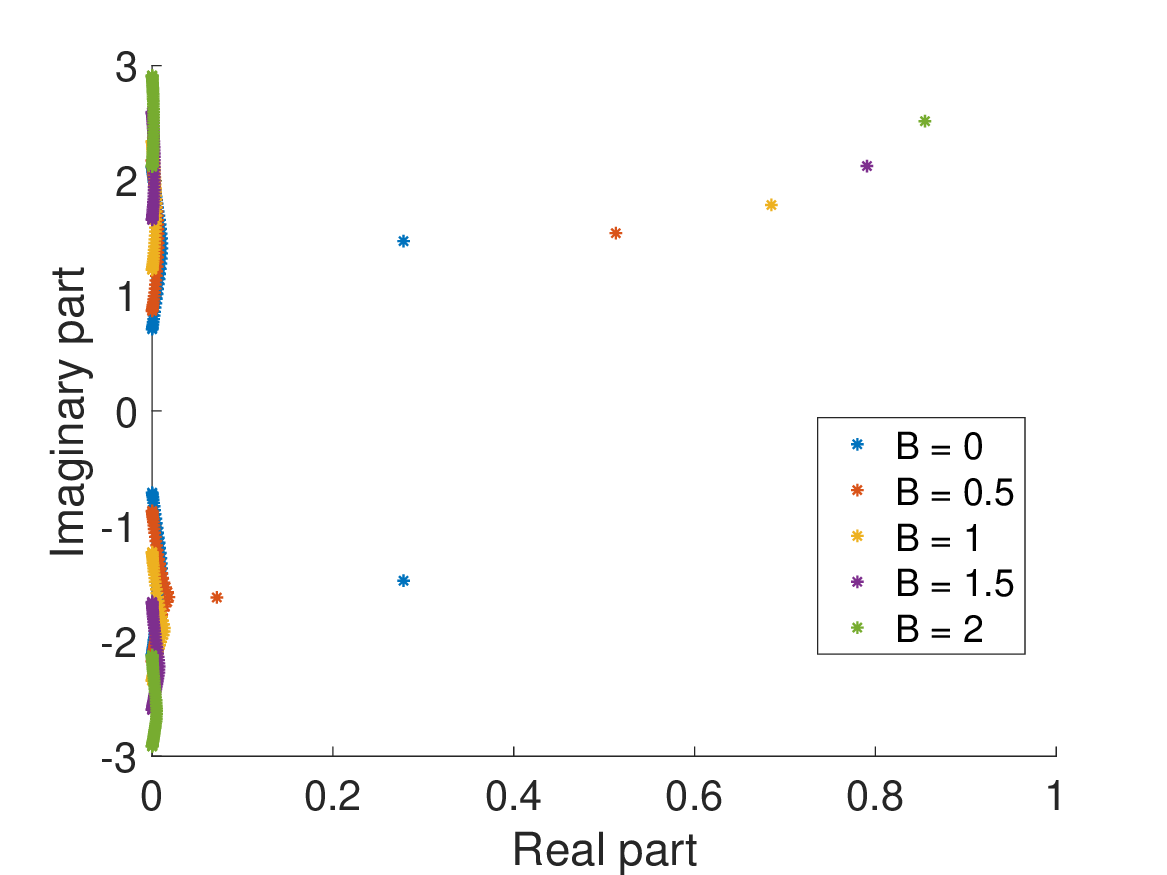}\\
    \includegraphics[width=7.5cm]{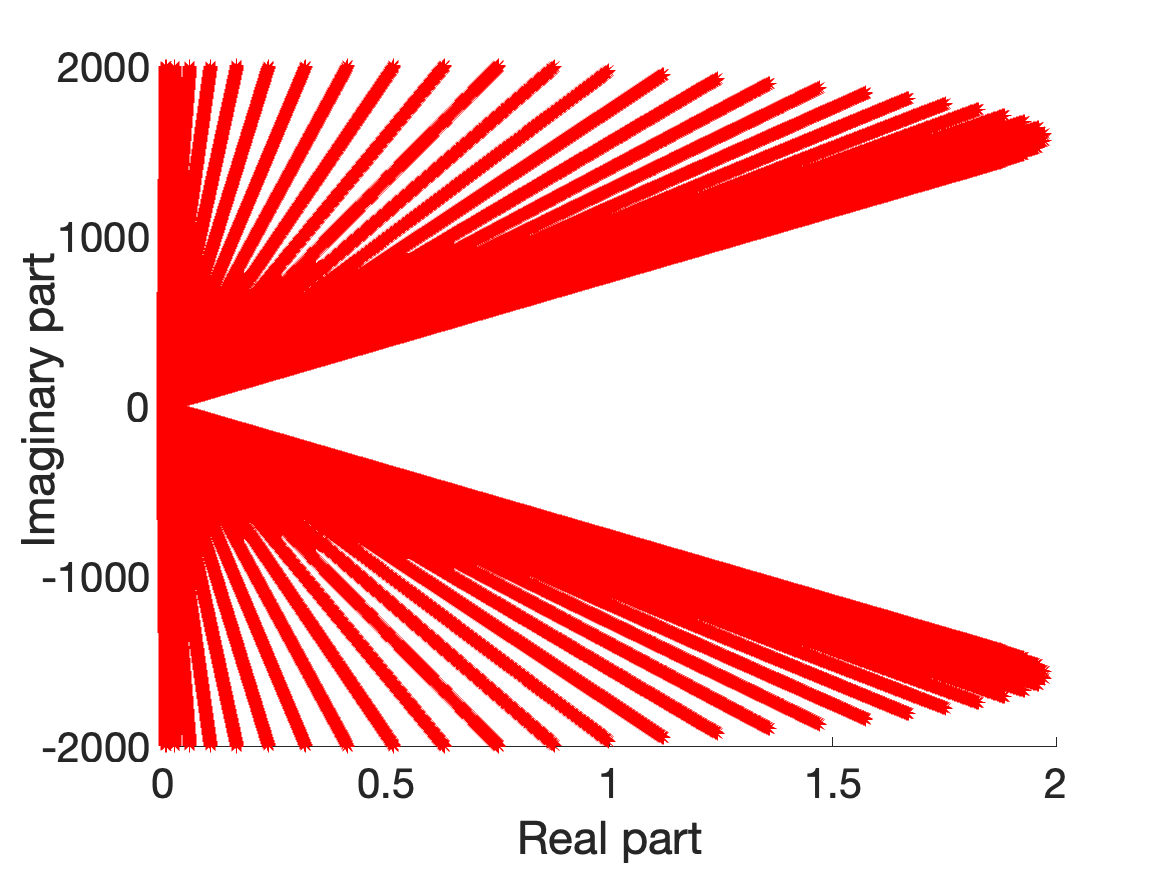}
    \includegraphics[width=7.5cm]{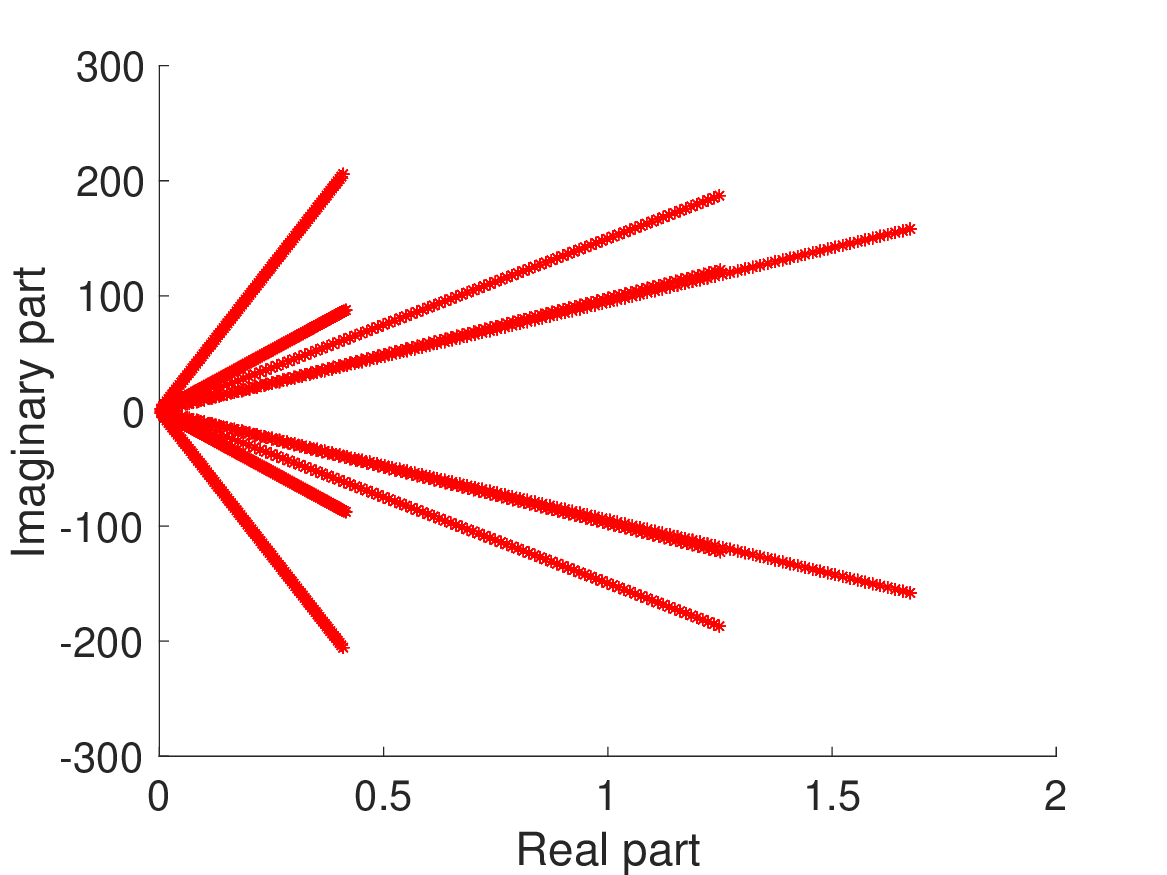}
    \caption{ $-i\Spec(A_{\Gamma})$ with $\gamma=0.1$ (top left), with $\gamma = 1$ (top right).  The magnetic field $B=B_0$ creates an asymmetry in the spectrum, but does not change it qualitatively. Strong friction creates eigenvalues with very large real part. 
    Spectrum of the lowest lying eigenvalues of full Fokker-Planck operator (bottom left)  for $N=50$, $B=0$ particles with $\gamma = 0.1$ and (bottom right) for $N=5$, $B=0.$ }
    \label{fig:my_label}
\end{figure}

\begin{theo} \label{theo: spectrum_magnetic field}
Let the friction $\gamma>0$ be sufficiently small, $F=\{1\},$ $B_0 \in \mathbb R$ arbitrary, not both $B_0$ and $\eta$ equal to zero, as well as $N$ large enough. The matrix $M_{[N]}$ has $4N$ eigenvalues $\lambda_i^{\pm}(j)$, with $j \in \ZZ_2$, $i \in [N]$ such that two of them are located in annuli 
\begin{equation}
\label{eq:annuli}
    \lambda_i^{\pm}(j) \in B_{\CC}(\mu_i^{\pm},c(B_0)\vert V_i^{\pm}(1)\vert^2) \backslash B_{\CC}(\mu_i^{\pm},\vert V_i^{\pm}(1)\vert^2  \varepsilon(\gamma,B_0)) \text{ for } j\in \ZZ_2
    \end{equation}and 
    \begin{equation}
    \label{eq:really}
\varepsilon(\gamma,B_0)\vert V_i^{\pm}(1)\vert^2 \lesssim \Re(\lambda_i^{\pm}(j)) \lesssim \vert V_i^{\pm}(1)\vert^2
    \end{equation}for some suitable $c>0$ and $\varepsilon(\gamma,B_0)$ small enough. For arbitrary $\gamma>0$ and $B_0\in \mathbb R$ the spectral gap $g(N):=\inf_{i,\pm,j}\Re(\lambda_i^{\pm}(j)) $ always decays like $1/N^3.$
\end{theo}
\begin{proof}
In our proof we shall focus on Dirichlet boundary conditions. The argument for Neumann boundary conditions is the same but uses the corresponding explicit eigendecomposition of the discrete Neumann Laplacian.
%\begin{proof}\hfill
By \eqref{eq:spectrum} it suffices to study $A_{\Gamma}$ instead of $M_{[N]}.$
We first need to ensure that the difference between the eigenvalues of $A$ is lower-bounded. We define for fixed $i \in [N]$ and sign $s \in \{\pm\}$ the differences $\kappa_j^{\pm} =-i(\mu_j^{\pm} - \mu_i^s$) and we estimate for $(\pm,j) \neq (s,i)$
\begin{equation}
\begin{split}
\label{eq:kappa_estm}
|\kappa_j^{\pm}| &\gtrsim 
|\mu_j^{\pm} -\mu_{j+1}^{\pm}| \gtrsim  |(\mu_j^{\pm})^2 -(\mu_{j+1}^{\pm})^2| \\ & 
= \frac{1}{m}\left\vert  \frac{B_0^2}{m}+ 4 \xi \sin^2\left( \tfrac{\pi j}{2(N+1)}\right) + \eta+\frac{B_0^2}{2} - \left(\frac{B_0^2}{m} + 4
\xi \sin^2\left( \tfrac{\pi (j+1)}{2(N+1)}\right) + \eta+\frac{B_0^2}{2} \right) \right\vert \\ & = 
\left\vert  \frac{4\xi}{m} \Big(\sin^2\left( \tfrac{\pi j}{2(N+1)}\right) -  \sin^2\left( \tfrac{\pi (j+1)}{2(N+1)}\right) \Big)\right\vert.
\end{split}
\end{equation}
In particular, the difference of the eigenvalues $\kappa_j^{\pm}$ does not decay faster than the squared norm of the first entry of the eigenvectors. By this we mean that since $\vert v_j(1) \vert^2 = \frac{2}{N+1} \sin\Big(\frac{\pi j}{N+1} \Big)^2 $, we find that $$\vert v_j(1) \vert^2 \lesssim \frac{ m \vert \lambda_{j}-\lambda_{j+1}  \vert}{\xi},\ \text{ as }$$ 
\begin{equation}
\begin{split}
\label{eq:uniformly_bounded}
\inf_{x \in [-\pi,\pi]} \lim_{N \to \infty} \frac{2\xi(N+1)\big\vert \sin^2\big(\tfrac{x}{2}+\frac{1}{2(N+1)}\big) - \sin^2(\tfrac{x}{2}) \big\vert}{m\sin(x)^2}=\inf_{x \in [-\pi,\pi]} \frac{\xi}{m\vert \sin(x)\vert} = \frac{\xi}{m}.
\end{split}
\end{equation}

\smallsection{Proof of \eqref{eq:annuli}:} 
Now we restrict to $F=\{1\}$ with friction constant $\gamma$. Since we want to study the solutions of the equation  $W_{F}(\lambda) + i =0 $ , we write
$$ W_{F}(\lambda) = \gamma \sum_{\pm,j} (\lambda - \mu_j^{\pm})^{-1} |V_j^{\pm}(1)|^2, $$
where $(\mu_j^{\pm},V_j^{\pm})$ is the eigensystem of $iA.$

We then localise our eigenproblem to the $i$-th eigenvalue with sign $s$, i.e. define $R_{F}(\lambda)=W_{F}(\lambda+\mu_i^{s})$. We also 
use the expansion
$$(\lambda-\mu)^{-1} = -\mu^{-1}\sum_{n\ge 0}(\lambda \mu^{-1})^n = -\mu^{-1}-\mu^{-2}\lambda-\mu^{-2}\lambda^2(\mu-\lambda)^{-1},$$
so that 
\begin{equation}
    \lambda(R_{F}(\lambda)+i)u = (f(\lambda)+g(\lambda))u
\end{equation}
where in terms of $\kappa_j^{\pm}:=\mu_j^{\pm}-\mu_i^s$
\[ f(\lambda) = \nu \lambda \text{ where } \nu = - i - \gamma \sum_{\pm,j \setminus\{s,i\}} \frac{\vert V_j^{\pm}(1)\vert^2}{ \kappa_j^{\pm} }\]
and 
\[ g(\lambda) =\gamma\Bigg(\vert V^{s}_i(1)\vert^2-\sum_{\pm,j \setminus\{s,i\}} \frac{\vert V_j^{\pm}(1)\vert^2 \lambda^2}{(\kappa_j^{\pm})^2}-\sum_{\pm,j \setminus\{s,i\}} \frac{\vert V_j^{\pm}(1)\vert^2 \lambda^2}{(\kappa_j^{\pm})^2(\kappa_j^{\pm}/\lambda-1)}\Bigg).\]  

We then notice that since $\kappa^{\pm}_j$ is uniformly bounded away from zero when $\pm \neq s$, we find
\begin{align} \label{eq:est1_term in f}
\sum_{\pm,j \setminus\{s,i\}} \frac{\vert V_j^{\pm}(1)\vert^2}{ \kappa_j^{\pm} }  = \mathcal O(1) +\sum_{j \neq i} \frac{\vert V_j^{s}(1)\vert^2}{ \kappa_j^{s} }, 
\end{align} 
since $\vert V_j^{\pm}(1)\vert^2 = \mathcal O(1/N)$ independent of $j$.
By symmetry, we may assume that $s=+$. Then given that  $|\kappa_j^{+}| =\frac{| (\mu_j^{+})^2-(\mu_i^+)^2|}{|\mu_j^{+}+\mu_i^+|} \sim | (\mu_j^{+})^2-(\mu_i^+)^2|$, where we used that $|\mu_j^{+}+\mu_i^+|$ is bounded uniformly from above and below, getting rid of the square roots on the eigenvalues. Thus,
we choose $x_*=\frac{j}{N+1}$ such that $\zeta:=\sin(\pi x_{*}/2)^2 $, where we write $\zeta$ itself as $\zeta =: \frac{a}{4\xi}$ for a fixed constant $a \in (0,4\zeta)$. We also introduce $x_{\downarrow}:=x_*-\frac{1}{N}$ and $x_{\uparrow}:=x_*+\frac{1}{N}.$
We remark that since one can split the interval $[0,1]$ into subintervals on which the function $h(u)=\frac{\sin(\pi u)^2}{4\xi\sin(\pi u/2)^2-a}$ is monotonous, by Appendix \ref{lemm:Riemann} we can approximate the sum by the Riemann integrals as follows:
%We have that 
%\[\begin{split}
%|\lambda_{j}-\lambda_{j+1} | = \left\vert 2\sin\Big( \frac{\pi (j+1)}{2(N+1)}\Big)-2\sin\Big( \frac{\pi j}{2(N+1)}\Big) \right\vert. \end{split}\]

\[\begin{split} \sum_{j \neq i} \frac{\vert V_j^{+}(1)\vert^2}{ \kappa_j^{+} }&\sim \int_{0}^{x_{\downarrow}} +\int_{x_{\uparrow}}^1\frac{\sin(\pi x)^2}{4\xi\sin(\pi x/2)^2-a} \ dx +  \frac{\vert V_{i-1}^{+}(1)\vert^2}{\kappa_{i-1}^+} +\frac{\vert V_{i+1}^{+}(1)\vert^2}{\kappa_{i+1}^+} \end{split}.\]
The last two terms are bounded by the first argument of the proof. 

We then simplify
\[\frac{\sin(\pi x)^2}{\sin(\pi x/2)^2-\frac{a}{4\xi}} =\frac{4\sin(\pi x/2)^2(1-\sin(\pi x/2)^2)}{\sin(\pi x/2)^2-\frac{a}{4\xi}} = \frac{4t^2(1-t^2)}{t^2-\zeta}\]
with $t:=\sin(\pi x/2).$ 
Using partial fraction decomposition, we then find 
\begin{equation}
\label{eq:sum}
\frac{4t^2(1-t^2)}{t^2-\zeta} = -4t^2 -4(\zeta-1) + \frac{4(\zeta-\zeta^2)}{t^2-\zeta}. 
\end{equation}

For the third term, we observe that 
\[\int_0^{x_{\downarrow}} \frac{1}{\sin(\pi x/2)^2-\zeta} \ dx = -\frac{2 \operatorname{arctanh}\Big(\sqrt{\frac{1-\zeta}{\zeta}} \tan\Big(\frac{\pi x_{\downarrow}}{2} \Big) 
\Big)}{\pi\sqrt{\zeta-\zeta^2}} \]
as well as
\[\int_{x_{\uparrow}}^{1} \frac{1}{\sin(\pi x/2)^2-\zeta} \ dx = \frac{2 \operatorname{arccoth}\Big(\sqrt{\frac{1-\zeta}{\zeta}} \tan\Big(\frac{\pi x_{\uparrow}}{2} \Big) 
\Big)}{\pi\sqrt{\zeta-\zeta^2}}. \]
This implies that, using \eqref{eq:sum},
\[
\begin{split}\int_{0}^{x_{\downarrow}} +\int_{x_{\uparrow}}^1&\frac{\sin(\pi x)^2}{\sin(\pi x/2)^2-\zeta} \ dx = O(\vert x_{\uparrow} - x_{\downarrow} \vert) \\
&+ \tfrac{8 \sqrt{\zeta - \zeta^2}}{\pi}\Big(\operatorname{arccoth}\Big(\sqrt{\tfrac{1-\zeta}{\zeta}} \tan\Big(\tfrac{\pi x_{\uparrow}}{2} \Big) 
\Big)-\operatorname{arctanh}\Big(\sqrt{\tfrac{1-\zeta}{\zeta}} \tan\Big(\tfrac{\pi x_{\downarrow}}{2} \Big) 
\Big)\Big).\end{split}
\]

For $\delta>0$ small
\[\operatorname{arccoth}(1+\delta) - \operatorname{arctanh}(1-\delta) = \mathcal O(\delta).\]
Hence, it follows that
\begin{equation}
\label{eq:previous}
\begin{split}
   &\sqrt{\zeta - \zeta^2}\Big(\operatorname{arccoth}\Big(\sqrt{\tfrac{1-\zeta}{\zeta}} \tan\Big(\tfrac{\pi x_{\uparrow}}{2} \Big) 
\Big)-\operatorname{arctanh}\Big(\sqrt{\tfrac{1-\zeta}{\zeta}} \tan\Big(\tfrac{\pi x_{\downarrow}}{2} \Big) 
\Big)\Big)\\
&= \mathcal O\Big(\frac{(1-\zeta)\vert x_{\downarrow}-x_{\uparrow} \vert}{\cos(\pi x_{\uparrow}/2)\cos(\pi x_{\downarrow}/2)}\Big) = \mathcal O\Big(\frac{\cos(\pi x_{*}/2)^2 \vert x_{\uparrow}-x_{\downarrow}\vert}{\cos(\pi x_{\uparrow}/2)\cos(\pi x_{\downarrow}/2)}\Big)\\
&=\mathcal O(\vert x_{\uparrow}-x_{\downarrow}\vert).\end{split}
\end{equation}
Since $\zeta 
=\sin(\pi x_{*}/2)^2,$ it follows that
\[\sum_{j \neq i} \frac{\vert V_j^{+}(1)\vert^2}{ \kappa_j^{+} }= \mathcal O(1).\] 
To study the terms appearing in the polynomial $g$, we estimate
\begin{align} \label{eq:est2_term in g}
\begin{split} \sum_{j \neq i} \frac{\vert V_j^{+}(1)\vert^2}{ (\kappa_j^{+})^2 }&\sim \int_0^{x_{\downarrow}} +\int_{x_{\uparrow}}^1\frac{\sin(\pi x)^2}{(\sin(\pi x/2)^2-\zeta)^2} \ dx + \frac{\vert V_{i-1}^{+}(1)\vert^2}{(\kappa_{i-1}^+)^2} +\frac{\vert V_{i+1}^{+}(1)\vert^2}{(\kappa_{i+1}^+)^2} \end{split}.
\end{align} 
The last two terms grow as $\mathcal{O}(N)$ since $$\frac{\vert V_{i-1}^{+}(1)\vert^2}{(\kappa_{i-1}^+)^2} = \frac{\sin^2\left(\frac{\pi (i-1)}{N+1}\right)}{N\left[\sin^2\left(\frac{\pi (i-1)}{2(N+1)}\right) - \sin^2\left(\frac{\pi (i-1)}{2(N+1)} + \frac{\pi}{2(N+1)}\right) \right]^2} \sim \mathcal{O}(N).$$

For the integral terms, we recall the partial fraction decomposition 
\begin{equation} \frac{4t^2(1-t^2)}{(t^2-\zeta)^2} = \frac{4(\zeta-\zeta^2)}{(t^2-\zeta)^2}- \frac{4(2\zeta-1)}{t^2-\zeta} - 4.
\end{equation}
By using the previous computation \eqref{eq:previous} and the partial fraction decomposition, we compute
\[\begin{split}\int_0^{x_{\downarrow}} +\int_{x_{\uparrow}}^1\frac{\sin(\pi x)^2}{(\sin(\pi x/2)^2-\zeta)^2} \ dx &=\mathcal O(\vert x_{\uparrow} - x_{\downarrow}\vert) + \mathcal O\Big(\frac{(2\zeta -1)\sqrt{1/\zeta-1} \vert x_{\downarrow}-x_{\uparrow} \vert}{\cos(\pi x_{\uparrow}/2)\cos(\pi x_{\downarrow}/2)}\Big) \\
&+\int_0^{x_{\downarrow}} +\int_{x_{\uparrow}}^1\frac{4(\zeta-\zeta^2)}{(\sin(\pi x/2)^2-\zeta)^2} \ dx.
\end{split}\]
We then find that 
\begin{equation} 
\begin{split}\int_0^{x_{\downarrow}} &+\int_{x_{\uparrow}}^1\frac{\zeta-\zeta^2}{(\sin(\pi x/2)^2-\zeta)^2} \ dx = \frac{1}{2\pi} 
\Big(\frac{\sin(\pi x_{\uparrow})}{\sin(\pi x_{\uparrow}/2)^2-\zeta}-\frac{\sin(\pi x_{\downarrow})}{\sin(\pi x_{\downarrow}/2)^2-\zeta}\Big) \\
+&(2\zeta-1) \frac{\Big(\operatorname{arccoth}\Big(\sqrt{\tfrac{1-\zeta}{\zeta}} \tan\Big(\tfrac{\pi x_{\uparrow}}{2} \Big) 
\Big)-\operatorname{arctanh}\Big(\sqrt{\tfrac{1-\zeta}{\zeta}} \tan\Big(\tfrac{\pi x_{\downarrow}}{2} \Big)\Big)\Big)}{\sqrt{\zeta-\zeta^2}} \\
&= \mathcal O\Big(\vert x_{\uparrow}-x_*\vert^{-1}\Big)+\mathcal O\Big( \frac{(2\zeta-1)\vert x_{\downarrow}-x_{\uparrow}\vert}{\zeta-\zeta^2}\Big)=\mathcal O\Big(\vert x_{\uparrow}-x_*\vert^{-1}\Big).
\end{split}
\end{equation}

The final term in $g$, i.e.
$$-\sum_{\pm,j \setminus\{s,i\}} \frac{\vert V_j^{\pm}(1)\vert^2 \lambda^2}{(\kappa_j^{\pm})^2(\kappa_j^{\pm}/\lambda-1)},$$ we notice that it differs from the second term only by the additional factor $\frac{\kappa_j^{\pm}}{\lambda}-1$ in the denominator and since $\lambda$ is on the circle with radius $c \vert V_i^s(1)\vert^2$ with $c$ small, the term satisfies by \eqref{eq:uniformly_bounded} 
\[ \Big\vert \tfrac{\kappa_j^{\pm}}{\lambda}-1 \Big\vert \ge \mathcal O(1).\]
We shall not discuss the last term in $g$, which is  since by the assumptions on $\lambda$ that we shall impose now, it can be treated just like the second term. For this, it suffices to recall that the distance between eigenvalues is not smaller than the decay rate of the first component of the eigenvector, as shown in \eqref{eq:uniformly_bounded}.

%For the third term appearing in the polynomial $g(\lambda)$, we see that its additional factor in the summand $(\frac{\kappa_j^{\pm}}{\lambda} -1)^{-1}$ is of order $|\lambda|$. Thus this  is a term of order $\mathcal{O}(|\lambda|^3 N)$. 

If $\vert \lambda \vert = c\vert V_i^{s}(1)\vert^2,$  for arbitrary $c>0$ and $\gamma(c)>0$ sufficiently small, then 
\begin{equation}
    \label{eq:f>g}
    \vert f(\lambda)\vert >
\vert g(\lambda)\vert
\end{equation}
with $f$ having precisely one zero at zero inside $B(0,\vert \lambda \vert).$
This shows, by Rouch\'e's theorem that there is precisely one eigenvalue, say $\lambda_i^{\pm}$,  in $B(\mu_i^s,c \vert V_i^{s}(1)\vert^2).$ 

If $\vert \lambda \vert =\varepsilon(\gamma) \vert V_i^{s}(1)\vert^2,$ for $\varepsilon(\gamma)$ sufficiently small, then $\vert g(\lambda)\vert >
\vert f(\lambda)\vert$ and $g$ does not have a zero inside $B(0,\vert \lambda \vert).$
This shows, by Rouch\'e's theorem that $|\lambda_i^{\pm}| >\varepsilon(\gamma)\vert V_i^{s}(1)\vert^2 $, so there is no eigenvalue in $B(\mu_i^s,\varepsilon(\gamma)\vert V_i^{s}(1)\vert^2).$ We remark that this lower estimate holds for any fixed $\gamma.$ This shows \eqref{eq:annuli}.

\smallsection{Proof of \eqref{eq:really}:} 
Next, we continue by showing \eqref{eq:really}. 
Let us assume that we have a solution $\lambda$, as above in an annulus, to $f(\lambda)+g(\lambda)=0$ with the property that $
\Im(\lambda) = 
\Re(\lambda)o(1)$, then by taking the imaginary part of  $f(\lambda)+g(\lambda)=0$, we find 
\begin{equation}\label{eq:identity2}
\Re(\lambda) +\Im(\lambda)\sum_{j \neq i} \frac{\vert V_j^{+}(1)\vert^2}{ \kappa_j^{+} } + 
\Im(g(\lambda))=0.
\end{equation}
Now since $\sum_{j \neq i} \frac{\vert V_j^{+}(1)\vert^2}{ \kappa_j^{+} } = \mathcal O(1)$, we have from \eqref{eq:identity2} \begin{equation} 
\label{eq:Real}
\Re(\lambda) (1+o(1)) + \Im(g(\lambda))=0. \end{equation}
This leads to a contradiction. Indeed, using \eqref{eq:est2_term in g}, we find 
\begin{equation} 
\label{eq:term2}
\vert \Im(g(\lambda))\vert = \mathcal{O}\Big(\vert x_{\uparrow}-x_*\vert^{-1}(\vert \Im(\lambda^2)\vert+\sup_j\vert\Im(\lambda^3/(\kappa_j-\lambda))\vert)\Big) 
\end{equation}

To further estimate the second term in \eqref{eq:term2}, we recall that $\vert \Im(\lambda)\vert = \vert \Re(\lambda) \vert o(1)$ as well as the identity $$\Im\Big(\frac{\lambda^3}{\kappa_j-\lambda}\Big)=\frac{\Re(\lambda)^3 \Im(\lambda)}{\Im(\lambda)^2 + z^2} + \frac{3 \Re(\lambda)^2 \Im(\lambda) z}{\Im(\lambda)^2 + z^2} - \frac{3 \Re(\lambda) \Im(\lambda)^3}{\Im(\lambda)^2 + z^2} - \frac{\Im(\lambda)^3 z}{\Im(\lambda)^2 + z^2}$$
with $z = \kappa_j - \Re(\lambda)$, which by \eqref{eq:uniformly_bounded} and the location of $\lambda$, see \eqref{eq:annuli}, satisfies
$$ \vert z \vert \gtrsim \vert x_{\uparrow}-x_*\vert \gtrsim \vert \Re(\lambda)\vert.$$
This implies that 
\[\Big\vert \Im\Big(\frac{\lambda^3}{\kappa_j-\lambda}\Big)\Big\vert = \mathcal O\Big(\vert x_{\uparrow}-x_*\vert^{-1}\vert\Re(\lambda)^2\Im(\lambda)\vert\Big). \]

Combining this with $\Im(\lambda^2)=2\Re(\lambda)\Im(\lambda)$ and our standing assumption $\vert \Im(\lambda)\vert = \vert \Re(\lambda) \vert o(1)$, we find
\begin{equation}
\label{eq:estn}
\begin{split}| \Im(g(\lambda))| &= \mathcal{O}\Big(\vert x_{\uparrow}-x_*\vert^{-1}(\vert \Im(\lambda^2)\vert+\sup_j\vert\Im(\lambda^3/(\kappa_j-\lambda))\vert)\Big) \\
&= \mathcal O\Big(\vert x_{\uparrow}-x_*\vert^{-1}(\vert \Re(\lambda)\Im(\lambda)\vert+\vert x_{\uparrow}-x_*\vert^{-1}\vert\Re(\lambda)^2\Im(\lambda)\vert )\Big) \\
&= \mathcal O(\vert \lambda \vert o(1)).\end{split}
\end{equation}

We conclude from $\vert \Im(\lambda)\vert = \vert \Re(\lambda) \vert o(1)$ and by combining \eqref{eq:Real} with \eqref{eq:estn} 
\[ \vert \lambda \vert = \vert \Re(\lambda)\vert +\vert \Im(\lambda)\vert  = \vert \Re(\lambda)\vert(1+o(1)) = \mathcal O(\vert \lambda \vert o(1)). \]

This, implies that $\lambda=0$ for large $N$ which is a contradiction to the absence of eigenvalues in $B(\mu_i^s,\varepsilon(\gamma) \vert V_i^{s}(1)\vert^2)$ showing \eqref{eq:really}.

\smallsection{Proof of spectral gap:} We finally show that even when dropping the assumption that $\gamma$ is small, the spectral decays like $1/N^3$. The upper bound is straightforward and follows readily from the argument around \eqref{eq:f>g} by choosing $\vert \lambda \vert = c \vert V_1^+(1)\vert^2$ and $c>0$ large enough. This choice does not impose any restrictions on $\gamma(c)>0$ to be small to ensure $\vert f(\lambda)\vert > \vert g(\lambda)\vert.$

For the lower bound, we proceed as follows. Since $W_F$, see \eqref{eq:WF}, depends only on the difference of $\lambda-i\mu_j^{\pm}$ we can shifted real quantities $i\mu_j^{\pm} \in \mathbb R$ by the real part of $\lambda$ to obtain new quantities $\nu_j^{\pm} \in \mathbb R$ and assume that any such solution $\lambda$ to $W_F(\lambda)+i=0$ is purely imaginary. We shall write $\lambda_{\text{imag}} \in i \mathbb R$ for it. In order to obtain a contradiction we assume that there is a solution $\lambda_{\text{imag}}$ purely imaginary that decays faster than $N^{-3}$, i.e. $\mathcal{O}(\lambda) = N^{-3}R_N^{-1}$ where $R_N \to \infty.$ Then
$$W_{F}(\lambda_{\text{imag}} = \tfrac{i}{N^3R_N}) = \gamma \sum_{\pm, j} (\tfrac{i}{N^3R_N}-\nu_j^{\pm})^{-1} \vert V_j^{\pm}(1)\vert^2=-i.$$

We thus have that by taking the imaginary part
\[ \Im(W_{F}(\lambda_{\text{imag}} = \tfrac{i}{N^3R_N})) = -\gamma \sum_{\pm,j} \frac{ \vert V_j^{\pm}(1)\vert^2 / (N^3 R_N)}{N^{-6}R_N^{-2}+ (\nu_j^{\pm})^2}. \]

We now want to sure that $\Im(W_{F}(\lambda_{\text{imag}} = \tfrac{i}{N^3R_N})) = o(1)$ as $N \to \infty$ to obtain a contradiction to $W_{F}(\lambda_{\text{imag}} = \tfrac{i}{N^3R_N}) = -i.$

As the argument is symmetric in the signs, we shall just focus on $s=+$.
There is one $\nu^+_{j_0}$ of smallest modulus (we allow $j_0$ to depend on $N$).
Let us now assume that for $\varepsilon=\varepsilon(\gamma)>0$ sufficiently small, there is $N_0(\varepsilon)$ such that for all $N \ge N_0(\varepsilon)$ we have $\vert \nu_{j_0}^+\vert \ge \varepsilon \vert V_{j_0}^{+}(1)\vert^2.$

Writing then
\[  \Im(W_{F}(\lambda_{\text{imag}} = \tfrac{i}{N^3R_N}))  = -\frac{\gamma}{N^3 R_N} \Bigg(\frac{ \vert V_{j_0}^{+}(1)\vert^2 }{N^{-6}R_N^{-2}+ (\nu_{j_0}^{+})^2} + \sum_{j \neq j_0} \frac{ \vert V_{j}^{+}(1)\vert^2}{N^{-6}R_N^{-2}+ (\nu_{j}^{+})^2}\Bigg),    \]
we see that the final sum can be estimated similar to $\sum_{j \neq j_0} \frac{\vert V_j^{+}(1)\vert^2}{ (\kappa_j^{+})^2 }$ such that the second term with prefactor $\frac{1}{N^3 R_N}$ tends to zero as $N$ tends to infinity.

Since we assumed that $\nu_{j_0}^{+} \notin B(0,\varepsilon\vert V_{j_0}^{+}(1)\vert^2),$ the first term satisfies 
\[ \frac{ \vert V_{j_0}^{+}(1)\vert^2 / (N^3 R_N)}{N^{-6}R_N^{-2}+ (\nu_{j_0}^{+})^2} = \mathcal O(1/(\varepsilon^2R_N)) = o(1).\]

Hence, we cannot have $W_F(\lambda_{\text{imag}})=-i$ and thus, we obtain a contradiction to our assumption implying $\nu_{j_0}^{+} \in B(0,o(1)\vert V_{j_0}^{+}(1)\vert^2)$. Reverting back from $\lambda_{\text{imag}}$ to $\lambda$ and $\nu_{j_0}$ to $\mu_{j_0}$, this condition implies that $\lambda \in B(\mu_{j_0}^{+},o(1) \vert V_{j_0}^{+}(1)\vert^2), $ but this has been ruled out before by showing \eqref{eq:eigenvalues}.

%$\Im(g(\lambda)) = \Re(\lambda)o(1)$ 
\end{proof}

 \section{Next-to-nearest-neighbour interactions}
 In this section we include next-to-nearest-neighbour interactions for the chain of oscillators. 
 Since we already discussed how to include the effect of a magnetic field in the previous section, we shall not consider it here anymore. Thus, it is also no longer necessary to consider the oscillators as particles in two dimensions, where the magnetic field couples positions and momenta in different directions. Thus, without loss of generality (separation of variables), we may restrict us to oscillators confined to one spatial dimension, i.e. position and momentum variables of the oscillators are one-dimensional variables, respectively.
 
 We recall that the generator of the dynamics is given by
\begin{equation}
\label{eq:longerrange}
 \mathcal{L} f(z)  = - \langle z,  M_{[N]}   \nabla_z f(z) \rangle + \langle \nabla_p ,  \Gamma \textbf{m}_{[N]}  \vartheta \nabla_p  f(z) \rangle
 \end{equation}
where $M_{[N]}  \in \mathbb C^{2N \times 2N}$ and $\Gamma \in \mathbb R^{N \times N}$, the matrix containing the friction parameters $\gamma_1, \gamma_N$, are matrices of the form
\begin{equation}
\begin{split}
M_{[N]}  &:= \left(\begin{matrix} \Gamma & -\textbf{m}_{[N]}^{-1} \\ B_{[N]}  & 0 \end{matrix}\right) \text{ and }
\Gamma = \text{diag}(\gamma_{1},0, \dots, 0, \gamma_{N} ).
\end{split}
\end{equation}
The matrix $\vartheta$ containing the temperatures is of the form $$\vartheta = \text{diag}(\beta_{1}^{-1},0,\dots,0, \beta_N^{-1}).$$

To specify the interaction matrix $B_{[N]}$ we shall specify our oscillator potential including both nearest neighbour and next-to nearest neighbour interactions
\begin{equation}
\label{eq:potential}
V(q) = \frac{1}{2}\sum_{i=1}^{N-1} (q_i-q_{i+1})^2 + \frac{\omega}{2} \sum_{i=1}^{N-2}(q_i-q_{i+2})^2 + \frac{\eta}{2}\sum_{i=1}^N q_i^2,
\end{equation}
with $\omega>0.$

As we shall see, when introducing long range interactions, unlike in the case of only nearest neighbour interaction, the boundary conditions do affect the behaviour of the spectral gap significantly.

\smallsection{Neumann boundary conditions}
When the heat flux at the terminal particles of the oscillator chain is zero, this corresponds to Neumann boundary conditions.
The matrix associated with the quadratic form of the potential \eqref{eq:potential} is given by
$$ \mathscr V_{[N]} = \begin{pmatrix}1 & -1 &  & & & & \\
-1 & 2 & \ddots & & & &   \\
   & \ddots & \ddots & \ddots & & &   \\
      & & \ddots & \ddots & \ddots &  &  \\
            & && \ddots & \ddots & \ddots &  &  \\
            & & & &-1 & 2 & -1   \\
                 & & & &&-1 & 1   \\
 \end{pmatrix} + \omega  \begin{pmatrix}1 &  &-1  & & &  &\\
 & 1 & & \ddots & &  & \\
-1   & & 2 & & \ddots &  & \\
      & \ddots& & \ddots & &  \ddots&    \\
            & &\ddots &  & 2 & &-1    \\
            & & &\ddots  &  & 1 &    \\
                 & & & &-1& & 1   \\
 \end{pmatrix} + \eta \operatorname{Id}. $$
 
 Hence, it follows that $V(q) = \langle q, \mathscr V_{[N]} q \rangle.$ By construction $\mathscr V_{[N]} =\mathscr V_1+\omega \mathscr V_2 + \eta \operatorname{Id}$, where both $\mathscr V_1 = \Delta_{[N]}^N$ and $\mathscr V_2$ are positive semi-definite matrices.
 
  We observe that we can write 
 \[ \mathscr V_2 = 4 \mathscr V_1 - \mathscr V_1^2 +\operatorname{diag}(\Lambda, 0,\cdots , 0, \Lambda)\text{ where }\Lambda= \begin{pmatrix} -1& 1 \\ 1 & -1 \end{pmatrix}.\]
 Hence, we have that 
\begin{equation}
\label{eq:Neumann_decomp}
 \mathscr V_{[N]} = \mathscr V_1 + 4\omega \mathscr V_1 + \eta \operatorname{Id}- \omega \mathscr V_1^2+ \omega \operatorname{diag}(\Lambda,0,\cdots, 0, \Lambda) =: T_\omega + \omega P 
 \end{equation}
 with $P= \operatorname{diag}(\Lambda,0,\cdots, 0, \Lambda).$
 
This implies that $\mathscr V_{[N]}$ has an almost explicit eigensystem, the one of the self-adjoint operator $T_{\omega}$, up to a rank $2$ perturbation $\operatorname{diag}(\Lambda,0,\cdots, 0, \Lambda).$
Indeed, in terms of $w_1 =  2^{-1/2}(-1,1,0,...,0)^t$ and $w_2 = 2^{-1/2}(0,...,0,-1,1)^t$, we have
\[ \operatorname{diag}(\Lambda,0,\cdots, 0, \Lambda)= -2w_1 \otimes w_1 -2 w_2  \otimes w_2.\]

%The $j$-th eigenfunctions of $T_\omega:= \mathscr V_1 + 4\omega \mathscr V_1- \omega \mathscr V_1^2$ is then given by
%\[ (v_{j})_i = \begin{cases}
%n^{- \frac{1}{2}} & \mbox{j = 1}\\
%\sqrt{\frac{2}{n}} \cos\left(\frac{\pi (j - 1)(i - 1/2)}{N}\right) & \mbox{otherwise}
%\end{cases}\]
%and the associated $j$-th eigenvalue is $\gamma_j= \lambda_j +4\omega \lambda_j-\omega \lambda_j^2 $ where 
%$\lambda_j = 4 \sin^2\left(\frac{\pi (j - 1)}{2N}\right).$

\smallsection{Dirichlet boundary conditions} By imposing Dirichlet boundary conditions we model terminal particles that are attached to fixed walls. Then the discrete Laplacian that describes the nearest-neighbour interactions is the Dirichlet Laplacian $\mathscr V_1 = \Delta_{[N]}^D$ with next-to nearest-neighbour interactions described by $\mathscr V_2$ $$\mathscr V_1 =  \begin{pmatrix}2 & -1 &  & & & & \\
-1 & 2 & \ddots & & & &   \\
   & \ddots & \ddots & \ddots & & &   \\
      & & \ddots & \ddots & \ddots &  &  \\
            & && \ddots & \ddots & \ddots &  &  \\
            & & & &-1 & 2 & -1   \\
                 & & & &&-1 & 2   \\
 \end{pmatrix} \text{ and }\mathscr V_2 = \begin{pmatrix}1 &  &-1  & & &  &\\
 & 2 & & \ddots & &  & \\
-1   & & 2 & & \ddots &  & \\
      & \ddots& & \ddots & &  \ddots&    \\
            & &\ddots &  & 2 & &-1    \\
            & & &\ddots  &  & 2 &    \\
                 & & & &-1& & 1   \\
 \end{pmatrix}. $$
The interaction matrix that includes the next-to-nearest-neighbour interactions of strength $\omega$ is in analogy with the Neumann case \eqref{eq:Neumann_decomp} given by $T_\omega +P$ where 
 \begin{align} \label{def:T_omega Dirichlet bc}
 T_\omega = (1+4\omega) \mathscr V_1 - \omega \mathscr V_1^2 + \eta \operatorname{Id}\ \text{ and }\ P= \operatorname{diag}(-2\omega, 0, \dots,0,-2\omega).
 \end{align}
 Therefore, also in the Dirichlet case the interaction matrix is still given in terms of powers of $\mathscr V_1$  perturbed by a matrix of rank 2.
\begin{figure}\label{with_or_without_P}
\includegraphics[width=7cm]{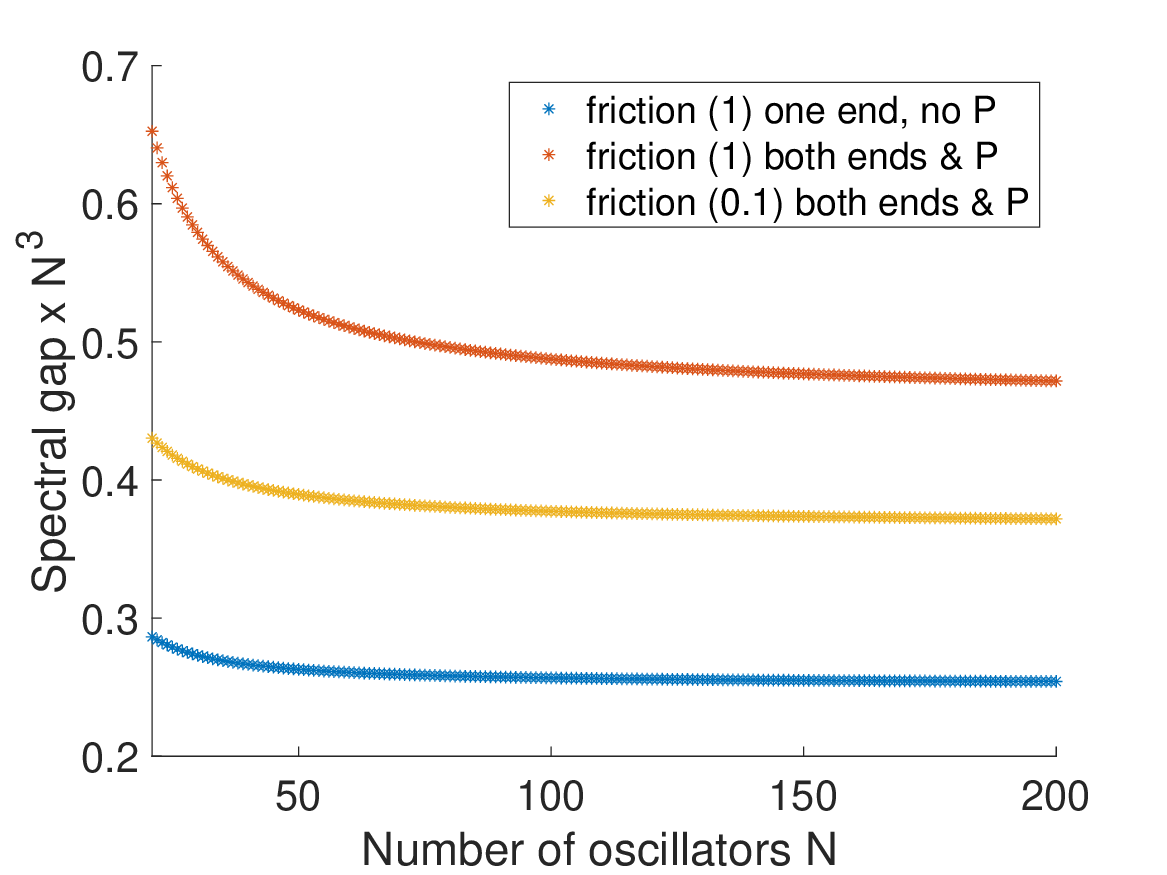}
\includegraphics[width=7cm]{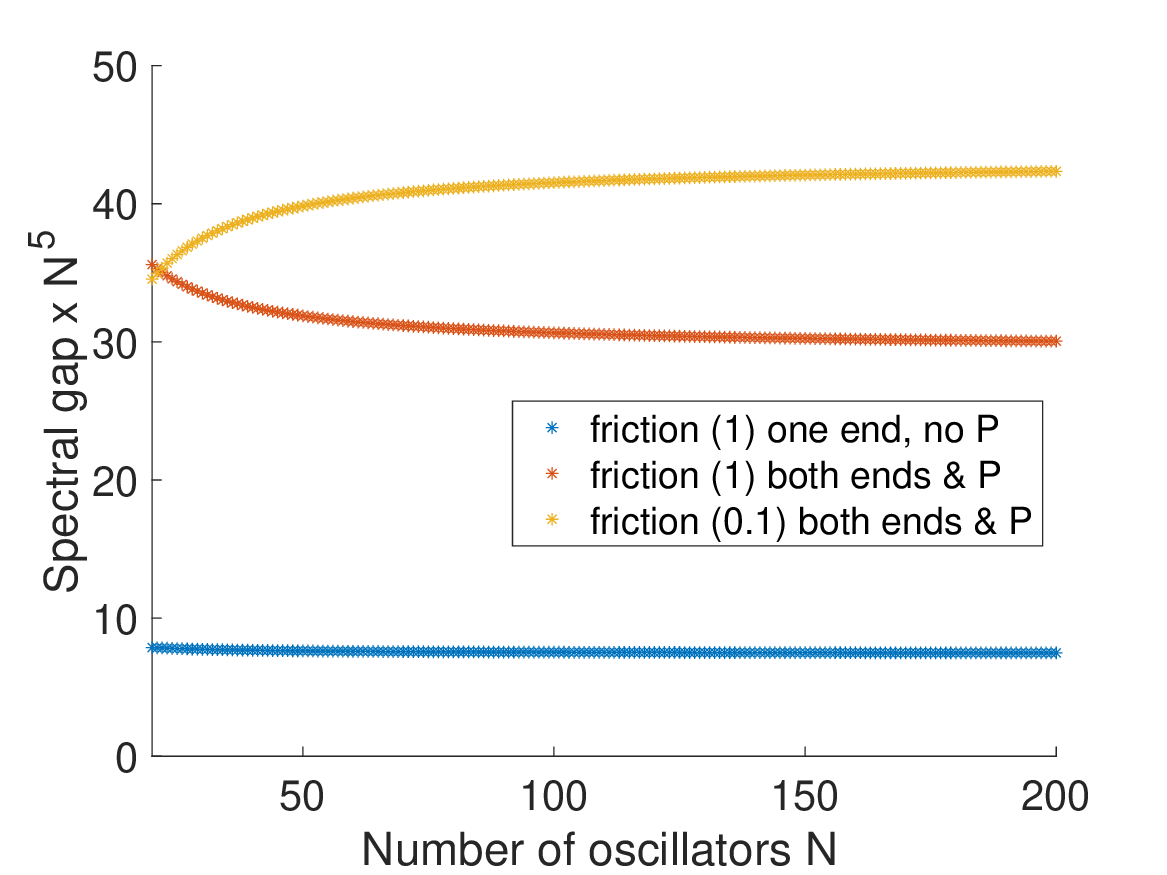}\\
\includegraphics[width=7cm]{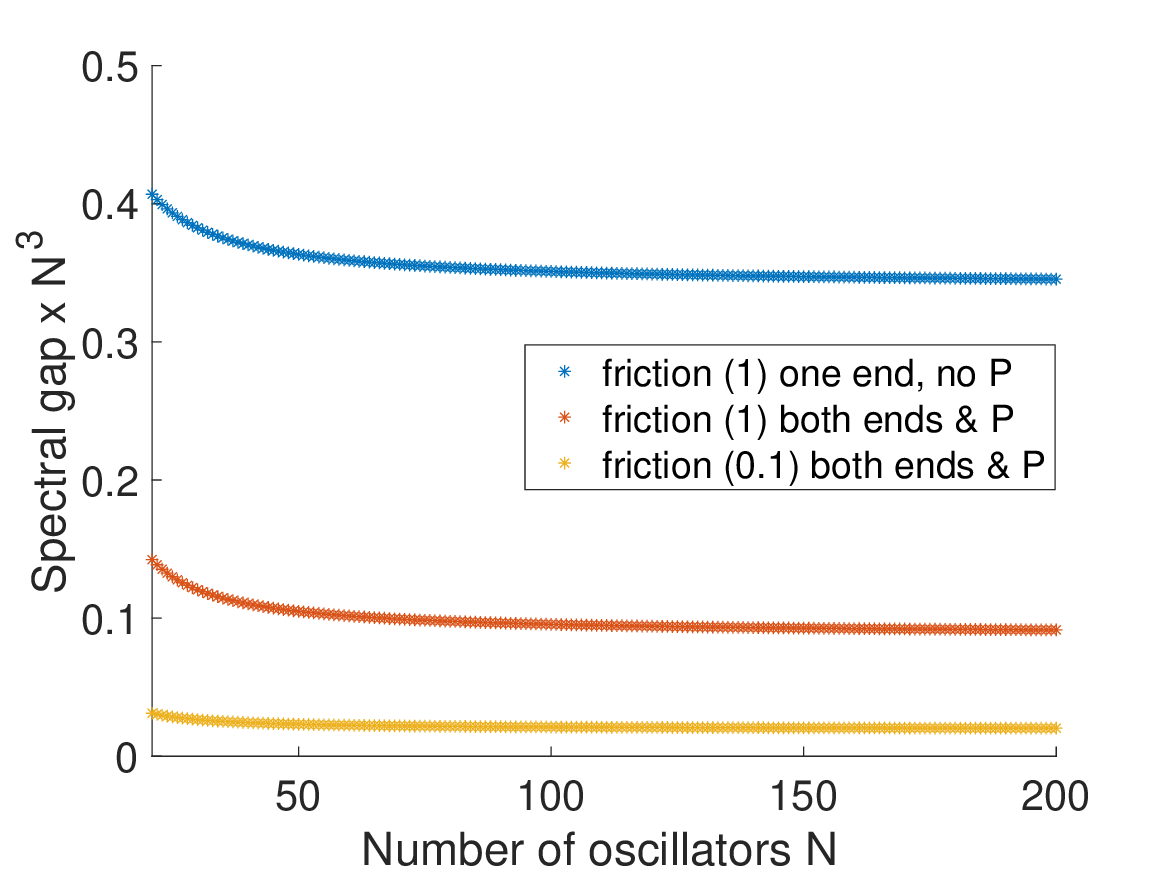}
\includegraphics[width=7cm]{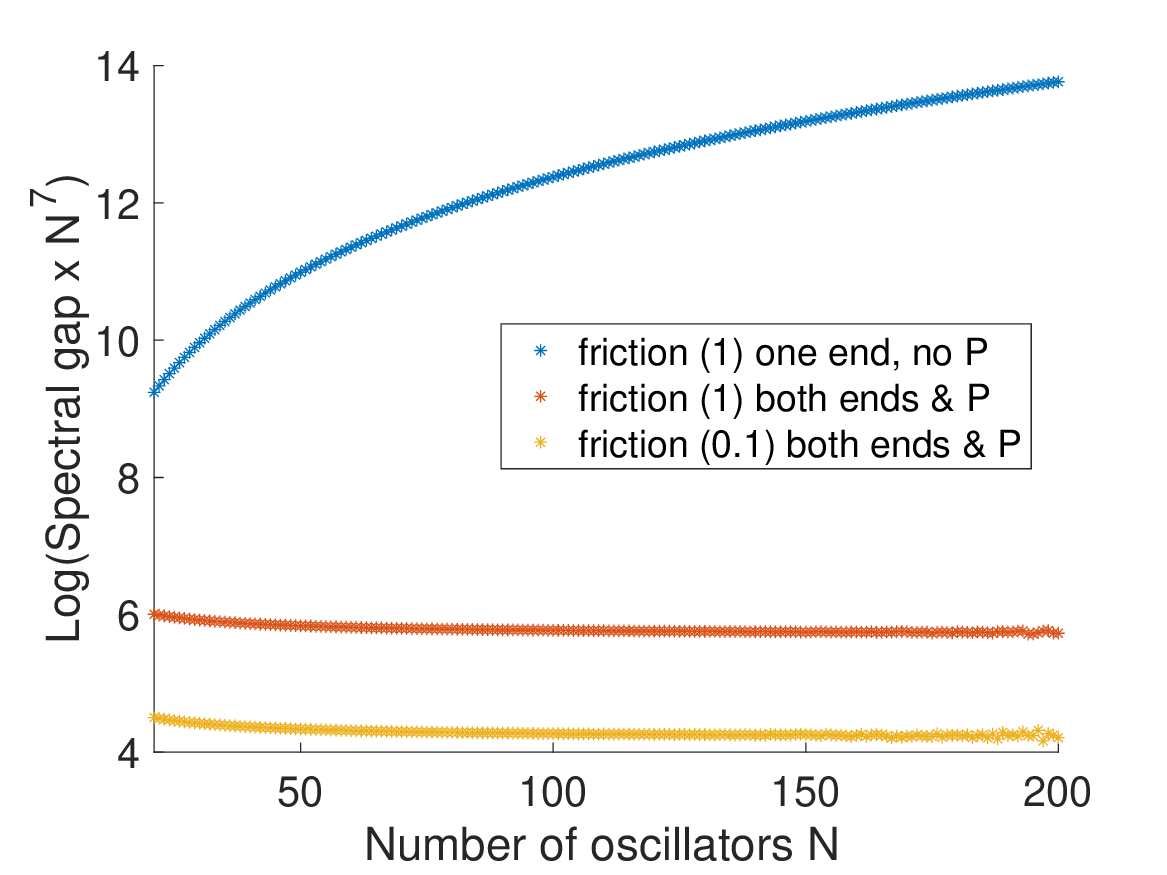}
\caption{Spectral gap (Dirichlet bc(top), Neumann bc(below)) for $\omega = 0.2$ (left) and Spectral gap for $\omega=1/4$ (right). We see the scaling of the spectral gap is exclusively determined by having friction $\gamma $ = value in () at a single end of the chain. The presence of $P$ is irrelevant in the Dirichlet case and essential in the Neumann case.}
\end{figure}

For our subsequent analysis, we also recall that the $j$-th eigenvalue of the discrete Dirichlet Laplacian $\mathscr V_1$ and the corresponding eigenvectors are 
\begin{equation}
\label{eq:eigensys}
 \lambda_j= 4 \sin^2 \left( \frac{\pi j }{2(N+1)} \right) \text{ with } v_{j}(i) = \sqrt{\frac{2}{N+1}} \sin \left(  \frac{ij \pi}{N+1} \right) \quad \text{ where }j \in [N]. 
 \end{equation}

\subsection{Criteria for lack of hypoellipticity under Dirichlet and Neumann boundary conditions} \label{subsec:crit for hypoell}
For both types of boundary conditions, depending on the strength of the next-to nearest neighbour interaction $\omega$, the generator of the chain of oscillators may not be hypoelliptic anymore. As we will see, these values of $\omega$ are different depending on the boundary conditions.

\subsubsection{Dirichlet boundary conditions} We observe that under Dirichlet boundaries there exists a necessary and sufficient condition in order to ensure that the operator is hypoelliptic.

 \begin{prop} 
Let $\eta \geq 0$ and let us assume $T_{\omega}$ with Dirichlet boundary conditions as above such that all eigenvalues of $T_{\omega}$ are simple.  
Then $\mathcal{L}$ is always hypoelliptic, regardless of $\omega$ and there exists a spectral gap. 
\end{prop}

\begin{proof} Taking the masses equal to $1$ for simplicity, we may consider instead of $M_{[N]}$ the similar matrix $\Omega_{[N]}:= \begin{pmatrix} \Gamma &- B_{[N]}^{1/2} \\ B_{[N]}^{1/2} & 0  
\end{pmatrix}$ with $B_{[N]}^{1/2} = \sqrt{T_\omega + P}$. The similarity is easily verified by noticing that 
\begin{equation}
\label{eq:similar}
 M_{[N]} = \operatorname{diag}(I,B_{[N]}^{1/2}) \Omega_{[N]}\operatorname{diag}(I,B_{[N]}^{-1/2}).
 \end{equation}

Now we consider $\lambda$ eigenvalue to $\Omega_{[N]}$ with  eigenvector $v=(u,i u)$  being of the form $u=(0,u_2,\dots,u_{N-1},0)$. Then $u$ is an eigenvector to $ (T_\omega + P)$ with eigenvalue $-\lambda^2$. If such an eigenvector exists then $Q_t v:= \int_0^t e^{-s\Omega_{[N]}} \operatorname{diag}(\Gamma \vartheta,0) e^{-s\Omega_{[N]}^*} v ds = \int_0^t e^{-s\Omega_{[N]}} \operatorname{diag}(\Gamma \vartheta,0) e^{-s\lambda} v ds = \int_0^t e^{-s\Omega_{[N]}} \operatorname{diag}(\Gamma \vartheta,0) v e^{-s\lambda} =0$, which implies that $\mathcal{L}$ is not hypoelliptic as $\ker(Q_t) \neq \{0\}.$ Indeed this makes the hypoellipticity condition on the invertibility of the covariant matrix $Q_t$, cf subsect. \ref{subsect: hypoell_OU} to fail. But this can not happen as this implies $ (T_\omega + P + \lambda^2)u=0$ or, since $P$ only acts on the first and the last entry, that $u$ with first and last entry equal to $0$ is an eigenvector of $T_{\omega}$ and thus of the Dirichlet Laplacian, due to the simplicity condition. Such an eigenvector to $\Delta_{[N]}^D$ does not exist and therefore $\mathcal{L}$ is hypoelliptic. 
\end{proof}
In fact, the above simplicity condition is necessary and sufficient as the following Proposition shows. 
 \begin{prop}
 Let $\eta\geq 0$ and $T_{\omega}$ with Dirichlet boundary conditions as above.  For an explicit dense and countable set $\Omega \subset (-\infty, -1/4) \cup (1/4,\infty)$, there exists a number of oscillators $N(\omega)$, such that $T_{\omega}$ has degenerate eigenvalues. In particular, the generator of this chain of finite size $N(\omega)$ is hypoelliptic if and only if $\vert \omega \vert \le 1/4.$
 \end{prop}
 \begin{proof} 
For $i,j \in [N]$, the difference of eigenvalues of $T_{\omega}$ is given by $$\kappa :=   (1+4\omega)\lambda_j- \omega \lambda_j^2 + \eta- ((1+4\omega)\lambda_i- \omega \lambda_i^2+ \eta),$$ where $\lambda_i$ are the eigenvalues of the Dirichlet-Laplacian. Thus we have $\kappa=0$ iff $$(1+4\omega)\lambda_j- \omega \lambda_j^2 = (1+4\omega)\lambda_i- \omega \lambda_i^2.$$ This is equivalent to looking for $\omega$ so that 
\begin{align*} 
(\lambda_j - \lambda_i) \big[ (1+4\omega) - \omega(\lambda_j + \lambda_i) \big] =0
\end{align*}
which means 
\begin{equation}
\label{eq:condition}
\omega= 1/(\lambda_i+\lambda_j-4).
\end{equation}
Once this condition is met, the eigenspace of $T_{\omega}$ with eigenvalue $\mu_j$ is at least two-dimensional. Hence, we may take two Dirichlet eigenvectors $v_i,v_j$ associated with eigenvalues $\lambda_i,\lambda_j$ respectively. Then we can define a new eigenvector
\[ w = v_j(1)v_i-v_i(1)v_j \text{ such that } w(1)=0.\]
Since all Dirichlet eigenvectors are even, we also have $w(N)=0.$ Thus, we have exhibited an eigenvector to $T_{\omega}$ vanishing at both terminal ends. This implies that the generator cannot be hypoelliptic.

We would finally like to point out that the above condition on $\omega$ is a dense set. 
Indeed, define 
$$f(x,y) = \frac{1}{4 \left[\sin^2(\tfrac{\pi}{2} x)+\sin^2(\tfrac{\pi}{2} y)\right]-4}.$$
Then $f: \RR^2 \to \mathbb R \setminus [-1/4,1/4]$ is onto, showing that necessarily $\vert \omega \vert >1/4.$ Since $\mathbb Q^2$ is dense in $\mathbb R^2$ this implies that $f(\mathbb Q^2)$ is dense in $\mathbb R \setminus [-1/4,1/4]$, but $f(\mathbb Q^2)$ precisely corresponds to the condition \eqref{eq:condition}.
\end{proof}

\subsubsection{Neumann boundary conditions} We expect this set of $\omega$'s to be dense in terms of $N$ as well, similarly to the Dirichlet case, as indicated by simulations, see Fig. \ref{fig:neumann}. In this case is, it is however harder to obtain such an explicit description, as the perturbation $P$ affects not just the terminal particles as in the Dirichlet case.

In the following examples, we exhibit some explicit computations for a fixed number of particles $N$, to illustrate the non-existence of spectral gaps. In the examples we either use H\"{o}rmander's hypoellipticity condition (by commutators) or the equivalence of hypoellipticity for an OU operator to the condition that 
\begin{equation}
\label{eq:Qt}
 Q_t := \int_0^t e^{-sM_{[N]}} \operatorname{diag}(\Gamma \vartheta, 0) e^{-sM_{[N]}^*} \ ds 
 \end{equation}
 satisfies $\operatorname{det}(Q_t)>0$ for all $t>0$, as discussed in subsection \ref{subsect: hypoell_OU}. The latter condition does not hold as soon as there is an eigenvector to $\mathscr{V}_N$ with the first and last entry being $0$. 
 
Setting for simplicity $m=1$ and using \eqref{eq:similar}, we can write \eqref{eq:Qt} as
 \[ Q_t = \int_0^t e^{-s\Omega_{[N]}} \operatorname{diag}(\Gamma \vartheta, 0) e^{-s\Omega_{[N]}^*} \ ds. \]
 \begin{ex} Take $\eta=0$ and $F=\{1,N\}$. We have $\ker(\mathscr V_N)= \ker(\mathscr V_1 + \omega \mathscr V_2) = \operatorname{span}\{(1,\cdots,1)\}.$
 
 In addition, we observe that e.g. for $\omega=\frac{1}{2}$ we obtain eigenstates that vanish at the boundary, indeed e.g. for $N=5$ oscillators, the vector $v=(0, 1, -2, 1, 0)$ is an eigenvector to $\mathscr V_N$ with eigenvalue $4.$  
 
We can then construct an eigenvector $u = (iv,v)$ to both $\Omega_{[N]}$ and $\Omega_{[N]}^*$ with eigenvalue $\lambda=\pm 2i.$ 

 Since in our case $\operatorname{diag}(\Gamma \vartheta ,0) u=0$, we have that 
\[\begin{split} Q_t u &= 
\int_0^t e^{-s \Omega_{[N]}} \operatorname{diag}(\Gamma \vartheta,0) e^{-s\Omega_{[N]}^*}u \ ds  =
\int_0^t e^{-s\Omega_{[N]}} \operatorname{diag}(\Gamma \vartheta,0)  e^{2is}u \ ds \\
&= \int_0^t e^{-s\Omega_{[N]}} \operatorname{diag}(\Gamma \vartheta,0) u e^{2is}  ds =0.  \end{split}\]
Hence, $Q_t$ has a non-trivial nullspace and therefore $\det(Q_t)=0$. Hence, the generator is not hypoelliptic.
One can also exhibit a similar example for $\omega=\frac{1}{3}$ and $N=7$ with eigenvector $v =(0, 1, -3, 4, -3, 1, 0)$ and eigenvalue $\lambda=4.$
 \end{ex}
Instead of trying to find eigenvectors matrices with vanishing end-points, one can also analyze hypoellipticity directly from the structure of the operator.
 \begin{ex}[Lack of hypoellipticity for $N=5$ via H\"{o}rmander's commutator condition]
Let the pinning coefficient $\eta \geq 0$.
 The generator \eqref{eq:L} can be equivalently expressed in H\"{o}rmander's form as 
\[ \mathcal{L} = X_0 +X_1^2+X_N^2\]
which for $\mathbf m=I$ read
$$X_i = \sqrt{\tau_i} \partial_{p_i} \text{ with } i \in \{1,N\}$$
and 
\begin{equation}
\begin{split}
X_0 =& \langle p,\nabla_q\rangle - \gamma_1 p_1 \partial_{p_1}- \gamma_N p_N \partial_{p_N} - \sum_{i=1}^{N-1} (q_i-q_{i+1})( \partial_{p_i} -\partial_{p_{i+1}}) \\
&-\omega \sum_{i=1}^{N-2} (q_i-q_{i+2})( \partial_{p_i} -\partial_{p_{i+2}}) + \eta \sum_{i=1}^N q_i \partial_{p_i}.
\end{split}
\end{equation}
Let $\mathcal A$ be the Lie algebra generated by iterated commutators of vectors fields involving $X_0,X_1,X_N.$
Indeed, we find 
\[ [\partial_{p_1},X_0] = \partial_{q_1}-\gamma_1 \partial_{p_1} \text{ and }[\partial_{p_N},X_0] = \partial_{q_N}-\gamma_N \partial_{p_N}. \]
Specializing to $N=5$, we have $\partial q_1,\partial p_1,\partial q_5,\partial p_5 \in \mathcal A.$
\[ [\partial_{q_1},X_0] = -(1+\omega+ \eta) \partial_{p_1} +\partial_{p_2}+\omega \partial_{p_3}.\]
Hence, $\partial_{p_2}+\omega \partial_{p_3},\partial_{q_2}+\omega \partial_{q_3} \in \mathcal A.$

\[ \begin{split} [\partial_{q_2}+\omega \partial_{q_3},X_0] &= \partial_{p_1}-2 \partial_{p_2}+\partial_{p_3}+ \omega( -\partial_{p_2}+\partial_{p_4}) \\
&+ \omega( \partial_{p_2}-2\partial_{p_3}+\partial_{p_4})+\omega^2(\partial_{p_1}-2\partial_{p_3}+\partial_{p_5}) + \eta(\partial_{p_2} + \omega \partial_{p_3}) \\
&= (1+\omega^2)\partial_{p_1}-2(\partial_{p_2}+\omega \partial_{q_3})+\partial_{p_3}-2\omega^2 \partial_{p_3}+2\omega \partial_{p_4}+ \omega^2 \partial_{p_5}  + \eta(\partial_{p_2} + \omega \partial_{p_3})\end{split}\]
We conclude that $(1-2\omega^2)\partial_{p_3} +2\omega \partial_{p_4},(1-2\omega^2)\partial_{q_3} +2\omega \partial_{q_4}  \in \mathcal A.$ We continue
\[ \begin{split} [(2\omega^2-1) \partial_{q_3} - 2\omega \partial_{q_4},X_0] &=(2\omega^2-1)((\partial_{p_2}-2 \partial_{p_3}+\partial_{p_4}) +\omega (\partial_{p_1}-2 \partial_{p_3}+\partial_{p_5})) \\
&\quad -2 \omega( \partial_{p_3}-2\partial_{p_4}+\partial_{p_5})+\omega^2(\partial_{p_2}-\partial_{p_4}) + \eta((2\omega^2-1) \partial_{p_3} - 2\omega \partial_{p_4}). \end{split}\]
We conclude that 
\[ \begin{split}
&(-1+3\omega^2) \partial_{p_2}+ (2-4(\omega^2+\omega^3)) \partial_{p_3}+ (-1+4\omega + \omega^2) \partial_{p_4} \text{ and } \\
&(-1+3\omega^2) \partial_{q_2}+ (2-4(\omega^2+\omega^3)) \partial_{q_3}+ (-1+4\omega + \omega^2) \partial_{q_4}  \in \mathcal A. \end{split}
\]
It suffices now to check when the linear combinations of vectors $\partial_{p_2},\partial_{p_3},\partial_{p_4}$ span a three-dimensional space, i.e. when
$$ \begin{pmatrix}
1 & 0 & -1 +3\omega^2 \\ \omega & 1 -2 \omega^2 & 2 -4(\omega^2 +\omega^3) \\ 0 & 2 \omega & -1 +4\omega+\omega^2
\end{pmatrix}$$
is non-singular.
Indeed, one readily computes the determinant $12\omega^4+\omega^2-1$ which is zero for $\omega= \pm \frac{1}{2}.$
 \end{ex} 
 
In the next theorem we show that for Neumann boundary conditions, if the strength of the next-to-nearest-neighbour interaction satisfies $| \omega | < 1/4,$  the operator is hypoelliptic.  
\begin{theo}
Let $\omega \in (-1/4,1/4)$ then the Fokker-Planck operator of the chain of oscillators under Neumann boundary conditions, with pinning potential $\eta\geq 0$ and friction at both terminal ends, is hypoelliptic. In particular, the operator exhibits a non-zero spectral gap. 
\end{theo}
\begin{proof}
Using \eqref{eq:similar}, it suffices to argue that $$\Omega_{[N]}:=\begin{pmatrix} \Gamma & -m^{-1/2}\sqrt{T_\omega+P} \\ m^{1/2} \sqrt{T_\omega +P} &0 \end{pmatrix}$$ does not have a non-trivial nullspace. However, let $u\neq 0$ satisfy $\Omega_{[N]} u =0$, then 
\[ \Re\langle \Omega_{[N]} u,u \rangle =\gamma_1 \vert u_1 \vert^2 + \gamma_N \vert u_N \vert^2.\]
This quantity is strictly positive unless both $u_1=u_N=0.$ Thus, it suffices to exclude the existence of an eigenfunction to $T_\omega +P$ that vanishes at both terminal ends.

Let $\sigma= \begin{pmatrix} 0 & \cdots &0 & 1\\ \vdots & \vdots& 1 &0 \\ \vdots & \Ddots & 0 & 0\\1 & 0 &\cdots & 0 \end{pmatrix}.$
The Hamiltonian $H_{\omega} = T_\omega+P$ then satisfies $[H_{\omega}, \sigma]=0.$
The matrix $\sigma$ has two invariant subspaces that are also invariant subspaces of $H_{\omega}$
\[ \mathbb C^n = (X_{\text{sym}}= \{u \in \mathbb C^{N};  \forall i: \ u_i = u_{N-i} \} ) \oplus  (X_{\text{asym}}= \{u \in \mathbb C^{N};  \forall i: \ u_i = -u_{N-i}  \} ).\]
We shall show that if the operator is not hypoelliptic, then $\vert \omega \vert \ge 1/4.$
So if the Fokker-Planck operator is not hypoelliptic, then there exists an eigenvector $u$ to the Hamilton function $H_{\omega} u = \lambda u$ with $\lambda \in \mathbb R$ such that $u_1=u_N=0.$ We can assume this eigenvector to be either symmetric or anti-symmetric.
We then define a linear injection $\Phi: \mathbb C^N \to \CC[X]$ with $(\Phi(u))(X)=\sum_{i=1}^{N} u_i X^{N+1-i}.$ Thus, for an eigenvector $u=(0,u_2,..,u_{N-1},0)\in \CC^N$ the eigenvalue identity $\Phi((H_{\omega}-\lambda )u)(X)=0$ is equivalent to
\begin{equation}
\begin{split}
\label{eq:EV_problem} 0=\Phi((H_{\omega}-\lambda )u)(X)&=(\Phi(u))(X)(-\omega(X^{-2}+X^2)-(X^{-1}+X)+2(1+\omega)+\eta-\lambda)\\
&\quad - u_2\omega (X-1)^2(X^{N-2}\pm 1/X)\end{split}
\end{equation}
where $\pm$ corresponds to symmetric/anti-symmetric eigenvectors, respectively. This identity holds for all $X \in \CC.$
Due to $X \mapsto 1/X$ symmetry of 
\begin{equation}
\label{eq:X} 
X \mapsto (-\omega(X^{-2}+X^2)-(X^{-1}+X)+2(1+\omega)+\eta-\lambda)
\end{equation} appearing in \eqref{eq:EV_problem}, roots of this rational function are of the form $r_1 e^{ 2\pi i \theta_1}, e^{ -2\pi i \theta_1}/r_1$ and $r_2 e^{ 2\pi i \theta_2},e^{-2\pi i \theta_2}/r_2 $ for some $\theta_i \in \RR/ \ZZ$ and $r_i>0$. Since the identity \eqref{eq:EV_problem} holds for all $X$, we can specialize to one of the roots of \eqref{eq:X}, but then also for $X$ being one of the roots $$u_2\omega (X-1)^2(X^{N-2}\pm 1/X)=0.$$ This implies that $r_1=r_2=1.$ Applying Vieta's formula for the sum of roots to the polynomial 
\[X^4 +\frac{1}{\omega}(X^3+X)-\Big(\frac{2(1+\omega)-\lambda}{\omega}\Big)X^2+1=0 \]
shows that
\[ \vert 1/\omega \vert = \vert  e^{ 2\pi i \theta_1}+e^{-2\pi i \theta_1}+e^{ 2\pi i \theta_2}+e^{-2\pi i \theta_2} \vert  \le 4,\]
which gives the desired bound on suitable $\omega$.
\end{proof}

The discussion above implies the following proposition.

\begin{prop} Assuming Neumann boundary conditions for $T_\omega$, then the generator $\mathcal L$ is hypoelliptic if and onlf if there exists an eigenvector $u=(0,u_2,\cdots,u_{N-1},0)$ so that $T_\omega u + \omega P u =\lambda u$, for some $\lambda>0$. 
\end{prop}

\begin{figure}
\includegraphics[width=7.5cm]{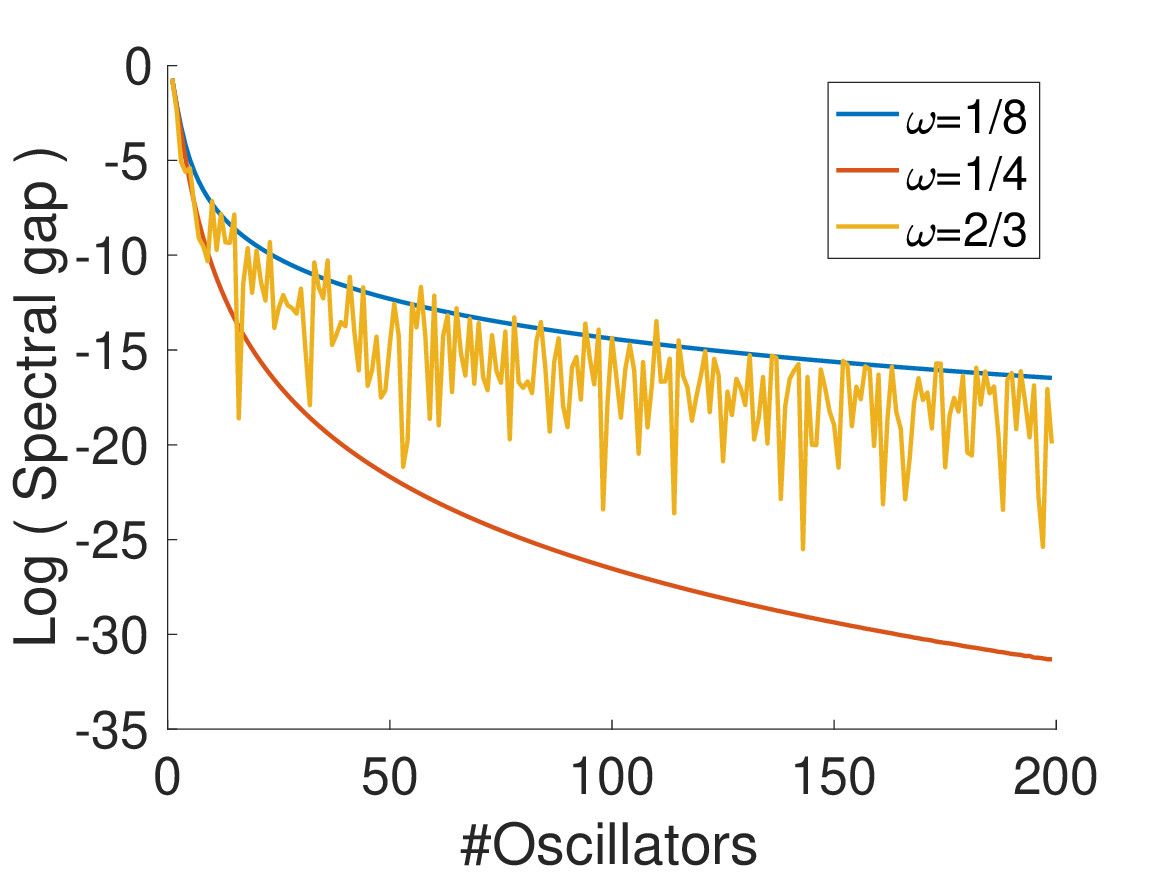}
\includegraphics[width=7.5cm]{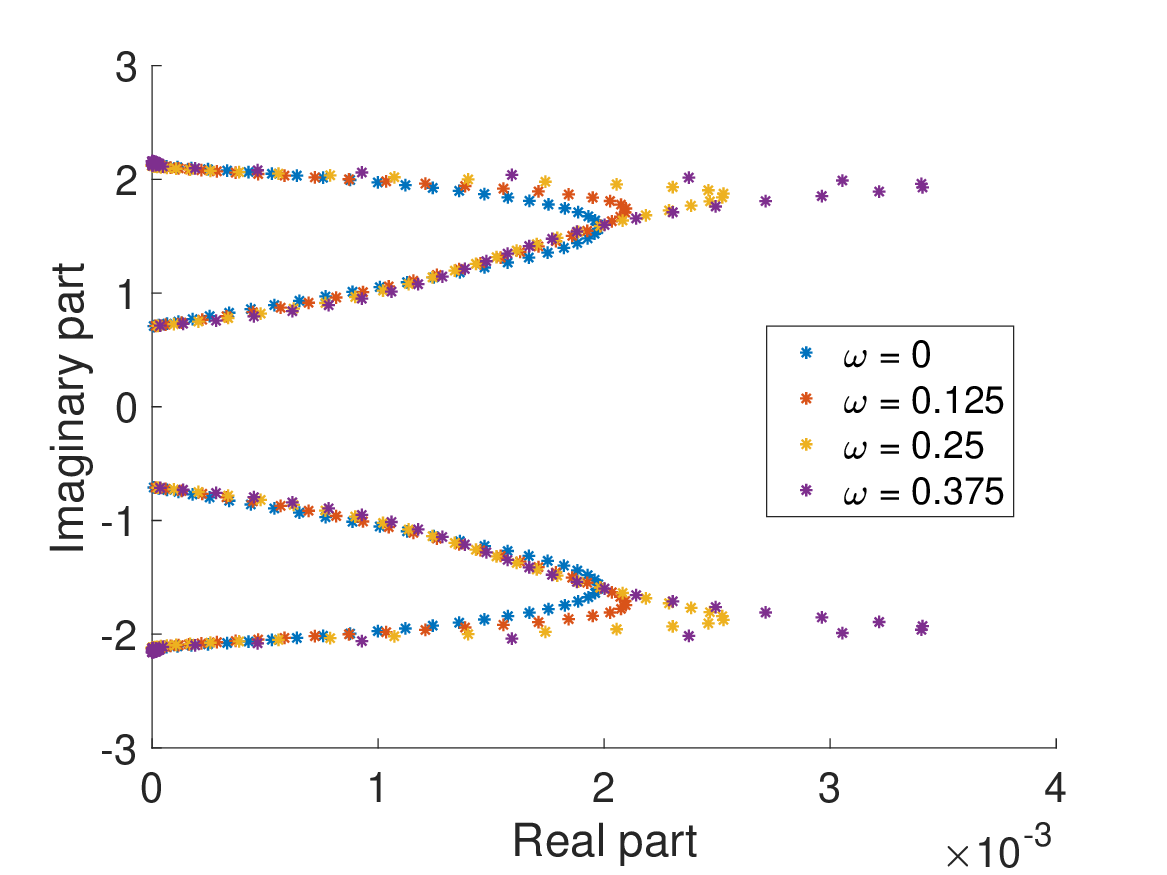}
\includegraphics[width=7.5cm]{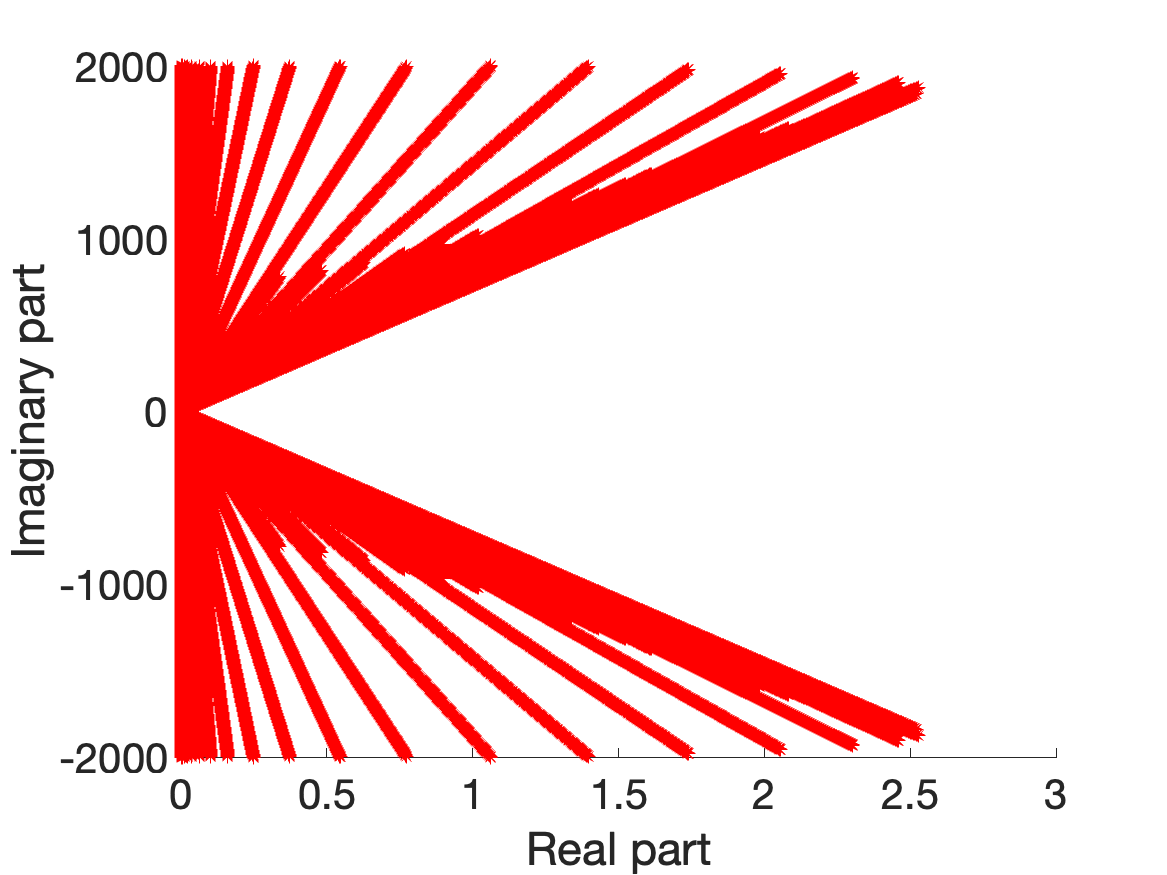}
\includegraphics[width=7.5cm]{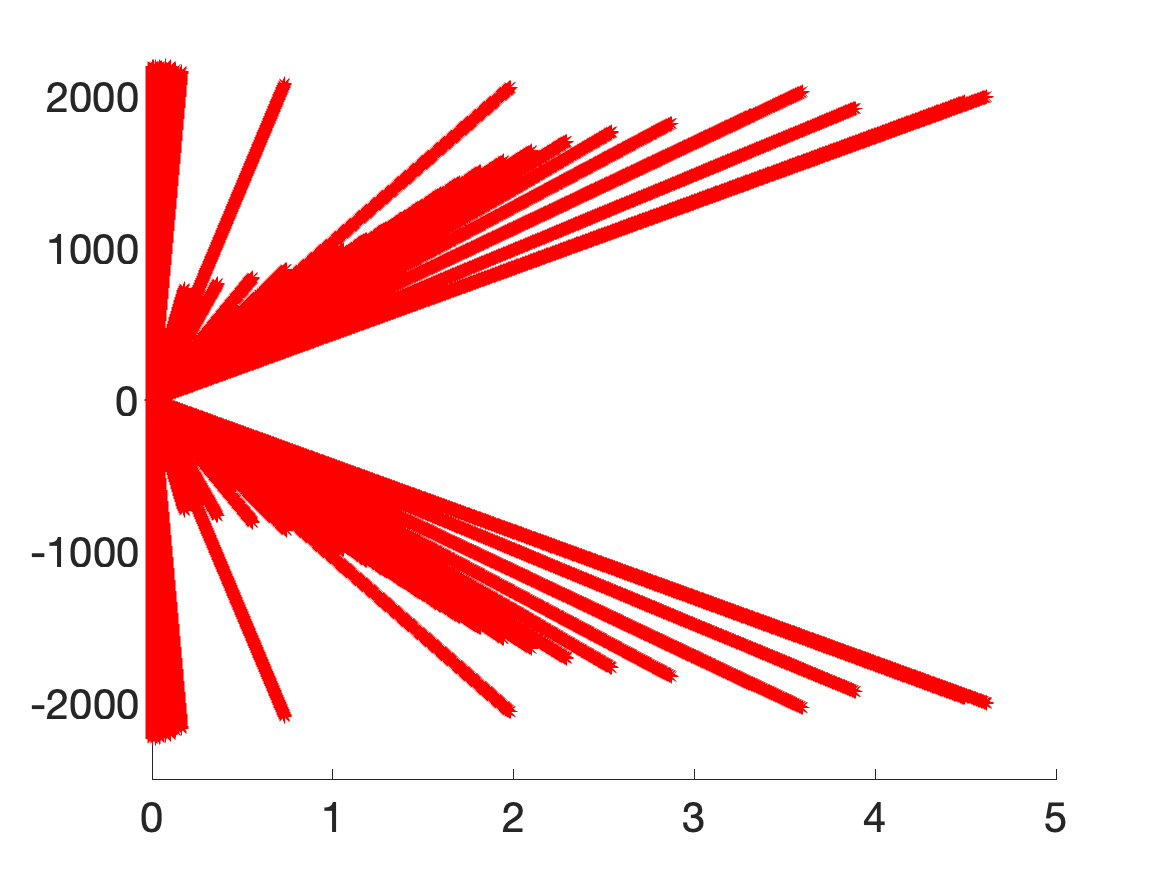}
\caption{Log-log plot of the spectral gap for the one-dimensional chain of oscillators under Neumann boundary conditions for three different regimes of $\omega$ (top left) and full spectrum of $M_{[N]}$ for $N=50$ (top right) for $\gamma=1/4.$
Full spectrum of Fokker-Planck operator for critical $\omega = 1/4$ and supercritical $\omega = 1/2 $ (bottom left and right).} 
\label{fig:neumann}
\end{figure}

 \subsection{Spectral gap under Dirichlet boundary conditions}
We shall now study finer estimates on the decay of the spectral gap with next-to nearest neighbour interactions and Dirichlet boundary conditions.
For the following results, we simplify the computations by assuming that
\begin{Assumption}
\label{ass:assump}
For the interaction, we assume Dirichlet boundary conditions such that for the next-to nearest neighbour interaction
\begin{enumerate}[label=(\roman*)]
\item The perturbation matrix $P$ is of rank $1$, i.e.  $\operatorname{diag}(-2\omega, 0, \dots,0).$
\item The friction at constant $\gamma$ is imposed only at one end of the chain, i.e. $F=\{1\}.$
\end{enumerate}
In addition, we assume all masses are normalized to one.
\end{Assumption}

The matrix $M_{[N]}$ is then as follows
 \begin{equation}
 \label{eq:matrix M rank1 perturb}
 M_{[N]} = \begin{pmatrix} 
 0& -I \\ T_\omega & 0 
 \end{pmatrix} + \begin{pmatrix} 
 \Gamma & 0 \\ P & 0 
 \end{pmatrix} \text{ with }A:= \begin{pmatrix} 
 0& -I \\ T_\omega & 0 
 \end{pmatrix}
\end{equation} 
where $ \Gamma= \operatorname{diag}(\gamma, 0 , \dots, 0)$ with $\gamma$ the friction constant and $ P = \operatorname{diag}(-2\omega, , 0 , \dots, 0)$. 
We observe that by assumptions (i) and (ii), the matrix $M_{[N]}$ is only a rank one perturbation of the matrix $A,$ whereas in case of friction at both ends and $P=-2\omega(e_1+e_n)$ or $F=\{1,N\}$ the perturbation is of rank two. 

Focussing on the case of a rank $1$-perturbations $ \begin{pmatrix} \Gamma & 0 \\ P & 0 \end{pmatrix}$, as outlined in Assumption \ref{ass:assump}, corresponds to a rank $1$-perturbation and we proceed by expanding the determinant of the eigenvalue problem that we are interested in. We have 
\[ \operatorname{det}(M_{[N]}+i\mu) =\operatorname{det}\Big(A+i\mu + \begin{pmatrix} \Gamma & 0 \\ P & 0 \end{pmatrix}\Big) =  \operatorname{det}\big(A+i\mu\big) \operatorname{det}\Big( I-i(\mu-iA)^{-1} \begin{pmatrix} \Gamma & 0 \\ P & 0 \end{pmatrix}\Big).\]
Since the perturbation is of rank one, we have
\begin{align} \label{eq: expanding_det}
\operatorname{det}\left( I- i( \mu-iA  )^{-1} \begin{pmatrix} 
\Gamma & 0 \\ P & 0
\end{pmatrix}
 \right) = 1-i \operatorname{Tr} \left( ( \mu-iA  )^{-1} \begin{pmatrix} 
\Gamma & 0 \\ P & 0
\end{pmatrix}   \right) =0
\end{align} 
where $A$ is the first matrix in the right-hand side of \eqref{eq:matrix M rank1 perturb}.  First we notice that we can easily find explicit expressions for the eigenvalues and eigenvectors to $A$. Indeed, let $S$ be the eigenbasis of the self-adjoint matrix $T_\omega$, then  
\begin{align} 
\begin{pmatrix} S^{-1} & 0 \\ 0& S^{-1}
\end{pmatrix} \begin{pmatrix} 0 & -I \\ T_\omega & 0
\end{pmatrix} \begin{pmatrix} S & 0 \\ 0& S \end{pmatrix} = \begin{pmatrix} 0 & -I \\ \operatorname{diag}(\nu_1, \cdots, \nu_N)& 0
\end{pmatrix}.
\end{align} 
%where $\nu_j= \sqrt{(1+4\omega)\lambda_j- \omega \lambda_j^2}$, $j \in [N]$.

This implies that the spectrum of $A$ is
 \begin{equation} \label{eq:spec A_0}
 \Spec(A)=\Bigg\{\pm   i \sqrt{(1+4\omega)\lambda_j- \omega \lambda_j^2}, j \in [N]\Bigg\}.
 \end{equation} 
The corresponding eigenvectors are $V_j= (v_j, -\mu_j^{\pm} v_j)^T$ where $v_j$ is the $j$-th eigenvectors of $T_\omega$, 
 $\mu_j^{\pm} = \pm i\sqrt{\nu_j}$, with $\nu_j :=(1+4\omega)\lambda_j -\omega \lambda_j^2$ where $\lambda_j$ are the eigenvalues of the Dirichlet discrete Laplacian.

\begin{figure} \label{fig:gap_forN_and_for_omega}
\includegraphics[width=7.5cm]{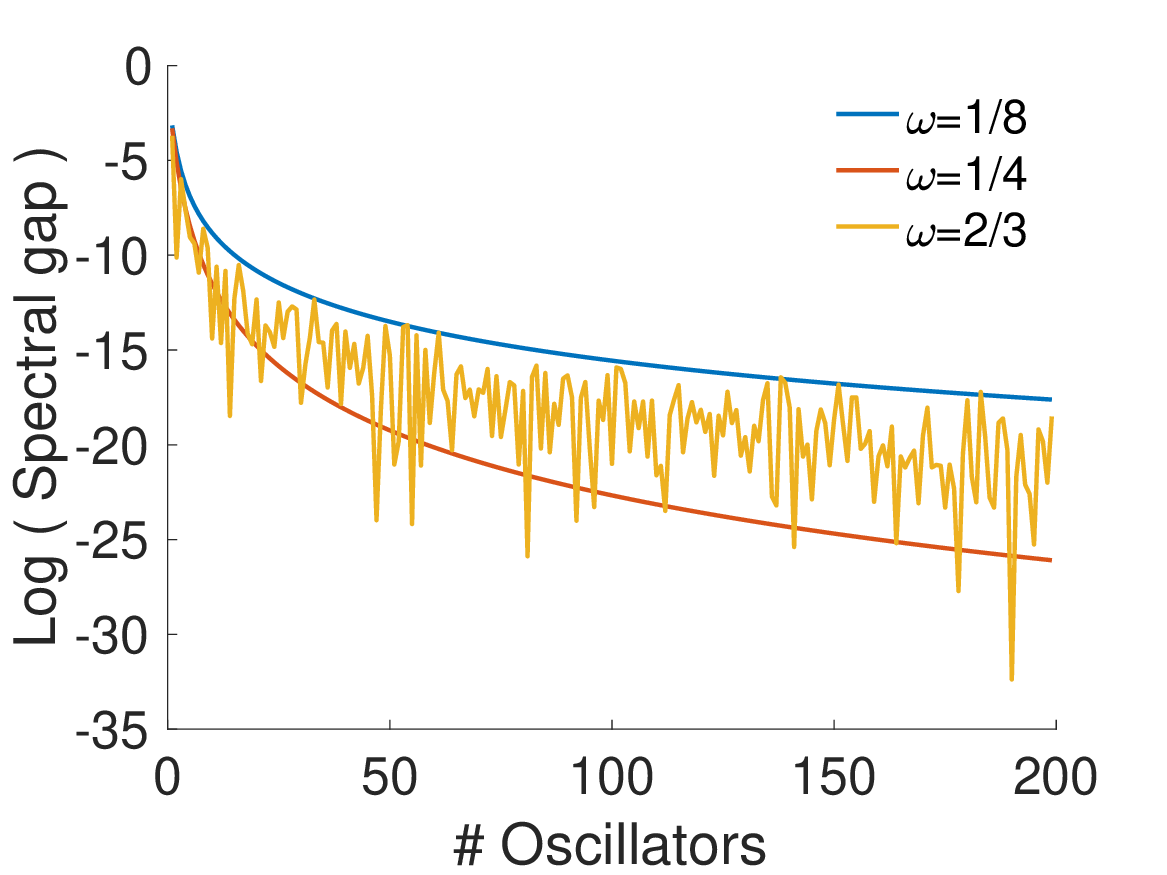}
\includegraphics[width=7.5cm]{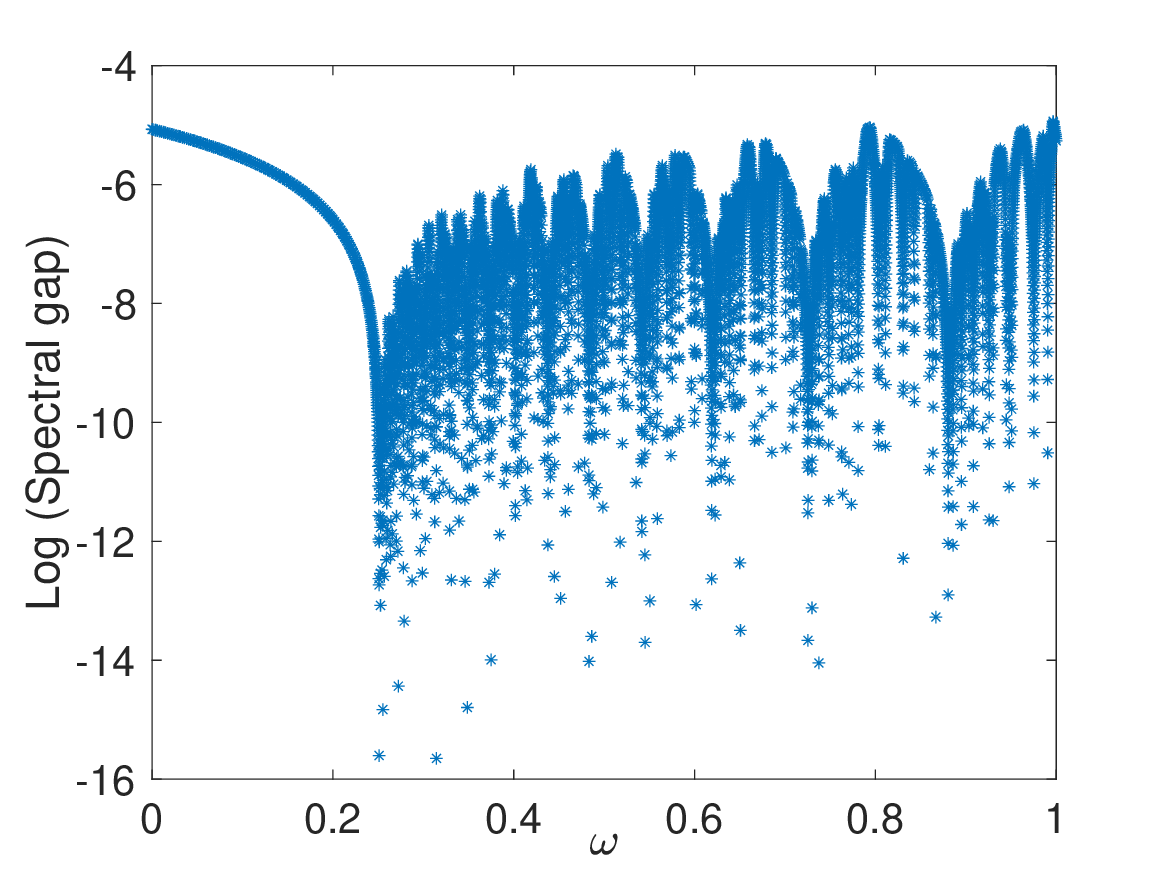}
\caption{Log-log plot of the spectral gap for the one-dimensional chain of oscillators under Dirichlet boundary conditions for three different regimes of $\omega$: below, above and at the value $\frac{1}{4}$. Log(spectral gap) for different $\omega$ and $N=50$ oscillators.} 
\label{fig:dirichlet}
\end{figure}

As indicated in Fig. \ref{fig:dirichlet}, for Dirichlet boundaries, the spectral gap changes behaviour for different values of $\omega$: For $\omega< \frac{1}{4}$ we expect it to decay exactly as $N^{-3}$, when $\omega=\frac{1}{4}$ as $N^{-5}$,  while for $\omega >\frac{1}{4}$ it exhibits an oscillatory behaviour in terms of $N$ in the sense that the spectral gap takes value very close to $0$, or $0$, for certain values of $N$. Note that under Neumann boundaries one can observe similar behaviour, see Fig. \ref{fig:neumann}. 

\subsection{Spectral gap in the critical case $\omega=1/4$}
In the following, we study, assuming; Dirichlet boundaries, the spectral gap for $\omega$'s less or equal to $\frac{1}{4}$. We start with the critical case $\omega=1/4$. 

Our proof of the following theorem only provides estimates on $|\mu(N) - \mu_N^{+}|$ of order $N^{-4}$. In particular, implies that $\Re(\mu(N))$ and thus also the spectral gap $g(N)$ decays at least like $N^{-4}.$ 

This result does not seem sharp, as numerical experiments suggest that $g(N)$ behaves like $N^{-5}$, see Fig. \ref{fig:gap_forN_and_for_omega}. Our principal estimate on $|\mu(N) - \mu_N^{+}|$ in the critical case $\omega=1/4$ does however seem to be optimal, since the decay of the real part is faster than the rate of convergence of the imaginary part to the eigenvalue $\mu_N^{+}$. This is illustrated in Fig. \ref{fig:where_to_localise}. 

\begin{figure} \label{fig:where_to_localise}
     \includegraphics[width=7cm]{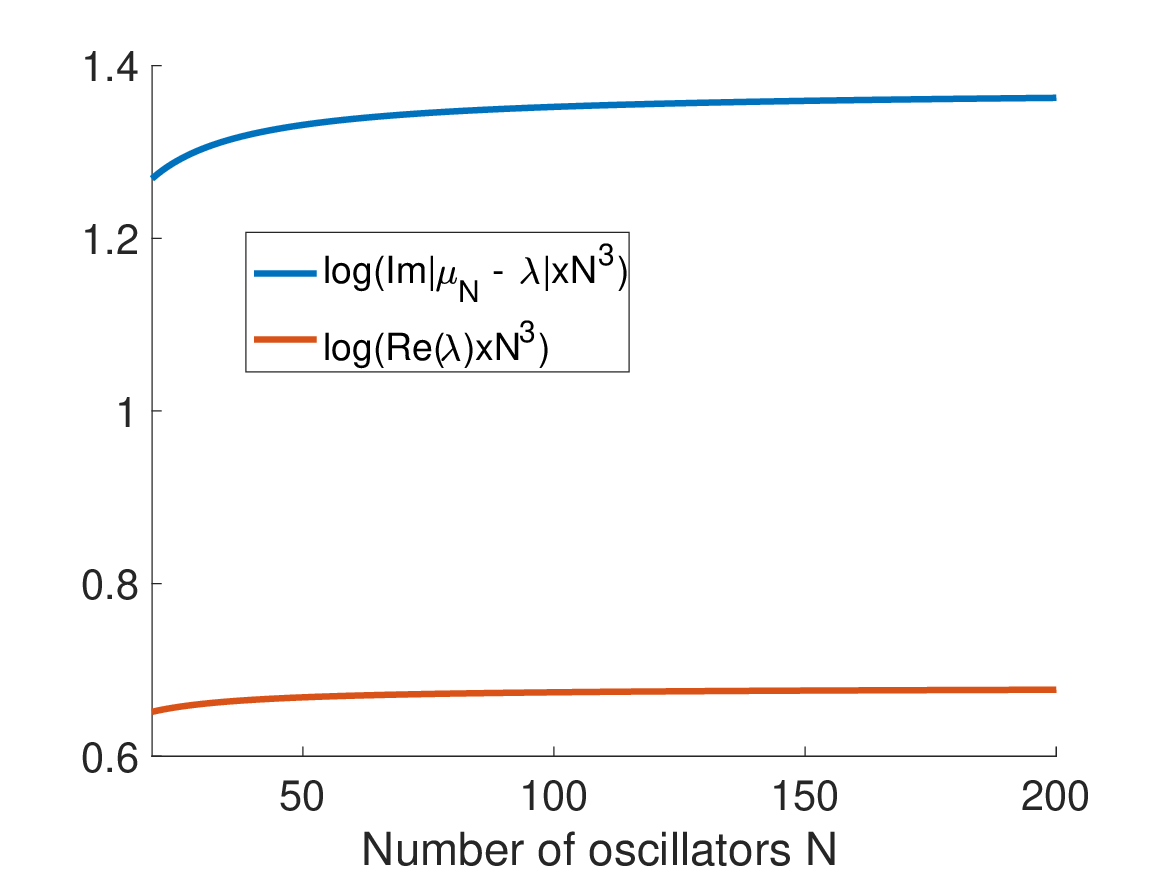}
    \includegraphics[width=7cm]{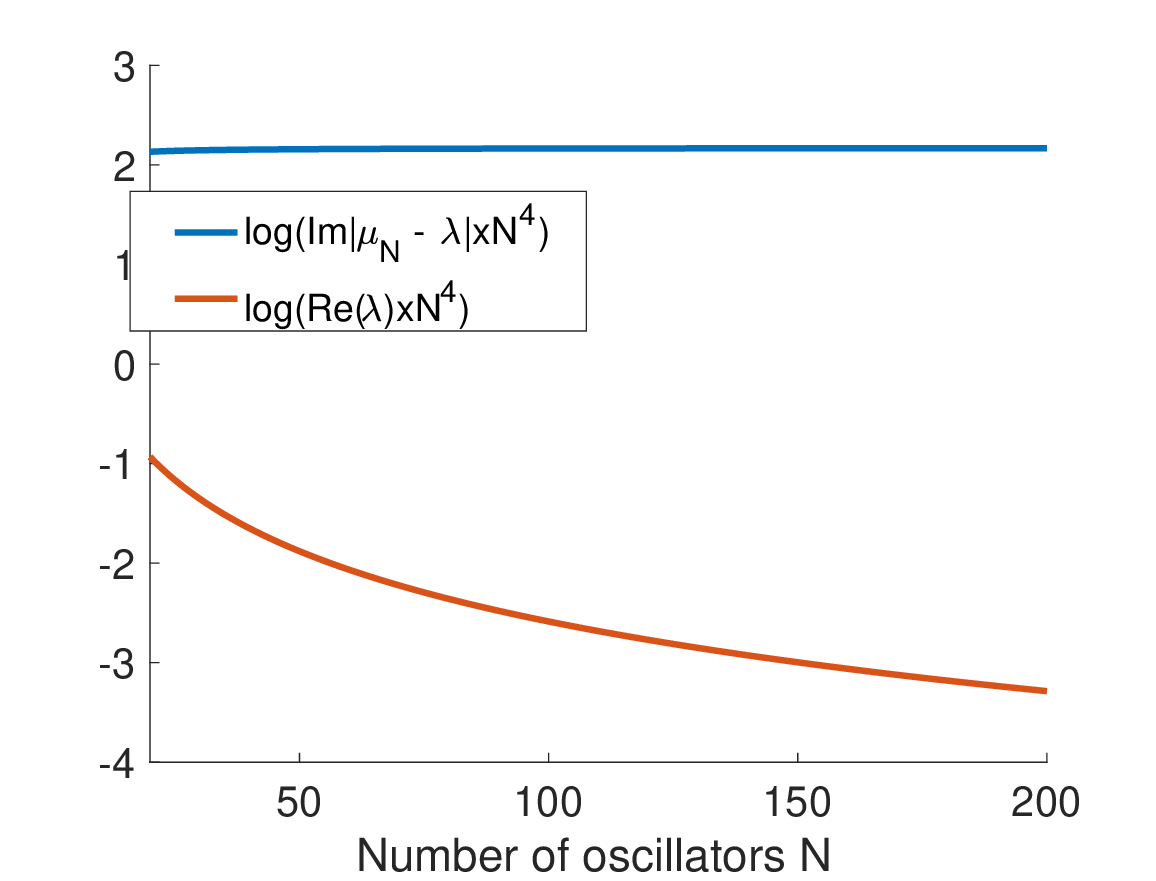}\\
\includegraphics[width=7cm]{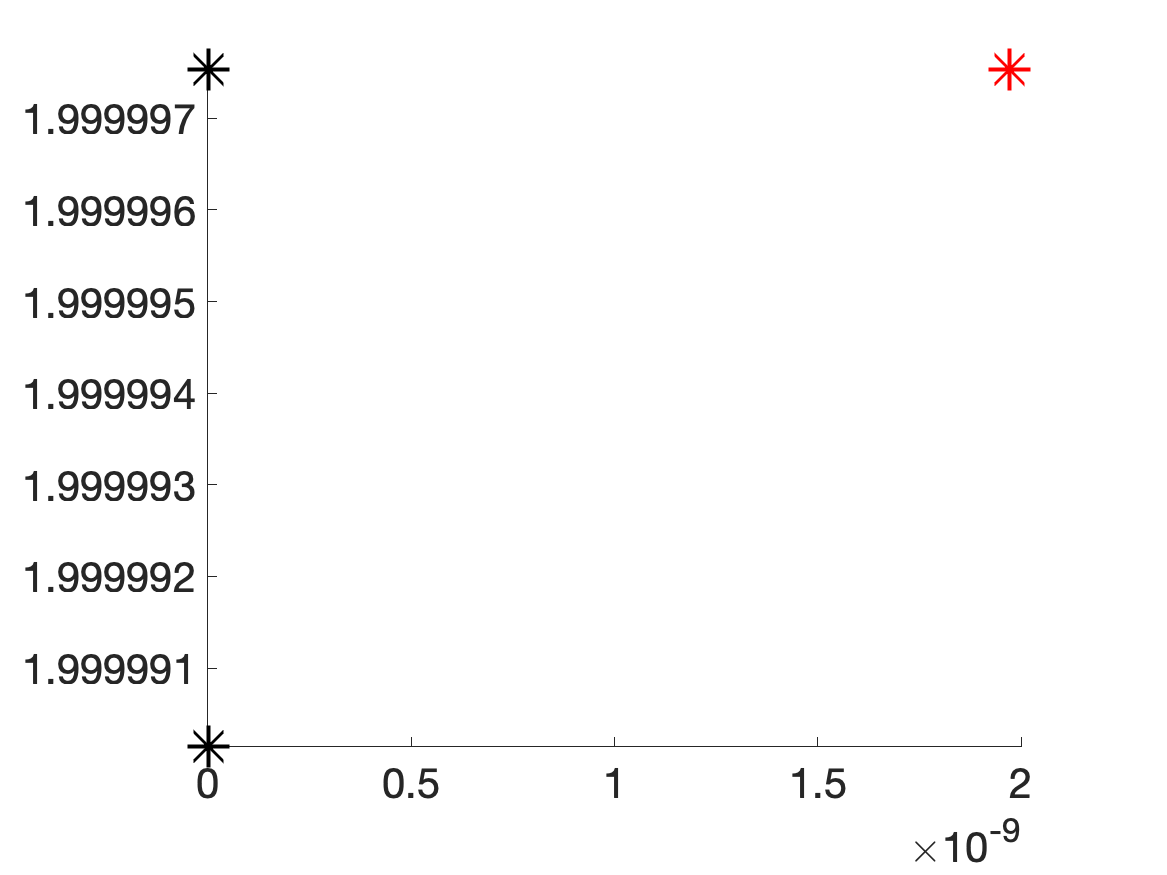}\includegraphics[width=7cm]{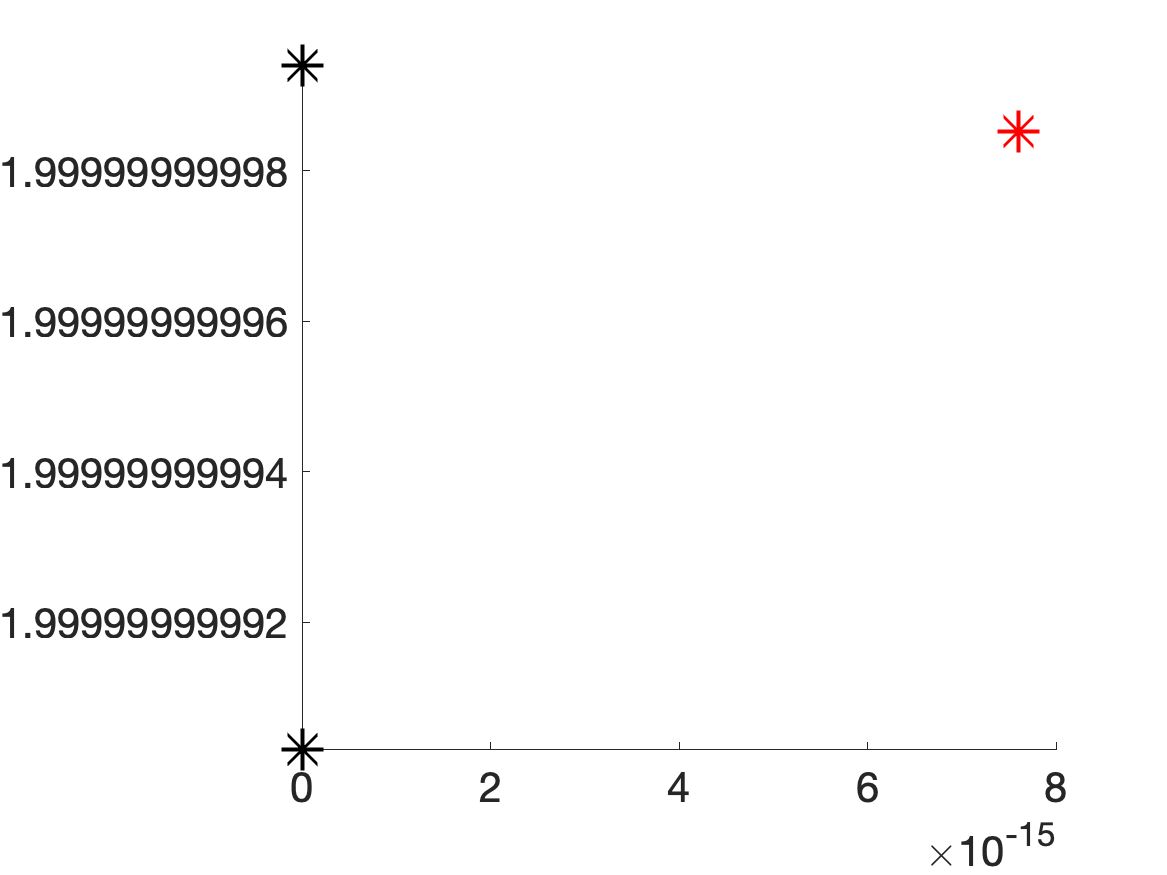}
    \caption{In the top row, we illustrate the scaling of the real part of the smallest eigenvalue and distance of imaginary eigenvalue to the largest element $\mu_N$ on the imaginary axis. The bottom row shows $\mu_{N-1}^+,\mu_N^+$ on the imaginary axis (black crosses) and the eigenvalue whose real part is closest to zero (red cross). On the left we illustrate the case $\omega=0$ and on the right we show $\omega=1/4.$ We see that for $\omega=1/4$ the decay of the spectral gap (real part) is faster than the convergence of the imaginary part to the nearest eigenvalue $\mu_N.$ }
    \label{fig:my_label2}
\end{figure}
\begin{theo}[Scaling of the spectral gap when $\omega=1/4$] \label{theo:spectral gap critical case}
Consider the chain of oscillators with next-to nearest neighbour boundary interactions, Dirichlet boundary conditions, with $\eta>0$, and with friction $\gamma>0$ at one terminal end of the chain, only. 

When $\omega=\frac{1}{4}$, there is an eigenvalue $\mu(N)$ of the generator $\mathcal{L}$,  such that 
%\begin{itemize}
%\item
 for $\mu_N^+$ being the largest eigenvalue of $A$, $$  |\mu(N) - \mu_N^+| \sim N^{-4}.$$ 
 In particular, for the spectral gap of the generator, $g(N)$, we find $ g(N)\lesssim N^{-4}$.
%\item When $\omega < \frac{1}{4}$, for the friction $\gamma$ and a constant $\varepsilon(\gamma)$ both sufficiently small, the eigenvalues $\mu$ of $M_{[N]}$ lie in the following regions $$\mu \in B_{\mathbb{C}}(\mu_j^{\pm}, |v_j^{\pm}(1)|^2) /B_{\mathbb{C}}(\mu_j^{\pm}, \varepsilon(\gamma)|v_j^{\pm}(1)|^2).$$  For all $\gamma$, the spectral gap of the generator satisfies $ N^{-3} \lesssim  \Re(\mu(N)) \lesssim N^{-3}$. 
%\end{itemize}
\end{theo}

\begin{proof} 
\noindent 
\emph{Step 1 - Reducing the dimension}: 
From \eqref{eq: expanding_det}, we reduce our spectral problem to studying the spectrum of the lower-dimensional \emph{Wigner} matrix, defined for $\mu \notin \Spec(i A)$ as
  \begin{equation} \label{eq:Wigner_mat}
  W_F (\mu):= \operatorname{Tr} \left( (  \mu-iA )^{-1} \begin{pmatrix} 
\Gamma & 0 \\ P & 0
\end{pmatrix}   \right).
\end{equation}
Then an eigenvalue $\mu$ of $i M_{[N]}$ corresponds to a solution $\mu$ satisfying $W_F (\mu) =-i$. In other words, we equivalently look, by the spectral decomposition of $A$ at solutions $\mu$ so that
 \begin{equation} 
  \sum_{\pm, i=1}^N(\mu-\tilde{\mu}_i^{\pm})^{-1} \left\langle  (v_i^{\pm}, i \tilde{\mu}_i^{\pm} v_i^{\pm})^T,  
 \begin{pmatrix} 
\Gamma & 0 \\ P & 0
\end{pmatrix}  
(v_i^{\pm}, i \tilde{\mu}_i^{\pm} v_i^{\pm})^T \right\rangle = -i
\end{equation}
where $\tilde{\mu}_j^{\pm} := i\mu_j^{\pm} =\pm \sqrt{\nu_j} \in \mathbb{R}$ are the eigenvalues of $iA$. As a first step we translate $W_F$ by $\tilde{\mu}_N^{+} :=  \sqrt{\nu_N}$ so that we localise to a single eigenvalue. Then the idea is to find a solution $\mu$ close to this largest eigenvalue $\tilde{\mu}_N^{+}$ so that it approximates $\tilde{\mu}_N^{+}$ at an explicit rate in $N$.  

\noindent 
\emph{Step 2 - Spacing between the eigenvalues - why $\omega=\frac{1}{4}$ is a special value}: We denote by $\kappa_j^{\pm} := \tilde{\mu}_j^{\pm} - \tilde{\mu}_N^{+}$ and we observe that for $j \neq N$ this quantity is lower bounded. 
 Indeed for a general $\omega$: 
 \begin{equation}
 \begin{split} 
\kappa_j^{\pm}  = \frac{(\tilde{\mu}_j^{+})^2 -(\tilde{\mu}_N^{+})^2}{\tilde{\mu}_j^{+}+\tilde{\mu}_N^{+}} &\sim (\tilde{\mu}_j^{+})^2 -(\tilde{\mu}_N^{+})^2 
= \nu_j - \nu_N 
=- (\lambda_N- \lambda_j) [(1+4\omega)- \omega(\lambda_N+\lambda_j) ] \\ &
= -(\lambda_N- \lambda_j)[ (1+4\omega) - 2\omega \lambda_N - \omega (\lambda_j- \lambda_N) ]
 \\ & 
\sim -(\lambda_N- \lambda_j)(1-4\omega) - \omega (\lambda_N- \lambda_j)^2 \\
& \sim - \left( 4 - 4 \sin^2 \left( \frac{\pi j }{2(N+1)} \right)  \right)  (1-4\omega) -  \omega \left( 4 - 4 \sin^2 \left( \frac{\pi j}{2(N+1)} \right)  \right)^2 
\\ & \sim -
(1-4\omega) \cos^2\left(  \frac{\pi j}{2(N+1)} \right)  - \omega \cos^4 \left(  \frac{\pi j}{2(N+1)} \right).
\end{split}
\end{equation}

In particular when $j=N-1$, \begin{equation}
\label{eq:Nm4}
\kappa_{N-1}^+ \sim N^{-2} [1-4\omega+ 2\omega N^{-2}] =
N^{-2} (1-4\omega) + 2\omega N^{-4}.
\end{equation}
Note that for $\omega=1/4$ the term of the lower order $N^{-2}$ in the right-hand side above vanishes and so in this case the spacing between eigenvalues is smaller. %\begin{align} 
%\kappa_i &\gtrsim \left(4  - 4 \sin^2 \left( \frac{\pi i}{2(N+1)} \right) \right)\left( 2-\frac{1}{4}\left( 4 + 4 \sin^2\left(  \frac{\pi i}{2(N+1)} \right) \right) \right)  \\ 
%& \gtrsim \cos^2\left(  \frac{\pi i}{2(N+1)} \right) \left(1-\sin^2\left(  \frac{\pi i}{2(N+1)} \right) \right) \\
%&\gtrsim N^{-4}.
%\end{align} 

\noindent
\emph{Step 3 - Scalar reduction}: 
We rewrite the Wigner matrix as 
 \begin{equation}
 \begin{split}
 W_F(\mu+\tilde{\mu}_N^+) =&  \sum_{j\in \pm[N] \setminus \{N\}} \Big((\mu-\kappa_j^{\pm})^{-1} |v_j^{\pm}(1)|^2 [\gamma -2\omega i \tilde{\mu}_j^{\pm}]    \Big)+ \gamma  \mu^{-1}  |v_N^{+}(1)|^2  [\gamma -2\omega i \tilde{\mu}_N^{+}].
 \end{split}
 \end{equation} 
 
 We then rewrite our eigenvalue problem as a sum of two polynomials in $\mu$: 
 $$ W_F(\mu+ \tilde{\mu}_N^+)=-i \Leftrightarrow  f(\mu) + g(\mu)  =0 $$ where $f(\mu) $ and $g(\mu) $ 
  are defined as follows:
  \begin{align*} 
f(\mu) =  i \mu + &   \sum_{(\pm,j) \neq (+,N)}  \frac{ [\gamma - \frac{i}{2} \tilde{\mu}_j^{\pm}] |v_j^{\pm}(1)|^2}{1 - \frac{\kappa_j^{\pm}}{\mu}}
\ \text{ and }\
    g(\mu) =  [\gamma - \tfrac{i}{2} \tilde{\mu}_N^{+}] |v_N^{+}(1)|^2. 
\end{align*}
We have then reduced the study of $\mu$ so that $W_F(\mu+\tilde{\mu}_N^+) =-i$ to finding $\mu$ so that 
$f(\mu) + g(\mu) =0.$

 %We then rewrite our eigenvalue problem as a sum of two polynomials in $\mu$: 
 %$$ W_F(\mu+ \tilde{\mu}_N^+)=-i \Leftrightarrow  f(\mu) + g(\mu)  =0 $$ where $f(\mu) $ and $g(\mu) $ 
  %are defined as follows:
%\begin{equation}
%\begin{split}
% f(\mu) = 
% - \sum_{(\pm,j) \neq (+,N)} \left[\gamma - 2 \omega i \left( \sqrt{(1+4\omega)\lambda_j- \omega \lambda_j^2} \right)^{\pm} \right] \left[ \frac{|v_j^{\pm}(1)|^2 \mu^2}{(\kappa_j^{\pm})^2} + \frac{|v_j^{\pm}(1)|^2 \mu^3}{(\kappa_j^{\pm})^2(\kappa_j^{\pm}-\mu)} \right] 
% \end{split}
%  \end{equation}
% and 
%\begin{equation} 
%\begin{split} 
 %g(\mu)  =i \mu - &\sum_{(\pm,j) \neq (+,N)} \frac{\mu}{\kappa_j^{\pm}}|v_j^{\pm}(1)|^2 \left[ \gamma - 2 \omega i \left(  \sqrt{(1+4\omega)\lambda_j - \omega \lambda_j^2}\right)^{\pm} \right] + \\ &  |v_N^{+}(1)|^2 \left[\gamma - 2 \omega i \sqrt{(1+4\omega)\lambda_N- \omega \lambda_N^2} \right].
 %-\\  -&  \sum_{i=1}^{N-1}\left[\gamma \pm 2 \omega i \sqrt{(1+4\omega)\lambda_i- \omega \lambda_i^2} \right] \left[ \frac{|v_i(1)|^2 \mu^2}{\kappa_i^2} + \frac{|v_i(1)|^2 \mu^3}{\kappa_i^2(\kappa_i-\mu)} \right]. 
% \end{split}
% \end{equation}
%We have then reduced the study of $\mu$ so that $W_F(\mu+\tilde{\mu}_N^+) =-i$ to finding $\mu$ so that 
%$f(\mu) + g(\mu) =0.$

\noindent
\emph{Step 4 - Estimates on $f,g$ for $\omega= \frac{1}{4}$ }:
We first fix a ball of radius $r_N$, that we denote by $B(0,r_N)$.  The radius will be chosen in the following to have that $$|f(\mu)| > |g(\mu)|\text{ on the boundary of the ball with } \mu \in \partial B(0,r_N).$$ 

We notice that $|g(\mu)|\sim N^{-3}$: 
This is since $|v_N^+(1)|^2 \sim N^{-3}$ and the coefficient $ [\gamma - \frac{i}{2} \tilde{\mu}_N^{+}]$ is bounded by constants independently of $N$ from below and above. 

The second term in $f$ on $\partial B(0,|\mu|)$ is 
$$ \left\vert
\sum_{(\pm,j) \neq (+,N)} \frac{|v_j^{\pm}(1)|^2}{1 - \frac{\kappa_j^{\pm}}{\mu}} \right\vert   =\left\vert \sum_{j \in [N-1]
} \frac{|v_j^{+}(1)|^2}{1 - \frac{\kappa_j^{+}}{\mu}}\right\vert  + \mathcal{O}(|\mu|).$$
The $\mathcal{O}(|\mu|)$ term accounts for the sum containing $\kappa_j^{-}$'s: Using that the denominator $\kappa_j^{-}$ is bounded from above and below, we have that it behaves as $|\mu|\sum_j |v_j^{-}(1)|^2 \sim \mathcal{O}(|\mu|)$. Note also that we have neglected the coefficient $ (\gamma - \frac{i}{2} \tilde{\mu}_j^{\pm}) $ as for all indices $j$, these factors are uniformly bounded in $N$. We write therefore 
 \begin{align*}
     |f(\mu)| \sim |\mu| + |\mu|\left\vert   \frac{1}{N} \sum_{j\in [N-1]} \frac{4 \sin^2\left( \frac{\pi j }{N+1}\right)}{\mu - \cos^4 \left( \frac{\pi j }{2(N+1)}\right)}\right\vert
 \end{align*}
  where regarding the $\kappa_j^+$'s terms we used that
\begin{align*} 
\kappa_j^{+}&\sim  2(\lambda_j-\lambda_N)+1/4(\lambda_j^2-\lambda_N^2)\sim  - (\lambda_j-\lambda_N)(2- 1/4(\lambda_j+\lambda_N)) 
\\ & \sim 1/4(\lambda_j-4)^2\sim 4 \cos^4 \left(\frac{\pi j}{2(N+1)}\right). 
\end{align*}
Thus, we need to examine how the following sum scales
 \begin{align}
     \left\vert \frac{1}{N} \sum_{j\in [N-1]} \frac{4 \sin^2\left( \frac{\pi j }{N+1}\right)}{\mu- \cos^4 \left( \frac{\pi j }{2(N+1)}\right)}\right\vert, 
 \end{align}
which by Appendix \ref{lemm:Riemann} can be reduced to studying the integral $ \int_{\frac{1}{N}}^{\frac{N-1}{N+1}} \frac{\sin^2(\pi x)}{\mu - \cos^4\left( \frac{\pi x}{2}\right)} dx$, at least as long as $\mu<\cos^4(\frac{\pi}{2N}).$
 %Here $\alpha\sim N^{-\theta} h(N)$ for $\theta > 0$. 
 This integral scales as
 \begin{align*}
     &\int_{\frac{1}{N}}^{\frac{N-1}{N+1}} \frac{\sin^2(\pi x)}{\mu - \cos^4\left( \frac{\pi x}{2}\right)} dx = 
     \Bigg[ 4x -  4\sqrt{\frac{\sqrt{\mu}+\mu}{\mu}} \arctan \left( \sqrt{\frac{\mu}{\sqrt{\mu}+\mu}}\tan(\pi x/2) \right) -\\
     &\quad 4\sqrt{\frac{\sqrt{\mu}-\mu}{\mu}} \operatorname{arctanh} \left( \sqrt{\frac{\mu}{\sqrt{\mu}-\mu}}\tan(\pi x/2) \right)
     \Bigg]_{\frac{1}{N}}^{\frac{N-1}{N+1}}    \sim \mu^{-1/4} \text{ as }\mu \downarrow 0.
 \end{align*}
 
 %{\color{blue} I don't think it is clear how you ended up having this scaling.}
 The last estimate follows since $\tan\Big( \frac{\pi(N-1)}{2(N+1)} \Big)
  $ scales as $N$, $ \tan\Big(\frac{\pi}{2(N+1)}\Big)$ scales as $N^{-1}$, $\sqrt{\frac{\mu}{\sqrt{\mu}+\mu}} \sim \mu^{1/4} $ and that $ [\arctan \left(\mu^{1/4}N  \right) -\arctan \left(\mu^{1/4}N^{-1}  \right) ] \sim \mathcal{O}(1)$. 
  %$$
  %\arctan \left( \sqrt{\frac{a}{\sqrt{a}+a}}\tan( \frac{\pi(N-1)}{2(N+1)} )
  %)\right) - \arctan \left( \sqrt{\frac{a}{\sqrt{a}+a}}\tan(\frac{\pi}{2(N+1)})  \right)
  %\sim 
  %\arctan ( a^{1/4} N ) - \arctan(a^{1/4} %N^{-1}) 
 %$$
 
 This means that $|f(\mu)| \sim |\mu| \left( 1+ |\mu|^{-1/4}  \right) 
 %\sim N^{-\theta} \left( 1+ N^{\theta/4} \right) 
 $, 
 while $|g| \sim N^{-3}$. 
 While the above computation is only valid for $0<\mu<\cos^4(\pi/(2N))$, since we are otherwise integrating over a singularity, a splitting argument as in \eqref{eq:est2_term in g} establishes the validity of the above scaling for the full range of $\mu.$
 
 \noindent
\emph{Step 5 - Upper and lowers bounds on the distance from $\mu_N^+$:}  
To summarize, we have shown in the previous section that 
\[ \vert g(\mu)\vert \sim N^{-3} \text{ and }\vert f(\mu)\vert \sim \mu^{3/4} \text{ as }\mu \downarrow 0.\]

Then choosing $r_N=|\mu|=\beta N^{-4}$, $\beta$ large enough, we see that on $\partial B(0,r_N)$,   $|f(\mu)| \sim \beta N^{-3} \gtrsim |g(\mu)|\sim N^{-3}$. 
Then this implies by 
Rouch\'{e}'s theorem since $f$ has one zero at zero, that there is one solution to $f(\mu)+ g(\mu)=0$ in $B(0,r_N)$ or since we have localised our eigen-problem to the eigenvalue $\tilde{\mu}_N^+$, there is one eigenvalue $\mu = \mu(N)$ of $i M_{[N]}$ in $B(\tilde{\mu}_N^+, r_N)$, i.e. \begin{align} \label{eq:UB_critical}
|\mu -\tilde{\mu}_N^+| \lesssim N^{-4}.
\end{align}
Therefore the real part of the smallest eigenvalue of the generator $\mathcal{L}$ decays at least as fast as $N^{-4}$. 
%From numerical results, see fig. \ref{fig:gap_forN_and_for_omega},  we know that the optimal expected scaling of the spectral gap is $N^{-5}$. From this proof we do not catch this rate in the upper bound and the reason is that we should, in contrast with the nearest-neighbour case, to localise our eigenproblem not to the largest eigenvalue $\tilde{\mu}_N^+$ as we do now, but in between the two largest eigenvalues $\tilde{\mu}_{N, N-1}^+$ with $\tilde{\mu}_{ N-1}^+< \tilde{\mu}_{N, N-1}^+ <\tilde{\mu}_{N}^+$. This is indicated in fig. \ref{fig:where_to_localise}. {\color{blue} This sounds odd, it sounds as if we could do it, but decide not to since we are lazy.}

%\noindent
%\emph{Step 6 - Lower bound on the spectral gap when $\omega=\frac{1}{4}$:}
On the other hand, if one takes $|\mu| = \beta N^{-4}$, with $\beta >0$ small enough one gets that always $|g(\mu)| > |f(\mu)|$ with $g$ not having any zero inside $B(0,r_N)$. This implies by 
Rouch\'{e}'s theorem  again that there is no solution of $f(\mu)+ g(\mu)=0$ in $B(0,r_N)$ or equivalently that \begin{align} \label{eq:LB_critical} |\mu -\tilde{\mu}_N^+| \gtrsim N^{-4}.\end{align}

\end{proof}

\subsection{Full spectrum for $ \omega$ sufficiently small} \label{subsec:spectrum small omega}
Now we move on to estimating the spectral gap for next-to-nearest-neighbour interactions with coupling strength $\omega<1/4$. In this case we show that the spectral gap behaves as in the case of nearest-neighbour interaction, only. 

We shall consider a two-spec procedure for our perturbation problem unlike what we did in the first part of the article. In particular we write our matrix $M_{[N]}$ as 
\begin{equation} 
M_{[N]} = 
\begin{pmatrix} 
0 & -I \\ T_\omega + P & 0 
\end{pmatrix} +
\begin{pmatrix} 
\Gamma & 0 \\ 0 & 0 
\end{pmatrix} : = 
A_{\omega} + 
\begin{pmatrix} 
\Gamma & 0 \\ 0 & 0 
\end{pmatrix}.
\end{equation}
 Since $P =\operatorname{diag}(-2\omega, 0, \dots, 0)$, the matrix $M_{[N]}$ is a rank-$1$ perturbation of the matrix 
$A_\omega$. We shall show that one can still obtain explicit estimates on the spacing between the eigenvalues and on the eigenvectors of the matrix $A_{\omega}$. This is because the spectrum of $T_\omega + P$ can be compared and explicitly written in terms of the spectrum of $T_\omega$. Thus the spectrum of $A_{\omega}$ can be directly compared with the spectrum of $A$ given by \eqref{eq:spec A_0}. This is the content of the following theorem. 

\begin{theo}
\label{theo:spectr of A_omega}Let $\omega$ be sufficiently small. Let $\mu_j^{\pm} = \pm i \sqrt{\nu_j}$ be the eigenvalues of $A$ with corresponding eigenvectors $V_j^{\pm} = (v_j, - \mu_j^{\pm}v_j)^T$, where $(\nu_j, v_j)$ is the explicit eigensystem of $T_\omega$. The eigenvalues $\xi_j$ of $T_\omega+P$ then satisfy \begin{equation}\label{eq:evalues of T+P}
    |\xi_j - \nu_j|  \sim \omega |v_j(1)|^2,\ j \in [N]
\end{equation} 
and interlace $\xi_1<\nu_1<\xi_2<\nu_2<...<\xi_n<\nu_n$
with associated eigenvectors
\begin{equation} \label{eq: evector w_j}
w_j = \frac{1}{\sqrt{\sum_{k=1}^N \frac{|v_k(1)|^2}{(\nu_k - \xi_j)^2} } } 
\left( \sum_{k=1}^N \frac{|v_k(1)|^2}{(\nu_k - \xi_j)}, \sum_{k=1}^N \frac{v_k(1)v_k(2)}{(\nu_k - \xi_j)}, \dots, \sum_{k=1}^N \frac{ v_k(1)v_k(N)}{(\nu_k - \xi_j)}   \right)^T.
\end{equation} 
In particular it holds
$$|w_j(1)|^2 \sim |v_j(1)|^2. $$
Consequently the eigensystem of $A_\omega$ is
$$\left(\pm i \sqrt{\xi_j}, W_j^{\pm}\right), \text{ with } W_j^{\pm} = \left(w_j, \mp i \sqrt{\xi_j} w_j \right)^T,\ j \in [N].$$
\end{theo}

\begin{proof}
    Regarding the eigenvalues of $T_\omega+P$, first we write $T_\omega+P = T_\omega - 2\omega e_1e_1^T$, so that Sylvester's identity, see Lemma \ref{lemm:Sylvester}, yields
    \begin{align*}
    &\det(T_\omega +P - \xi \operatorname{Id}_{\mathbb{C}^{N\times N}}) = \det (T_\omega- \xi \operatorname{Id}_{\mathbb{C}^{N\times N}} ) \det( \operatorname{Id}_{\mathbb{C}^{N\times N}}  - 2\omega (T_\omega- \xi \operatorname{Id}_{\mathbb{C}^{N\times N}} )^{-1}e_1e_1^T )\\ 
    &=  \det (T_\omega- \xi \operatorname{Id}_{\mathbb{C}^{N\times N}} ) 
     ( 1 - 2\omega e_1^T (T_\omega- \xi \operatorname{Id}_{\mathbb{C}^{N\times N}} )^{-1}e_1). 
    \end{align*} 
    This implies by the spectral decomposition of $T_{\omega}$, that for $\xi \notin \Spec(T_\omega)$, the eigenvalues $\xi$ of $T_\omega+P$ solve 
    $1= 2\omega \sum_{k \in [N]} \frac{|v_k(1)|^2}{\nu_k -\xi} : = W_{\omega} (\xi)$, 
    where $(\nu_j, v_j)$ is the eigensystem of $T_{\omega}$. We note that all the eigenvalues are real and we translate our problem by $\nu_j$: $R_\omega(\xi):= W_\omega(\xi + \nu_j) = 2\omega\sum_{k} \frac{|v_k(1)|^2}{\kappa_k -\xi} = 1$. Here $\kappa_k := \nu_k  - \nu_j$ which is lower bounded by the same calculation as below in \eqref{eq:spacing of evalues xi_j}. Finding now a solution to $R_\omega(\xi)=1$ is a simper version of the eigenproblem solved in Theo. \ref{theo: spectrum_magnetic field}, as the matrices are symmetric so the spectrum lies on the real line. In the following, we shall construct rational functions $f(\xi), g(\xi)$ so that $(f + g)(\xi) = \xi ( R_\omega(\xi) -1)$.
  Using appropriate estimates, we shall then exhibit solutions such that 
\begin{equation}
\label{eq:bounds}
 C_1 \omega |v_j(1)|^2 \leq |\xi_j - \nu_j| \leq  C_2 \omega |v_j(1)|^2
 \end{equation}
  for some universal constants $C_1, C_2>0$ independent of $N, \omega, j$. Indeed take 
  $$ f(\xi) = \xi \left( - 1 - 2\omega \sum_{k \neq j} \frac{\vert v_k(1)\vert^2}{ \kappa_k }\right),\ 
g(\xi) = 2\omega\Bigg(\vert v_j(1)\vert^2-\sum_{k \neq j} \frac{\vert v_k(1)\vert^2 \xi^2}{\kappa_k^2}-\sum_{k \neq j} \frac{\vert v_k(1)\vert^2 \xi^2}{(\kappa_k)^2(\kappa_k/\xi-1)}\Bigg).  $$ 
The estimates in the proof of Theorem \ref{theo: spectrum_magnetic field} yield for some positive constant $c$: 
$$ |f(\xi)| \sim |\xi|,\ 2\omega |v_j(1)|^2 - 2\omega c N |\xi|^2 \lesssim |g(\xi)| \lesssim 2\omega |v_j(1)|^2 +2\omega c N |\xi|^2.$$
Now \emph{for the upper bound} in \eqref{eq:bounds}, we define the radius of the ball $B(0,r_N)$, $r_N:= \alpha \omega |v_j(1)|^2$ for $\alpha$ large enough and $\omega$ sufficiently small. We want to have that $|g(\xi)| < |f(\xi)| \sim |\xi|$ on $\partial B(0,r_N)$. That is when
$2\omega |v_j(1)|^2 +2\omega c N |\xi|^2 < \alpha \omega  |v_j(1)|^2 $ which is the case as long as $ 2\omega + 2 c \omega^3 \alpha^2< \alpha \omega$ or $\omega < (2\alpha c)^{-1/2}$. Then since $f$ has one zero in $B$, so does $f+g$ in $B$, say $\xi_j$. This implies the upper bound on $|\nu_j -\xi_j|$. In particular, it also implies $\vert \xi_{j}-\nu_{j-1}\vert \ge C_2 \omega \vert \nu_j(1)\vert^2$, since there is precisely one zero. The interlacing property follows since $P$ is a monotone rank $1$-perturbation (Weyl inequalities).  \\
\emph{For the lower bound}, we argue analogously: Let $|\xi| = \varepsilon \omega |v_j(1)|^2$ for $\varepsilon>0$ small enough, then as $g$ is lower bounded by a leading order term $2\omega |v_j(1)|^2$, it does not have any solution in $B(0,|\xi|)$. Moreover $|g| > |f|$ on $\partial B(0,\vert \xi\vert)$, implying by Rouch\'{e}'s theorem that neither $g+f$ has a solution in $B$. Therefore $|\xi_j - \nu_j |\geq C_1\omega|v_j(1)|^2$ for some universal constant $C_1$.

 Concerning the (normalised) eigenvectors $w_j$ of $T_\omega+P$, with corresponding eigenvalue $\xi_j$, we arrive at the stated formula \eqref{eq: evector w_j} by applying \cite[Theo. 5]{BNS78}. In particular the general expression for these eigenvectors is $w_j = Q (D-\xi_j \operatorname{Id} )^{-1} Q^T x_\omega/ \| (D-\xi_j \operatorname{Id} )^{-1} Q^Tx_\omega \|_2 $, where $x_\omega$ is so that $x_\omega x_\omega^T = -P$ and where $Q DQ^T$ is the orthogonal decomposition of $T_{\omega}$: the columns of $Q$ are made of the eigenvectors $v_i$ of $T_\omega$. In our notation where $x_\omega = (\sqrt{2\omega}, 0, \cdots, 0)^T$,  we get the formula \eqref{eq: evector w_j}.
 
 %we have the explicit formula, \cite[Theo. 5]{BNS78} : 
%\begin{equation} \label{eq: evector w_j}
%w_j = \frac{1}{\sqrt{\sum_{k=1}^N \frac{|v_k(1)|^2}{(\nu_k - \xi_j)^2} } } 
%\left( \sum_{k=1}^N \frac{|v_k(1)|^2}{(\nu_k - \xi_j)}, \sum_{k=1}^N \frac{v_k(1)v_k(2)}{(\nu_k - \xi_j)}, \dots, \sum_{k=1}^N \frac{ v_k(1)v_k(N)}{(\nu_k - \xi_j)}   \right)^T.
%\end{equation} 

We then write $\xi_j = \nu_j + \delta_j \vert v_j(1)\vert^2$ and find 
\begin{equation} \label{eq:formula for w_j}
\begin{split}
 |w_j(1)|^2 &= \left\vert
 \frac{\sum_{k=1}^N \frac{|v_k(1)|^2}{(\nu_k - \xi_j)}}{\sqrt{\sum_{k=1}^N \frac{|v_k(1)|^2}{(\nu_k - \xi_j)^2} }}\right\vert^2 
 =
\left\vert \frac{\sum_{k \neq j} \frac{|v_k(1)|^2}{(\nu_k - \nu_j -\delta_j \vert v_j(1)\vert^2)} - \frac{1}{ \delta_j }   }{  \sqrt{  \sum_{k\neq j} \frac{|v_k(1)|^2}{(\nu_k - \nu_j -\delta_j |v_j(1)|^2 )^2} + \frac{|v_j(1)|^2}{ \delta_j^2 |v_j(1)|^4 } }  }\right\vert ^2\\
 &=\vert v_j(1)\vert^2 
\left\vert \frac{\sum_{k \neq j} \frac{|v_k(1)|^2}{(\nu_k - \nu_j -\delta_j \vert v_j(1)\vert^2)} - \frac{1}{ \delta_j }   }{  \sqrt{ \vert v_j(1)\vert^2 \sum_{k\neq j} \frac{|v_k(1)|^2}{(\nu_k - \nu_j -\delta_j |v_j(1)|^2 )^2} + \frac{1}{ \delta_j^2 } }  }\right\vert ^2.
%\\ &= \left(  \frac{\mathcal{O}(1)}{\mathcal{O}(N) +  \mathcal{O}(|v_j(1)|^{-2})  }\right)^2
 \end{split}
\end{equation}
To see that $\vert w_j(1)\vert^2 \sim \vert v_j(1)\vert^2$ it then suffices to observe that 
\begin{equation}
\label{eq:estimate_1}
\left\vert \frac{\sum_{k \neq j} \frac{|v_k(1)|^2}{(\nu_k - \nu_j -\delta_j \vert v_j(1)\vert^2)} - \frac{1}{ \delta_j }   }{  \sqrt{ \vert v_j(1)\vert^2 \sum_{k\neq j} \frac{|v_k(1)|^2}{(\nu_k - \nu_j -\delta_j |v_j(1)|^2 )^2} + \frac{1}{ \delta_j^2 } }  }\right\vert ^2 \sim 1.
\end{equation}

This is readily shown as in \eqref{eq:est1_term in f}, as $\delta_j \vert v_j(1)\vert^2$ is a negligible shift that does not affect the scaling.
Thus, we have that 
\[ \sum_{k \neq j} \frac{|v_k(1)|^2}{(\nu_k - \nu_j -\delta_j \vert v_j(1)\vert^2)}  = \mathcal O(1).\]
Thus, since by \eqref{eq:bounds}, we have $\vert 1/\delta_j\vert \sim \frac{1}{\vert \omega\vert}$, the numerator behaves $\sim \vert \omega \vert^{-1}.$

For the denominator, we argue as in \eqref{eq:est2_term in g} to see that 
\[ \vert v_j(1)\vert^2 \sum_{k\neq j} \frac{|v_k(1)|^2}{(\nu_k - \nu_j -\delta_j |v_j(1)|^2 )^2} = \mathcal O(1).\] Thus, the leading contribution is given by $ 1/\delta_j^2 \sim 1/\omega^2$ for $\omega$ small. Combining the estimates for both the numerator and denominator, we find \eqref{eq:estimate_1}.

\end{proof}

\begin{theo}[Scaling of the spectral gap when $\omega$ sufficiently small] \label{theo:spectrum for small omega}
Consider the chain of oscillators with $\eta>0$ and with friction at one end. Also consider Dirichlet boundary conditions so that $P$ is a rank-$1$ perturbation: $T_\omega+ P = (1+4\omega) \mathscr V_1 - \omega \mathscr V_1^2 + \operatorname{diag}(-2\omega, 0, \dots, 0).$
As long as $ \omega$ is small enough, the smallest real part of the eigenvalues of $\mathcal{L}$, the spectral gap $g(N)$ satisfies 
\begin{equation} 
g(N) \sim N^{-3}.
\end{equation} 
Taking also the friction $\gamma$ and a constant $\varepsilon(\gamma)$ both sufficiently small, the eigenvalues $\mu$ of $M_{[N]}$ lie in the following regions $$\mu \in B_{\mathbb{C}}(\xi_j^{\pm}, |W_j^{\pm}(1)|^2) /B_{\mathbb{C}}(\xi_j^{\pm}, \varepsilon(\gamma)|W_j^{\pm}(1)|^2)$$
where $\xi_j^{\pm}:=\pm i \sqrt{\xi_j}$ are the eigenvalues to $A_{\omega}$ with corresponding eigenvectors $W_j^{\pm} = (w_j, \mp i \sqrt{\xi_j} w_j )^T$ in the notation of Theorem \ref{theo:spectr of A_omega}.
\end{theo}

\begin{proof}
\emph{Step 1 - Spacing of the eigenvalues}: Following the notation above we denote by $(\xi_j,w_j)$ the eigensystem to $T_\omega+ P$. Then the eigensystem to $A_{\omega}$ is given by $(\pm i \sqrt{\xi_j}, W_j^{\pm})$, with $W_j^{\pm} = (w_j, \mp i \sqrt{\xi_j} w_j )^T$, $j \in [N]$.  In order to work with real numbers, we multiply everything by $i$ so that 
 $ z_j^{\pm}:= \mp \sqrt{\xi_j} \in \R$, $j \in [N]$  are the eigenvalues of $iA_{\omega}$.

% We have the following two facts regarding the spectrum of $T_\omega+P$:

We start by recalling that the spacing $\kappa_j^{\pm}:=z_j^{\pm}-z_i^{\pm}$ between the eigenvalues is lower-bounded in the same way as in the proof of Theorem \ref{theo: spectrum_magnetic field}: 
 \begin{equation} \label{eq:spacing of evalues xi_j}
 \begin{split}
  \vert {\kappa}_j^{\pm}\vert&=  \vert z_j^{\pm} - z_i^{\pm}\vert \gtrsim   \vert \sqrt{\xi_j} - \sqrt{\xi_i} \vert \sim \vert \xi_j - \xi_i\vert
  \\ &\geq  \vert \nu_j-\nu_i\vert  -C \omega (\vert v_j(1)\vert^2+\vert v_i(1)\vert^2) \\
  &\ge (1-4\vert \omega\vert)\vert \lambda_j-\lambda_i\vert - C \vert \omega\vert (\vert v_j(1)\vert^2+\vert v_i(1)\vert^2)\\
  &\gtrsim_{\omega} \vert \lambda_i-\lambda_j\vert
  \end{split}
 \end{equation}
 where we used $\vert \nu_j-\nu_i\vert \ge (1-4\vert \omega\vert) \vert \lambda_j-\lambda_i\vert$ for $\omega \in (-1/4,1/4)$ and the smallness of $\omega$ together with \eqref{eq:uniformly_bounded} in the last step.

\noindent
\emph{Step 2 - Reduction of the dimension}: 
As usual we reduce our spectral problem to studying the spectrum of the lower-dimensional Wigner matrix, just as in \eqref{eq:Wigner_mat},
defined for $\mu \notin \Spec(i A_\omega)$ as
  \begin{equation} 
  W_F (\mu):= \operatorname{Tr} \left( (  \mu-i A_\omega )^{-1} \begin{pmatrix} 
\Gamma & 0 \\ 0 & 0
\end{pmatrix}   \right).
\end{equation}
An eigenvalue $\mu$ of $i M_{[N]}$ is then a solution $\mu$ of $W_F (\mu) =-i$. Thus equivalently we look for solutions $\mu$ at
 \begin{equation} \label{eq: wigner equation}
  \sum_{\pm, k=1}^N(\mu-z_k^{\pm})^{-1} \left\langle  (w_k^{\pm}, \pm i \sqrt{\xi_k}w_k^{\pm})^T,  
 \begin{pmatrix} 
\Gamma & 0 \\ 0& 0
\end{pmatrix}  
 (w_k^{\pm}, \pm i \sqrt{\xi_k}w_k^{\pm})^T \right\rangle = -i.
\end{equation}

Noticing that for a fixed arbitrarily chosen $i \in \pm [N]$, the difference of the eigenvalues is lower bounded, cf \eqref{eq:spacing of evalues xi_j},  
brings us to the same  situation as in the proof of Theorem \ref{theo: spectrum_magnetic field} (or in \cite[Proposition 3.2]{BM20}). Thus we are able to study all the eigenvalues around $\xi_j^{+}$ for all $j$'s, rather than just around the largest eigenvalue $\xi_N^{+}$. This is by localising our eigenproblem around every eigenvalue $z_j^{\pm}:= \pm  \sqrt{\xi_j}$ and finding a solution of the corresponding polynomial inside a ball around $z_j^{\pm}$.

 Take without loss of generality $s=+$. First, we translate $W_F$ by $z_i^{+} :=  \sqrt{\xi_i}$, for a fixed index $i \in [N]$, so that we localise our problem around the eigenvalue $z_i^{+}$. The purpose is to find a solution $\mu$ close to this eigenvalue $z_i^{+}$ and quantify in $N$ the convergence rate towards $z_i^{+}$. 

We further reduce the dimension of our problem to a scalar problem. Thus eventually looking for $\mu$ satisfying  
 \begin{equation}\label{eq: W_F=-i localised} 
 \begin{split}
 W_F(\mu+z_i^+) =&  \gamma \sum_{(\pm,j) \setminus \{+,i\}} (\mu-\kappa_j^{\pm})^{-1} |w_j^{\pm}(1)|^2   + \gamma  \mu^{-1}  |w_i^{+}(1)|^2 = -i.  
 \end{split}
 \end{equation}

\emph{Step 3 - Upper Bound on the eigenvalues for $\omega$ sufficiently small}: 
Using the expansion
$$(\lambda-\mu)^{-1} = -\mu^{-1}\sum_{n\ge 0}(\lambda \mu^{-1})^n = -\mu^{-1}-\mu^{-2}\lambda-\mu^{-2}\lambda^2(\mu-\lambda)^{-1},$$ 
%\begin{align*}
   % &\kappa_j^{+} \sim (\tilde{\mu}_j^+)^2-(\tilde{\mu}_i^+)^2 
    %= \nu_j - \nu_i =  -(1+4\omega) (\lambda_i - \lambda_j) + \omega (\lambda_i - \lambda_j)(\lambda_i + \lambda_j) \\
    %&\sim- 4(1-4\omega)  \left(\sin^2\left( \frac{\pi i}{2(N+1)} \right) - \sin^2\left( \frac{\pi j}{2(N+1)} \right) \right), 
%\end{align*} 
we define the polynomials $f,g$ as follows: 
\begin{equation} \label{eq:f,g_omega<1/4}
\begin{split}
 f(\mu) = i \mu - &\gamma\sum_{(\pm,j) \neq (s,i)} \frac{\mu}{{\kappa}_j^{\pm}}|w_j^{\pm}(1)|^2  
 %- \\ & 
 %- \sum_{i=1}^{N-1}\left[\gamma \pm 2 \omega i \sqrt{(1+4\omega)\lambda_i- \omega \lambda_i^2} \right] \left[ \frac{|v_i(1)|^2 \mu^2}{\kappa_i^2} + \frac{|v_i(1)|^2 \mu^3}{\kappa_i^2(\kappa_i-\mu)} \right] 
 \end{split}
  \end{equation}
 and 
\begin{equation} 
\begin{split} 
 g(\mu)  =&  \gamma |w_i^+(1)|^2 - \gamma \sum_{(\pm,j) \neq (s,i)} \left[ \frac{|w_j^{\pm}(1)|^2 \mu^2}{({\kappa}_j^{\pm})^2} + \frac{|w_j^{\pm}(1)|^2 \mu^3}{({\kappa}_j^{\pm})^2({\kappa}_j^{\pm}-\mu)} \right]. 
 \end{split}
  \end{equation}
  We are now equivalently looking for solutions to the equation $f(\mu) + g(\mu)  =0$, since these also solve $R_{F}(\mu):= W_{F}(\mu+z_i^{+})=-i$,   
  %$$\mathcal{R}_{F}(\mu) =-i  \Leftrightarrow  f(\mu) + g(\mu)  =0. $$

%We first notice that 
%\begin{equation}
%\kappa_i \gtrsim (1-4\omega) \cos^2 \left( \frac{\pi i}{2(N+1)} \right).
%\end{equation} 

 Since the eigenvalue difference \eqref{eq:spacing of evalues xi_j} satisfies  as in the proof of Theorem \ref{theo: spectrum_magnetic field}, see \eqref{eq:kappa_estm}, and the eigenvectors obey the same asymptotics by Theorem \ref{theo:spectr of A_omega}, we have that for $\mu \in \partial B(0,r_N)$
 \begin{equation} 
\begin{split} 
\gamma( |w_i^+(1)|^2  -& Nr_N^2)  \lesssim  |g(\mu)| \lesssim   \gamma(|w_i^+(1)|^2  + Nr_N^2 )
\end{split} 
\end{equation}  
while $$ |f(\mu)| \sim r_N = |\mu|.$$ 

Thus choosing $r_N = c  |w_i^+(1)|^2 $ for some constant $c>0$ and the friction $\gamma$ small enough, we have $|f(\mu) |> |g(\mu)|$ on $\partial B(0, r_N)$. This allows to conclude by Rouch\'{e}'s theorem that there exists a solution to $f+g$ inside this ball. This implies the existence of one eigenvalue $\mu(N) \in B_{\mathbb{C}}(z_i^+,c|w_i^+(1)|^2)$. The upper bound on the spectral gap $N^{-3}$ follows, as $|w_N^+(1)|^2 \sim |v_N^+(1)|^2 \sim N^{-3}$, by Theorem \ref{theo:spectr of A_omega}.  
The smallness condition on $\gamma$ here is needed because for an arbitrary index $i$, the eigenvectors merely satisfy $|w_i^+(1)|^2 \sim |v_i^+(1)|^2= \mathcal{O}( N^{-1})$ which is the same order of decay as the other terms in $g$ and $f$. Thus we need to make the coefficients $\gamma$ small enough to get that $|f(\mu)| > |g(\mu)|$ when we localise around any index $i$. Smallness on $\gamma$ is not required however when $i=N$, i.e. in order to get an upper bound on the spectral gap.

\noindent
\emph{Step 4 - Lower Bound on the eigenvalues for $ \omega$ sufficiently small}:
For the lower bound when $|\mu|= \varepsilon(\gamma)|w_i^+(1)|^2$, with $\varepsilon(\gamma)$ small enough, 
we see that there is no eigenvalue $ \mu$ in   $B(z_i^+, \varepsilon(\gamma)|w_i^+(1)|^2 )$. This is since $|f(\mu)|<|g(\mu)|$ on $B(0,|\mu|)$ and there is no solution of $g$ in this ball. This holds for any index $i$, as this was chosen arbitrarily in the beginning, and also for any fixed friction $\gamma$. This gives a lower bound on $| \mu(N) - z_i^+ |$.

Now we want to show that is the imaginary part of $\mu(N)$ that is responsible for the decay (we remind that $\mu(N)$ is an eigenvalue of $i M_{[N]}$). By contradiction, say that  we have a solution $\mu$ to our spectral problem with $\operatorname{Im}(\mu) = o(1)\operatorname{Re}(\mu)$. Taking the imaginary part of $f(\mu) + g(\mu) =0$, implies by restricting to positive signs that $$\operatorname{Re}(\mu) + \operatorname{Im}(\mu)  \sum_{j \neq i}\frac{\gamma  |w_j^{+}(1)|^2}{{\kappa}_j^+} +  \operatorname{Im}(g(\mu))=0.$$ The sum in the above equation is of order $\mathcal{O}(1)$ as was estimated already in Theo. \ref{theo: spectrum_magnetic field}. This implies, taking into account our assumption, that 
\begin{equation} \label{eq:Im(f+g)}
 \operatorname{Re}(\mu)(1+o(1)) +  \operatorname{Im}(g(\mu))=0. 
 \end{equation}
From the estimates on the terms of $g$: $|\operatorname{Im}(g(\mu))| = \mathcal{O}\left(  [\operatorname{Im}(\mu^2) + \operatorname{Im}(\mu^3/({\kappa}_j^+ -\mu) )  ]N \right) \leq \mathcal{O}(N \operatorname{Im}(\mu^2) ) = \mathcal{O}(N \operatorname{Im}(\mu)\operatorname{Re}(\mu)) = \mathcal{O}(N \operatorname{Re}(\mu)^2o(1))= o(1)$. As in \eqref{eq:Im(f+g)} we find that $\operatorname{Re}(\mu)=0$ for large $N$. That is a contradiction as this would imply that there is an eigenvalue in the ball $B(z_i^+, \varepsilon(\gamma)|w_i^+(1)|^2 )$. 
%%%%%%%%

\noindent
\emph{Step 5 - Lower Bound on the spectral gap for $ \omega$ sufficiently small but without the smallness condition on $\gamma$}:
Now, without restricting to small $\gamma$, we want to make sure that it is the imaginary part of the spectral gap $\mu$ that it is responsible for this decay on $|\mu|$, as this would imply that the spectral gap of $M_{[N]}$ has indeed a lower bound of order $N^{-3}$ for all bounded frictions. 
%Before we showed that we have $\Im(\mu)=o(1) \Re(\mu)$ \emph{for any index $i$}, meaning that the imaginary part would decay slower around every single eigenvalue $\xi_i^{\pm}$ which required the smallness condition on $\gamma$. 

To this end we run the same argument as in the lower bound in Theorem \ref{theo:mag}: we shift the eigenvalues $z_j^{\pm}$ horizontally so that we may have $\Re(\mu)=0$ and we assume by contradiction that 
 $\mu = \mathcal{O}(N^{-3}R_N^{-1}) =\mathcal{O}(|w_N^+(1)|^2R_N^{-1})$, for $R_N \to \infty$. Then let us denote by $\phi_j^{\pm}$ the shifted $\xi_j^{\pm}$'s, we take the imaginary part on both sides of \eqref{eq: wigner equation} to find
\begin{equation} \label{eq:longrange_lower bd}
    \begin{split}
-1&= \Im \left(W_{F}\left(  \mu = \frac{i}{N^{3}R_N} \right)  \right)   = \Im \left( \sum_{j, \pm} (\frac{i}{N^{3}R_N} - \phi_j^{\pm} )^{-1} \gamma   |w_j^{\pm}(1)|^2 \right) \\ 
&= \Im \left( \sum_{j, \pm} \gamma  \frac{-N^3 R_N (i+ \phi_j^{\pm}N^3R_N)}{1+(\phi_j^{\pm})^2N^6R_N^2}|w_j^{\pm}(1)|^2 \right)  = 
-\gamma \sum_{j, \pm} 
\frac{ N^3 R_N |w_j^{\pm}(1)|^2}{1+(\phi_j^{\pm})^2N^6R_N^2}.
\end{split}
\end{equation}
We shall restrict us now to positive signs and denote the smallest of the $|\phi_j^{+}|$ by $|\phi_{j_0}^+|$. We may then assume that $|\phi_{j_0}^+|> \varepsilon(\gamma) |w_{j_0}^+(1)|^2$. This is since otherwise: if $\phi_{j_0}^+ \in B(0,\varepsilon(\gamma) |w_{j_0}^+(1)|^2 )$ as the spectral gap is also assumed to decay as $o(1)N^{-3}$, would imply that $\mu \in B(z_{j_0}^+,\varepsilon(\gamma) |w_{j_0}^+(1)|^2 )$ which contradicts the fact that there is no eigenvalue $ \mu$ in such balls. 

As then $|\phi_{j_0}^+| \geq \varepsilon(\gamma) |w_{j_0}^+(1)|^2$, we estimate the terms in \eqref{eq:longrange_lower bd} as follows
\begin{align*}
    \frac{\gamma \vert w_{j_0}^{+}(1)\vert^2  N^{-3} R_N^{-1} }{N^{-6}R_N^{-2}+ (\phi_{j_0}^{+})^2} \lesssim \frac{\gamma R_N^{-1}}{\varepsilon(\gamma,\omega)^2} = o(1)\ \text{ and } \frac{\gamma}{N^3 R_N}\sum_{j \neq j_0} \frac{ \vert w_{j}^{+}(1)\vert^2}{N^{-6}R_N^{-2}+ (\phi_{j}^{+})^2} = o(1).
\end{align*}
Arguing as in the proof of Theo. \ref{theo: spectrum_magnetic field}, this leads to a contradiction.

\end{proof} 

%Then we write  $i \Omega_{[N]} =A + i \mathcal{K}$ where $A  = \begin{pmatrix} 0 & - i T(\omega)^{1/2} \\   i T(\omega)^{1/2} & 0  \end{pmatrix}$ and $$ \mathcal{K}= \mathcal{\hat{F}} + \omega \begin{pmatrix} 0 &  0\\  \operatorname{diag}(L,0,\cdots, 0, L) T(\omega)^{-1/2}&0  \end{pmatrix}. $$

 We expect our result, Theorem \ref{theo:spectrum for small omega}, on the spectral gap to hold for all $\omega \in [-1/4, 1/4[$. The restriction to smaller $\omega$'s is needed here in order to be able to characterise the spectrum of $A_{\omega}$ in terms of $A$. 
 
The technical problem, extending to the full subextremal range of $\omega$, when considering the perturbation as in \eqref{eq:matrix M rank1 perturb} is that the additional $P$ part creates an imaginary part in the coefficients of our Wigner equation making the argument in Step 5 in proof of Theo. \ref{theo:spectrum for small omega} to fail. By studying first $A_{\omega}$, we overcome this problem but we pay the price of reducing the range of $\omega$'s.  

In fact if one is only interested in an upper bound on the spectral gap, one can still obtain the $1/N^3$ scaling for all $\omega \in [-1/4, 1/4[$ by studying the perturbation problem $A + \begin{pmatrix}
    \Gamma &0 \\ P&0
\end{pmatrix}$ and following the same machinery as above.

\begin{appendix}

\section{Integral estimates}  
We recall the following very simple fact.

\begin{lemm}[Riemann sum]
\label{lemm:Riemann}
Let $f$ be a Riemann integrable strictly monotonically increasing function, then 
\[ \int_{(k_1-1)/N}^{1-k_2/N} f(t) \ dt \le \frac{1}{N} \sum_{i=k_1}^{N-k_2} f(i/N) \le \int_{k_1/N}^{1-(k_2-1)/N} f(t) \ dt. \]
\end{lemm}
\begin{proof}
By monotonicity 
\[ f(i/N) = N \int_{i/N}^{(i+1)/N} f(i/N) \ dt  \le N \int_{i/N}^{(i+1)/N} f(t) \ dt.\]
Hence, 
\[ \frac{1}{N} \sum_{i=k_1}^{N-k_2} f(i/N) \le \sum_{i=k_1}^{N-k_2} \int_{i/N}^{(i+1)/N} f(t) \ dt = \int_{k_1/N}^{1-(k_2-1)/N} f(t) \ dt\]
and analogously for the lower bound.
\end{proof}
\end{appendix}
%%%%%%%%%%%%%%%%%%%%%%%%%%%%%%%%%%%%%%%%%%%%%%%%%%%%%%%%%%%%%%%%%%%%%%%%%%%%%%%%
%                                 BIBLIOGRAPHY                                 %
%%%%%%%%%%%%%%%%%%%%%%%%%%%%%%%%%%%%%%%%%%%%%%%%%%%%%%%%%%%%%%%%%%%%%%%%%%%%%%%%
 \smallsection{Acknowledgements} The authors are grateful to Laure Saint-Raymond for bringing this problem to our attention and to Stefano Olla for useful references. AM acknowledges support from a fellowship at IHES and would like to thank the Max Planck Institute for Mathematics in the Sciences, Leipzig, for support and hospitality where part of this work was undertaken. 

\bibliographystyle{alpha}
\bibliography{bibliography}

\begin{thebibliography}{EPRB99b}

\bibitem[AE]{AE14}
A.~Arnold and J.~Erb.
\newblock Sharp entropy decay for hypocoercive and non-symmetric
  {F}okker-{P}lanck equations with linear drift.
\newblock arXiv:1409.5425.

\bibitem[BLRB00]{BLR00}
F.~Bonetto, J.~L. Lebowitz, and L.~Rey-Bellet.
\newblock Fourier's law: a challenge to theorists.
\newblock In {\em Mathematical physics 2000}, pages 128--150. Imp. Coll. Press,
  London, 2000.

\bibitem[BM22]{BM20}
S.~Becker and A.~Menegaki.
\newblock The optimal spectral gap for regular and disordered harmonic networks
  of oscillators.
\newblock {\em Journal of Functional Analysis}, 282(2):109286, 2022.

\bibitem[BNS78]{BNS78}
J.~R. Bunch, C.~P. Nielsen, and D.~C. Sorensen.
\newblock Rank-one modification of the symmetric eigenproblem.
\newblock {\em Numerische Mathematik}, 31(1):31--48, 1978.

\bibitem[Car07]{Car07}
P.~Carmona.
\newblock Existence and uniqueness of an invariant measure for a chain of
  oscillators in contact with two heat baths.
\newblock {\em Stochastic Process. Appl.}, 117(8):1076--1092, 2007.

\bibitem[Dha08]{Dhar08}
A.~Dhar.
\newblock Heat transport in low-dimensional systems.
\newblock {\em Advances in Physics}, 57, 08 2008.

\bibitem[EH00]{EH00}
J.-P. Eckmann and M.~Hairer.
\newblock Non-equilibrium statistical mechanics of strongly anharmonic chains
  of oscillators.
\newblock {\em Comm. Math. Phys.}, 212(1):105--164, 2000.

\bibitem[EH03]{EckmannHairer03}
J.~P. Eckmann and M.~Hairer.
\newblock Spectral properties of hypoelliptic operators.
\newblock {\em Communications in Mathematical Physics}, 235(2):233--253, 2003.

\bibitem[EPRB99a]{EPR99a}
J.-P. Eckmann, C.-A. Pillet, and L.~Rey-Bellet.
\newblock Entropy production in nonlinear, thermally driven {H}amiltonian
  systems.
\newblock {\em J. Statist. Phys.}, 95(1-2):305--331, 1999.

\bibitem[EPRB99b]{EPR99b}
J.-P. Eckmann, C.-A. Pillet, and L.~Rey-Bellet.
\newblock Non-equilibrium statistical mechanics of anharmonic chains coupled to
  two heat baths at different temperatures.
\newblock {\em Comm. Math. Phys.}, 201(3):657--697, 1999.

\bibitem[FB19]{BF19}
P.~Flandrin and C.~Bernardin, editors.
\newblock {\em Fourier and the {S}cience of {T}oday / {F}ourier et la {S}cience
  d'aujourd'hui}, volume 20, Issue 5.
\newblock Comptes Rendus Physique, 2019.

\bibitem[GCB21]{BCBD21b}
A.~Dhar G.~Cane, J. Majeed~Bhat and C.~Bernardin.
\newblock Localization effects due to a random magnetic field on heat transport
  in a harmonic chain.
\newblock arXiv:2107.06827, 2021.

\bibitem[Hoe67]{Ho69}
L.~Hoermander.
\newblock Hypoelliptic second order differential equations.
\newblock {\em Acta Math.}, 119:147--171, 1967.

\bibitem[JMBD21]{BCBD21a}
C.~Bernardin J.~M.~Bhat, G.~Cane and A.~Dhar.
\newblock Heat transport in an ordered harmonic chain in presence of a uniform
  magnetic field.
\newblock arXiv:2106.12069, 2021.

\bibitem[KO17]{KomOlla17}
T.~Komorowski and S.~Olla.
\newblock Diffusive propagation of energy in a non-acoustic chain.
\newblock {\em Archive for Rational Mechanics and Analysis}, 223(1):95--139,
  2017.

\bibitem[Lep16]{Lep16}
S.~Lepri.
\newblock {\em Thermal Transport in Low Dimensions: From Statistical Physics to
  Nanoscale Heat Transfer}, volume 921.
\newblock 01 2016.

\bibitem[Men20]{Me20}
A.~Menegaki.
\newblock Quantitative {R}ates of {C}onvergence to {N}on-equilibrium {S}teady
  {S}tate for a {W}eakly {A}nharmonic {C}hain of {O}scillators.
\newblock {\em J. Stat. Phys.}, 181(1):53--94, 2020.

\bibitem[Mon19]{Mon15}
P.~Monmarch\'{e}.
\newblock Generalized {$\Gamma$} calculus and application to interacting
  particles on a graph.
\newblock {\em Potential Anal.}, 50(3):439--466, 2019.

\bibitem[MPP02]{MPP02}
G.~Metafune, D.~Pallara, and E.~Priola.
\newblock Spectrum of {O}rnstein-{U}hlenbeck operators in {$L^p$} spaces with
  respect to invariant measures.
\newblock {\em J. Funct. Anal.}, 196(1):40--60, 2002.

\bibitem[RBT02]{RBT02}
L.~Rey-Bellet and L.~E. Thomas.
\newblock Exponential convergence to non-equilibrium stationary states in
  classical statistical mechanics.
\newblock {\em Comm. Math. Phys.}, 225(2):305--329, 2002.

\bibitem[RLL67]{RLL67}
Z.~Rieder, J.~L. Lebowitz, and E.~Lieb.
\newblock Properties of a harmonic crystal in a stationary nonequilibrium
  state.
\newblock {\em Journal of Mathematical Physics}, 8(5):1073--1078, 1967.

\bibitem[SS18]{SS18}
K.~Saito and M.~Sasada.
\newblock Thermal conductivity for coupled charged harmonic oscillators with
  noise in a magnetic field.
\newblock {\em Comm. Math. Phys.}, 361(3):951--995, 2018.

\bibitem[Sud22]{Suda22}
H.~Suda.
\newblock Superballistic and superdiffusive scaling limits of stochastic
  harmonic chains with long-range interactions.
\newblock {\em Nonlinearity}, 35(5):2288--2333, apr 2022.

\bibitem[TS18]{TS18}
S.~Tamaki and K.~Saito.
\newblock Nernst-like effect in a flexible chain.
\newblock {\em Physical Review E}, 2018.

\bibitem[TSS17]{TSS17}
S.~Tamaki, M.~Sasada, and K.~Saito.
\newblock Heat transport via low-dimensional systems with broken time-reversal
  symmetry.
\newblock {\em Phys. Rev. Lett.}, 119(11):110602, Sep 2017.

\bibitem[Vil09]{Villani09}
C.~Villani.
\newblock Hypocoercivity.
\newblock {\em Mem. Amer. Math. Soc.}, 202(950):iv+141, 2009.

\end{thebibliography}

\end{document}